\documentclass[11pt,a4paper,reqno]{amsart}

\makeatletter
\newcommand{\sectionnotoc}[1]{%
  \begingroup
  \let\addcontentsline\@gobblethree
  \section*{#1}%
  \endgroup
}
\makeatother

\usepackage[T1]{fontenc}
\usepackage[utf8]{inputenc}
\usepackage{lmodern}
\usepackage[english]{babel}
\usepackage{microtype}
\usepackage{amsmath,amssymb,amsthm,mathtools}
\usepackage[a4paper,margin=1in]{geometry}
\usepackage{bbm}
\usepackage{color}
\usepackage{xcolor}
\usepackage{mathrsfs}
\usepackage[most]{tcolorbox}

\usepackage[colorlinks=true]{hyperref} 
\usepackage{aliascnt}
\usepackage[capitalize,noabbrev]{cleveref}

\let\mathbb\mathbbm
\DeclareMathOperator{\sech}{sech}

\numberwithin{equation}{section}

\theoremstyle{plain}
\newtheorem{teo}{Theorem}[section]

\crefname{teo}{Theorem}{Theorems}
\Crefname{teo}{Theorem}{Theorems}
\newtheorem*{teo*}{Theorem}

\newaliascnt{cor}{teo}
\newtheorem{cor}[cor]{Corollary}
\aliascntresetthe{cor}

\crefname{cor}{Corollary}{Corollaries}
\Crefname{cor}{Corollary}{Corollaries}
\newtheorem*{cor*}{Corollary}

\newaliascnt{lemma}{teo}
\newtheorem{lemma}[lemma]{Lemma}
\aliascntresetthe{lemma}

\crefname{lemma}{Lemma}{Lemmas}
\Crefname{lemma}{Lemma}{Lemmas}

\newaliascnt{prop}{teo}
\newtheorem{prop}[prop]{Proposition}
\aliascntresetthe{prop}

\crefname{prop}{Proposition}{Propositions}
\Crefname{prop}{Proposition}{Propositions}
\newtheorem*{prop*}{Proposition}

\theoremstyle{definition}
\newaliascnt{definition}{teo}
\newtheorem{definition}[definition]{Definition}
\aliascntresetthe{definition}

\crefname{definition}{Definition}{Definitions}
\Crefname{definition}{Definition}{Definitions}

\newaliascnt{ex}{teo}
\newtheorem{ex}[ex]{Example}
\aliascntresetthe{ex}

\crefname{ex}{Example}{Examples}
\Crefname{ex}{Example}{Examples}

\theoremstyle{remark}
\newaliascnt{remark}{teo}
\newtheorem{remark}[remark]{Remark}
\aliascntresetthe{remark}

\crefname{remark}{Remark}{Remarks}
\Crefname{remark}{Remark}{Remarks}

\newcommand{\cI}{\mathcal I}
\newcommand{\cL}{\mathcal L}

\newcommand{\cC}{\mathcal C}
\newcommand{\cP}{\mathcal P}
\newcommand{\cH}{\mathcal H}
\newcommand{\ii}{\mathrm i}
\newcommand{\tr}{\operatorname{tr}}
\newcommand{\supp}{\operatorname{supp}}
\newcommand{\Skew}{\operatorname{Skew}}
\newcommand{\opnorm}[1]{\left\| #1\right\|}
\newcommand{\trnorm}[1]{\left\|#1\right\|_{\mathrm{tr}}}
\newcommand{\wnorm}[1]{\left\| #1\right\|_{W_1}}
\newcommand{\ind}{\mathbbm 1}
\newcommand{\Diag}{\operatorname{Diag}}
\newcommand{\cost}{\operatorname{cost}}
\newcommand{\proofstep}[1]{\par\medskip\noindent\emph{#1}}

\newtcolorbox{standbybox}{
  colback=gray!8,
  colframe=orange!60!black,
  title={In standby},
  fonttitle=\bfseries,
  coltitle=black,
  breakable
}

\title[Dobrushin]{Rapid mixing for Gibbs samplers for quantum spin systems with all-to-all interactions}
\author{Nicola Bertoni}
\address{Dipartimento di Matematica, Alma Mater Studiorum -- Universit\`a di Bologna, Piazza di Porta San Donato 5, 40126 Bologna, Italy}
\email{nicola.bertoni4@unibo.it}

\author{Giacomo De Palma}
\address{Dipartimento di Matematica, Alma Mater Studiorum -- Universit\`a di Bologna, Piazza di Porta San Donato 5, 40126 Bologna, Italy}
\email{giacomo.depalma@unibo.it}

\begin{document}

    \begin{abstract}
        Understanding when quantum many-body systems rapidly converge to thermal equilibrium is a central problem in quantum statistical mechanics. While spatial locality provides powerful tools for systems with an underlying geometry, fully connected models require estimates that accommodate interactions between every pair of spins. We study Gibbs-state preparation for two-local quantum Hamiltonians on arbitrary interaction graphs, including fully connected systems, via Lindblad dynamics that satisfies Kubo--Martin--Schwinger (KMS) detailed balance. Our approach is based on a quantum analogue of the classical Dobrushin condition, which controls the influence of local perturbations on the rest of the system. Using the quantum Wasserstein distance of order $1$, we quantify how these perturbations propagate under the dynamics and derive contraction estimates in a high-temperature regime controlled by the maximum total interaction strength per site. These estimates require no assumptions on spatial geometry or graph growth and yield mixing times logarithmic in the system size and the inverse error tolerance, uniformly over all initial states.
    \end{abstract}

    \maketitle

    \setcounter{tocdepth}{1}
    \tableofcontents

    \section{Introduction}

    Gibbs states describe the thermal equilibrium of quantum many-body systems.
    Preparing these states is a central task in quantum simulation: it gives
    access to finite-temperature observables and provides a way to investigate
    the equilibrium behaviour of interacting systems. It also raises a
    dynamical question. If a system evolves under a quantum Markov process
    designed to preserve its Gibbs state, how long does it take to approach
    equilibrium, and how does this time depend on the temperature, the number
    of spins, and their interactions?

    Stationarity alone does not answer this question. An evolution may
    preserve the desired state while converging to it too slowly to provide a
    useful preparation procedure. Quantitative mixing bounds therefore connect
    the study of relaxation in open quantum systems with the computational
    problem of Gibbs-state preparation. High temperature is a natural regime
    in which to seek such bounds: thermal fluctuations can suppress the
    influence that a perturbation at one site exerts on the rest of the system.
    Making this intuition precise is particularly challenging for
    noncommuting interactions, where even updates built from single-site
    operators can acquire support on many spins.

    In this paper, we study this problem for two-body quantum spin
    Hamiltonians on arbitrary interaction graphs, including fully connected
    systems. Our main result is a sufficient condition for a mixing time logarithmic in the number of spins; the condition is expressed
    in terms of the total interaction strength incident on a single spin, and requires neither a spatial embedding nor a restriction on
    the growth of neighbourhoods in the interaction graph. It therefore
    allows us to treat sparse and dense systems within the same framework,
    while retaining the strengths of the individual interactions.

    \subsection{Related work}

    In classical spin systems, a standard approach to Gibbs sampling is
    Glauber dynamics: a spin is updated according to its conditional
    distribution given the remaining configuration; see \cite{Martinelli99} for a systematic account. Dobrushin's uniqueness
    theory quantifies how changing one spin affects these conditional
    distributions through an influence matrix
    \cite{Dobrushin68,Dobrushin70}. When the combined influences are
    sufficiently small, local disturbances cannot sustain one another.
    This principle underlies criteria for uniqueness of infinite-volume
    Gibbs measures and, through coupling arguments, rapid convergence of
    finite-volume sampling dynamics. Path coupling provides a general way
    to extend contraction between neighbouring configurations to arbitrary
    pairs \cite{BubleyDyer97}. Further developments clarify the roles of
    influence bounds, graph structure, and matrix norms in uniqueness and
    mixing \cite{Weitz05,DGJ09}.

    The quantum setting requires a replacement for this description in
    terms of conditional probabilities and coupled spin configurations.
    Early work developed quantum stochastic dynamics and ergodicity
    criteria for lattice systems \cite{MZ95}, as well as quantum
    oscillation norms measuring the sensitivity of observables to local
    changes \cite{RW96}. Related oscillation-norm estimates were applied to high-temperature stabilizer dynamics by Temme and Kastoryano \cite{TK15}. Crawford, De Roeck, and Sch\"utz proved stability of exponential relaxation for weakly coupled quantum Markov dynamics on lattices \cite{CDS15}, providing a perturbative antecedent to high-temperature mixing arguments. A transportation metric particularly suited to
    the present problem is the quantum Wasserstein distance of order~$1$
    introduced by De Palma, Marvian, Trevisan, and Lloyd \cite{DTML21}.
    It resolves differences between states into single-site contributions
    and recovers the classical Wasserstein distance associated with Hamming
    distance for states diagonal in the computational basis. This makes it
    possible to track where a perturbation is located and how it propagates.
    De Palma and Rouz\'e developed quantum Dobrushin uniqueness conditions and proved high-temperature contraction in this Wasserstein distance for commuting Gibbs heat-bath dynamics on arbitrary hypergraphs \cite[Sections~3--4]{DPR22}. Gao and Rouz\'e developed a general coarse Ricci curvature framework for quantum transport contraction \cite{GR24}.

    For commuting Hamiltonians, substantial progress has been made by
    relating equilibrium correlations to the convergence of Gibbs samplers.
    Kastoryano and Brand\~ao established a connection between strong
    clustering and a uniform spectral gap \cite{KB16}. Lucia, P\'erez-Garc\'ia, and P\'erez-Hern\'andez developed a related approach through a canonical purified Hamiltonian whose gap is controlled by spatial mixing of the state \cite{LPP25}. A gap bound alone,
    however, does not generally imply a mixing time logarithmic in the
    system size. Quantum logarithmic Sobolev inequalities provide a stronger route to mixing estimates \cite{KT13}. For one-dimensional commuting heat-bath dynamics, Bardet et al. obtained a size-independent modified logarithmic Sobolev bound under a mixing condition and strong quasi-factorisation of relative entropy \cite{BCLPR21}. Capel, Rouz\'e, and Stilck Fran\c{c}a established size-independent modified logarithmic Sobolev inequalities for commuting nearest-neighbour Schmidt dynamics under spatial-mixing assumptions \cite{CRF21}. For finite-range, translation-invariant commuting spin chains, Bardet et al. proved polylogarithmic mixing of Davies dynamics at every fixed positive temperature \cite{BCGLPR24}. Uniform entropy-decay results were obtained by Kochanowski,
    Alhambra, Capel, and Rouz\'e for classes of commuting models, using
    functional inequalities and suitable clustering properties
    \cite{KACR25}. Capel, Gondolf, Kochanowski, and Rouz\'e also developed
    sampling bounds in normalized noncommutative transport metrics under
    correlation-decay assumptions \cite{CGKR25}.

    For general noncommuting Hamiltonians, Chen, Kastoryano, and Gily\'en
    constructed an efficiently implementable quantum Gibbs sampler
    satisfying exact detailed balance \cite{CKG23}. Related thermal-simulation algorithms are presented by Chen, Kastoryano, Brand\~ao, and Gily\'en \cite{CKBG25}. Further exact constructions include the finite-jump KMS samplers of Ding, Li, and Lin \cite{DLL25} and the quantum Glauber and Metropolis constructions of Gily\'en, Chen, Doriguello, and Kastoryano \cite{GCDK26}. Establishing logarithmic convergence of a specified construction requires a separate mixing argument. Rouz\'e,
    Stilck Fran\c{c}a, and Alhambra proved polynomial-time high-temperature thermalization
    \cite{RFA25} and logarithmic mixing for
    Hamiltonians satisfying suitable Lieb--Robinson locality estimates,
    including lattice models and classes of long-range interactions
    \cite{RFA26}. These arguments use control of operator propagation,
    whose general study goes back to Lieb and Robinson \cite{LR72}.
    Graph-based bounds offer another way to retain the structure and
    weights of the interactions in this analysis \cite{CL21}.

    The approach most directly related to ours is due to Bakshi, Liu, Moitra, and Tang \cite{BLMT25}. They formulate a quantum Dobrushin condition using the Wasserstein distance of \cite{DTML21} and show
    how bounds on the propagation of local perturbations imply rapid mixing. Their framework uses transport plans and nonnegative update
    matrices to keep track of the contributions associated with individual sites. They apply it to a detailed-balanced Gibbs sampler, obtaining
    high-temperature mixing conditions involving interaction locality, degree, and graph-growth parameters. More precisely, for interaction terms normalized to have operator norm at most $1$, their continuous-time rapid-mixing theorem assumes the inverse temperature $\beta$ to satisfy
    \[
        \beta < \frac{1}{10000\,\mathfrak K^3\mathfrak b^2\mathfrak d},
    \]
    where $\mathfrak K$ is the maximum support size of an interaction term, $\mathfrak d$ is the maximum degree of the dual interaction graph (and therefore counts other terms), and $\mathfrak b$ controls the exponential growth of neighbourhoods in the site interaction graph \cite[Definitions~3.22--3.23 and Theorem~2.1]{BLMT25}. Here and below, numbered references to this work refer to its full-length arXiv version. Their work also establishes
    decay of conditional mutual information under appropriate locality assumptions. We use their sampler and the Dobrushin-to-mixing framework, and derive a criterion adapted to arbitrary two-body
    interaction graphs.

    A closely related result beyond geometric locality is Bergamaschi's analysis of all-to-all quantum systems \cite{Bergamaschi26}. For general few-body qudit Hamiltonians with controlled interaction strengths, that work proves a system-size independent spectral gap for the sampler of \cite{CKG23} at sufficiently high temperature. The resulting trace-norm mixing bound scales as $\mathcal O(n+\log(1/\varepsilon))$. Our result concerns the different sampler of \cite{BLMT25} and two-body qubit interactions, and establishes $\mathcal O(\log(n/\varepsilon))$ mixing through direct transportation contraction. The distinction is therefore the logarithmic mixing guarantee for this class and sampler, while \cite{Bergamaschi26} covers a broader class of interactions.
    Our proof shares with \cite[Sections~3--4]{Bergamaschi26} the use of connected imaginary-time expansions, tree summation of interaction weights, marked-site estimates, and a comparison with infinite-temperature contraction. We specify these connections at their points of use below; the sampler, transport-plan framework, and Dobrushin-to-mixing criterion are taken from \cite{BLMT25}.

    Other results address complementary regimes. \v{S}m\'id, Meister, Berta, and Bondesan establish rapid mixing for noninteracting reference systems and stability results covering weakly interacting qudits and perturbed non-hopping fermions \cite{SMBB26}. Their perturbative criteria use spatial-decay and volume-growth control, together with oscillator-norm estimates for filtered KMS samplers. Bakshi, Moitra, and Tan establish logarithmic mixing for a field-resonant Gibbs sampler with arbitrary on-site fields \cite{BT26}. Their temperature condition depends on the locality and degree of the intersite interactions and is independent of the field strength. Our criterion instead retains the individual interaction weights; it does not provide an analogous guarantee uniform in arbitrary on-site fields.

    At every fixed finite temperature, Bergamaschi and Chen proved a system-size independent gap for the CKG sampler on one-dimensional finite-range systems, together with a separate preparation algorithm of polylogarithmic circuit depth \cite{BC26}. Kim established logarithmic mixing for one-dimensional finite-range systems using block Petz recovery updates and the quantum Wasserstein distance of order~$1$ \cite{Kim26}. Rouz\'e and Stilck Fran\c{c}a \cite{RSF26} developed quantum Dobrushin--Shlosman conditions for smoothed heat-bath block dynamics constructed by quantum belief propagation, obtaining logarithmic mixing for finite-range chains and stability of a strict contraction condition under small local perturbations on graphs of polynomial volume growth. Capel, P\'erez-Garc\'ia, and Scandi proved an inverse-logarithmic ordinary modified logarithmic Sobolev bound for Petz heat-bath and regularised heat-bath dynamics on finite-range, possibly inhomogeneous open chains \cite{CPS26}.

    For Gaussian CKG dynamics, Chen and Ding obtained a one-dimensional trace-norm mixing bound of $\mathcal O(\log^4 n+\log(1/\varepsilon))$ through spectral estimates beyond the gap \cite{CD26}. Size-independent functional inequalities yielding logarithmic mixing were established by Gao, Guo, and Wang for one-dimensional Gaussian and shifted-Metropolis CKG dynamics \cite{GGW26}, and by Chen, Ding, and Lu for Gaussian CKG dynamics on nearest-neighbour chains \cite{CDL26}. The latter work also covers finite-range systems at sufficiently high temperature in each fixed spatial dimension. These results complement our high-temperature criterion, with geometric hypotheses and sampler constructions different from the arbitrary-graph setting considered here.

    Ramkumar and Soleimanifar study mixing for random sparse matrix Hamiltonians \cite{RS24}; sparsity there concerns the matrix ensemble rather than the degree of a spin interaction graph. In a complementary all-to-all fermionic setting, Cai, Zhan, and Zlokapa proved that the disordered Sachdev--Ye--Kitaev model admits a Gibbs sampler with a system-size independent spectral gap at sufficiently high constant temperature, with high probability over the disorder, yielding polynomial-time preparation \cite{CZZ26}. Conversely, bottleneck arguments show that quantum Gibbs samplers can mix slowly for several families of low-temperature models \cite{GKZ24}, underscoring the need for assumptions on the temperature or interactions. High-temperature expansions also have
    an algorithmic role beyond mixing: the quantum cluster expansion of Neto\v{c}n\'y and Redig \cite{NR04} underlies, for example, classical
    approximation algorithms for quantum partition functions developed by Mann and Helmuth \cite{MannHelmuth21}. A different high-temperature preparation approach uses separability and efficient classical sampling of product-state decompositions \cite{BLMT24}; this does not constitute a mixing analysis of the Lindbladian studied here.

    \subsection{Our results}

    We consider the two-body model of \cref{Hamiltonian}: $n$ qubits with a Hamiltonian
    \begin{equation}
    \label{Hamiltonian introduction}
        H=\sum_{1\leq i<j\leq n}H_{i,j},
    \end{equation}
    where each $H_{i,j}$ is Hermitian and supported on $\{i,j\}$; zero
    terms are allowed. The interactions need not commute. The target
    state is the Gibbs state

    \[
    \sigma=\frac{e^{-\beta H}}{\tr(e^{-\beta H})}
    \]
     at inverse
    temperature $\beta>0$. The parameter controlling our result is
    \[
        \Lambda=\max_{i\in[n]}\sum_{j\neq i}\|H_{i,j}\|,
    \]
    where $H_{j,i}=H_{i,j}$ for $i<j$. Thus, $\Lambda$ is the largest
    total interaction strength incident on any qubit. Unlike the
    unweighted degree of the interaction graph, it distinguishes strong
    interactions from weak ones. In particular, a site may have many
    interaction partners while its total interaction strength remains
    controlled.

    This parameter is the fixed-decomposition analogue of the local norm introduced by De Palma and Pastorello \cite[Definition~3.1]{DP25}. Their definition minimizes over decompositions of the observable and includes a factor of two. In particular, our prescribed decomposition gives
    \[
        \|H\|_{\mathrm{loc}}\leq 2\Lambda.
    \]
    We keep the specified two-body decomposition throughout this paper.

    We consider the Lindbladian $\cL$ specified in \cref{Lindbladian}, introduced in \cite[Definition~4.1, Section~4]{BLMT25}. It satisfies the KMS detailed balance condition and has the Gibbs state $\sigma$ as a stationary state. The associated evolution is described by the quantum Markov semigroup $\{e^{t\cL}\}_{t\geq0}$.

    Our main result shows that, at sufficiently high temperature relative
    to the interaction strength $\Lambda$, this evolution converges
    rapidly to $\sigma$, uniformly over the initial state. More precisely,
    the mixing time grows only logarithmically with the number of qubits
    and the inverse target accuracy. The following statement summarizes
    \cref{satisfaction of Dobrushin condition}, where the complete result and proof are given.

    \begin{teo*}[Rapid mixing for high temperature Gibbs states]
        Let $H$ be a Hamiltonian as in \eqref{Hamiltonian introduction} and let $\sigma$ be its Gibbs state at inverse temperature $\beta$ satisfying
        \[
        \beta \Lambda < \frac{1}{100}.
        \]
        Then, for every $0< \varepsilon < 1$, and for every initial state $\rho_0$,
        \[
        \left\|e^{t\cL }(\rho_0)-\sigma\right\|_{\mathrm{tr}}
        \leq\varepsilon
        \qquad\mathrm{whenever}\qquad
        t\geq\frac13\log\frac{2n}{\varepsilon}.
        \]
    \end{teo*}

    All mixing times in this paper use the extensive generator
    \[
        \cL=\sum_{i=1}^n\cL_i,
    \]
    with unit coefficients for the three Pauli-labelled terms at each site, as in \cite[Theorem~2.1]{BLMT25}.

    Here, the relevant high-temperature scale is set by the local
    interaction strength $\Lambda$, rather than the norm of the full
    Hamiltonian. No bounded-degree or graph-growth hypothesis is needed.

    The proof analyses how the dynamics changes the quantum Wasserstein distance between two states. Using transport plans, we track the contributions associated with individual sites and estimate how they are redistributed by the interactions. We study the dissipative and coherent parts of the generator separately in \cref{dissipative evolution,coherent evolution}, obtaining bounds in terms of the total interaction strength $\Lambda$. Combining these estimates in \cref{lem:high-temperature-dobrushin}, we verify the quantum Dobrushin condition at sufficiently high temperature. This establishes contraction of the continuous-time evolution and hence rapid convergence to the Gibbs state.

    Obtaining these bounds without assumptions on the growth of the interaction graph is the main technical challenge. We control the imaginary-time expansions through weighted sums over connected sequences of interactions, as quantified in \cref{lem:weighted-clusters}. The coherent part requires an additional step, since its Fourier representation involves an average over real-time evolutions. To handle this average, we construct the reversible evolution of transport plans in \cref{lem:reversible-plan-lift} and bound its sitewise costs through the nonnegative comparison matrices of \cref{lem:two-mark-input}. Keeping track of the individual
    interaction strengths in both steps allows us to express the bounds in terms of $\Lambda$.

    Furthermore, we consider an application of the theorem to the case of an all-to-all interaction graph. Specifically, for $n\geq2$ and $J>0$, we consider
    \begin{equation} \label{Alltoall Hamiltonian}
        H=\frac{1}{\sqrt {n-1}}\sum_{i<j}h_{i,j},
        \qquad \|h_{i,j}\|\leq J.
    \end{equation}
    The following is the complete-graph specialization of \cref{cor:degree-normalized-mixing}.
    \begin{cor*}
        Let $H$ be a Hamiltonian as in \eqref{Alltoall Hamiltonian}, and let $\sigma$ be its Gibbs state at inverse temperature $\beta$. Then, if
        \[
        \beta < \frac{1}{100 J \sqrt{n-1}},
        \]
        for every initial state \(\rho_0\) and every \(0<\varepsilon<1\), the evolution $\{e^{t\cL}\}_{t \geq 0}$ satisfies
        \[ \| e^{t \cL}(\rho_0) -\sigma \|_{\mathrm{tr}} \leq \varepsilon,\]
        for a time
        \[
        t = \mathcal O(\log(n/\varepsilon)).
        \]

    \end{cor*}

    To compare this condition with that of \cite{BLMT25}, suppose that $n\geq3$, $J>0$, and each $h_{i,j}$ has support exactly $\{i,j\}$. For this complete interaction graph, the dual-graph convention of \cite[Definition~3.22]{BLMT25} gives
    \[
        \mathfrak K=2,\qquad \mathfrak b=n,\qquad \mathfrak d=2n-4.
    \]
    Apply their theorem to the rescaled quantities
    \[
        \widehat H:=\frac1J\sum_{i<j}h_{i,j},\qquad
        \widehat\beta:=\frac{\beta J}{\sqrt{n-1}},\qquad
        \widehat\beta\widehat H=\beta H.
    \]
    This gives the sufficient condition
    \[
        \beta < \frac{\sqrt{n-1}}{80000\,J\,n^2(2n-4)}.
    \]
    Thus, for the same Hamiltonian normalization, the sufficient inverse-temperature threshold in \cite{BLMT25} scales as $n^{-5/2}/J$, whereas the threshold in the above corollary scales as $n^{-1/2}/J$. The improvement in the guaranteed  temperature regime comes from controlling weighted interaction strengths directly, without the additional graph-growth factor.

    This result controls the interaction strength per site and does not require $\beta\|H\|$ to be small. In particular, the allowed regime includes Hamiltonians for which this global quantity grows linearly with $n$.

    The conclusion is a bound on the time required by the specified continuous-time Gibbs sampler. Translating it into a quantum circuit complexity also requires controlling the cost of implementing that evolution.

    \section{Notations and preliminaries}

    \subsection*{Notations.}
    We use $\mathcal O(\cdot)$ for asymptotic upper bounds.

    We work in the Hilbert space corresponding to a system of $n$ qubits,
    \[
    \cH_n:=(\mathbb C^2)^{\otimes n},
    \]
    and denote its dimension by $N:=2^n$. We use $I$ for the identity matrix and, on the auxiliary tuple spaces introduced below, for the identity operator on that space. We reserve $\cI$ for the identity map on physical operators. The symbol $0$ denotes the zero scalar, vector, matrix, or map according to its domain.

    For a subset of indices $S$, we let $e_S$ be the indicator vector of the set: it is $1$ when the index $i\in S$ and $0$ otherwise. When $|S|=1$, we sometimes abuse notation and write $e_i$ instead of
    $e_{\{i\}}$. We let the matrix $E_S$ be defined entrywise by
    \[
        (E_S)_{i,j} =
        \begin{cases}
        1, \qquad \text{if } i\in S\text{ and }j\in S \\
        0, \qquad \text{otherwise}.
        \end{cases}
    \]

    We write $\mathbf{1}:=e_{[n]}$ for the all-ones vector and $\ind_{\mathcal A}$ for the scalar indicator of a condition $\mathcal A$.

    For a vector $v$, we use $\|v\|$ for its Euclidean norm, $\|v\|_1$ for its $\ell_1$ norm, and $\|v\|_\infty$ for its maximum absolute entry. For vectors with real entries, the relations $u\leq_v v$ and $u<_v v$ mean, respectively,
    \[
        u_i\leq v_i\quad\text{for every }i,
        \qquad
        u_i<v_i\quad\text{for every }i.
    \]
    Comparisons between numerical matrices are entrywise. For a vector $p$, the notation $\Diag(p)$ denotes the diagonal matrix with diagonal $p$.

    For an operator $X$ on $\cH_n$, we denote its conjugate transpose by $X^\dagger$, its operator norm by $\|X\|$ and its trace norm by $\|X\|_{\mathrm{tr}}$.
    For a numerical matrix $A\in\mathbb C^{n\times n}$, we write
    \[
        \|A\|_{1\to1} = \max_{i\in[n]}\|Ae_i\|_1
    \]
    for the norm induced by the vector $\ell_1$ norm. The corresponding maximum-row-sum norm is
    \[
        \|A\|_{\infty\to\infty}:=\max_{i\in[n]}\sum_{j=1}^n|A_{i,j}|.
    \]
    We also define
    \begin{equation*}
         \Skew(X):=\frac{X-X^\dagger}{2\ii}.
    \end{equation*}
    Moreover, for a linear map $\mathcal K$ on operators, we denote by $\mathcal K^\dagger$ its adjoint with respect to the Hilbert--Schmidt inner product, characterized by
    \[
        \tr \!\left(A^\dagger\mathcal K(B)\right) = \tr\!\left(
            (\mathcal K^\dagger(A))^\dagger B
        \right)
    \]
    for all operators $A,B$. We use $\|\mathcal K\|_{1\to1}$ to denote the norm induced by the trace norm:
    \[
        \|\mathcal K\|_{1\to1} := \sup_{X\neq0} \frac{\|\mathcal K(X)\|_{\mathrm{tr}}} {\|X\|_{\mathrm{tr}}}.
    \]

    We will also work with channels, i.e., trace-preserving, completely positive linear maps from density matrices to density matrices; see \cite{Wat18} for standard background. We will sometimes relax these conditions and work with more general classes of functions on operators. We will call such a function a \emph{linear map} if it is linear and takes Hermitian matrices to Hermitian matrices. We use $\cI $ to denote the identity map.

    A useful basis for us to work in is the basis of tensor products of
    Pauli matrices. In particular, for each $i\in[n]$, we denote by
    \[
        \cP_i :=\{\sigma_X^{(i)},\sigma_Y^{(i)},\sigma_Z^{(i)}\}
    \]
    the set of nontrivial single-site Pauli operators at site $i$, where
    \[
        \sigma_\alpha^{(i)} := I^{\otimes(i-1)}\otimes\sigma_\alpha \otimes I^{\otimes(n-i)}, \qquad \alpha\in\{X,Y,Z\}.
    \]
    \subsection{Lindbladian evolution}

    We introduce the general notion of Lindbladian evolution, which is the quantum analogue of continuous-time Markov chain evolution. The structure of its generator is due to Gorini, Kossakowski, and Sudarshan \cite{GKSL76} and Lindblad \cite{Lindblad76}.
    \begin{definition}[General Lindbladian]
    Let \(G \in \mathbb{C}^{2^n \times 2^n}\) be Hermitian, and let a family \(\{K_k\}_k \subset \mathbb{C}^{2^n \times 2^n}\). The Lindbladian \(\cL \) with coherent term \(G\) and jump operators \(K_k\) is the linear map acting on density matrices \(\rho\) by
    \[
    \cL(\rho) = -\ii [G,\rho] + \sum_k \left( K_k \rho K_k^\dagger - \frac{1}{2}\{K_k^\dagger K_k,\rho\} \right),
    \]
    where
    \[
    [A,B] := AB-BA \qquad\text{and}\qquad \{A,B\} := AB+BA
    \]
    denote the commutator and anticommutator, respectively.
    \end{definition}
    The Lindbladian \(\cL \) induces a one-parameter family of quantum channels \(\{e^{t\cL }\}_{t\geq 0}\), defined by
    \begin{equation} \label{semigroup}
    e^{t\cL } = \lim_{m\to\infty}
    \left(\cI +h_m\cL \right)^m,\qquad h_m:=\frac{t}{m},\quad t\geq0.
    \end{equation}
    Equivalently, \(e^{t\cL }\) is the solution map of the master equation
    \[
    \frac{d}{dt}\rho(t) = \cL(\rho(t)), \qquad \rho(0)=\rho_0.
    \]
    The family \(\{e^{t\cL }\}_{t\geq 0}\) forms a dynamical semigroup, meaning
    that for all \(s,t\geq 0\),
    \[
    e^{s\cL } e^{t\cL } =
    e^{(s+t)\cL }, \qquad e^{0\cL }=\cI .
    \]
    Conversely, under suitable boundedness conditions, every time-homogeneous quantum dynamical semigroup has a generator of the above Lindblad form. Hence, in a formal sense, all Markovian, time-independent open-system dynamics may be described by Lindbladian evolution. As in classical Markovian dynamics, quantum thermalization is not always exactly modeled by such a semigroup. A derivation from system--bath dynamics therefore requires additional physical assumptions, as in the weak-coupling limit studied by Davies \cite{Davies74}. The following lemma controls the possible increase in trace norm under a dissipative update. We will use this estimate to bound the costs of the corresponding transport plans.

    \begin{lemma}[Trace-norm bound for a dissipative update] \label{size of lind step}
        Let $A$ be an operator and, for $s \geq 0$, define
        \[
            \Phi_s(X)
            := X+s\left(AXA^\dagger
            -\frac12\{A^\dagger A,X\}\right).
        \]
        Then, for every Hermitian operator $X$,
        \[
            \|\Phi_s(X)\|_{\mathrm{tr}}
            \leq
            \left(1+\frac{s^2}{2}\|A^\dagger A\|^2\right)
            \|X\|_{\mathrm{tr}}.
        \]
    \end{lemma}

    \begin{proof}
        Set $V:=A^\dagger A$ and introduce the completely positive map
        \[
            \mathcal T_s(X)
            :=
            \left(I-\frac{s}{2}V\right)
            X
            \left(I-\frac{s}{2}V\right)
            +sAXA^\dagger.
        \]
        Expanding the product gives
        \[
        \Phi_s(X)=\mathcal T_s(X)-\frac{s^2}{4}VXV.
        \]
         Write $X=X_+-X_-$, where $X_+,X_-\geq0$ and
        $X_++X_-=|X|$. Since $\mathcal T_s$ is positive,
        \begin{align*}
            \|\mathcal T_s(X)\|_{\mathrm{tr}} \leq \|\mathcal T_s(X_+)\|_{\mathrm{tr}} + \|\mathcal T_s(X_-)\|_{\mathrm{tr}} = \tr\mathcal T_s(|X|).
        \end{align*}
        By cyclicity of the trace,
        \begin{align*}
            \tr\mathcal T_s(|X|) = \tr\left[
                \left(I+\frac{s^2}{4}V^2\right)|X| \right] \leq \left(1+\frac{s^2}{4}\|V\|^2\right) \|X\|_{\mathrm{tr}}.
        \end{align*}
        Finally, we obtain
        \begin{align*}
            \|\Phi_s(X)\|_{\mathrm{tr}} \leq \|\mathcal T_s(X)\|_{\mathrm{tr}} +\frac{s^2}{4}\|VXV\|_{\mathrm{tr}} \leq \left(1+\frac{s^2}{2}\|V\|^2\right) \|X\|_{\mathrm{tr}}.
        \end{align*}
    \end{proof}

    \subsection{Quantum Wasserstein distance}
    We recall the definition of the quantum Wasserstein distance of order \(1\), introduced by De Palma, Marvian, Trevisan, and Lloyd \cite{DTML21}, as a distance metric on quantum states. We describe the main properties of this metric, and refer the reader to the original work for proofs.

    \begin{definition}[Quantum Wasserstein distance {\cite[Definition~6]{DTML21}}]
    Let \(X\) be a traceless Hermitian operator. The quantum Wasserstein norm of
    order \(1\) of \(X\) is defined by
    \[
    \|X\|_{W_1} := \min \left\{ \frac{1}{2}\sum_{i=1}^n \|X_i\|_{\mathrm{tr}} \;:\; \tr_i(X_i)=0,\; X=\sum_{i=1}^n X_i,\; X_i=X_i^\dagger \right\}.
    \]
    \end{definition}
    The above minimum is achieved by some collection $(X_i)_{i\in[n]}$. For \(X=\rho-\sigma\), where \(\rho\) and \(\sigma\) are quantum states, this quantity defines a transportation metric: it can be interpreted as the minimum transportation cost required to transform \(\rho\) into \(\sigma\) through single-site contributions satisfying \(\tr_i(X_i)=0\). To control these transportation costs, we introduce the notion of a transport plan and its site-specific cost vector.

    \begin{definition}[Transport plan and cost vector] \label{definition of cost vector}
    Let \(X\) be a traceless Hermitian operator. A transport plan for \(X\) is a collection of operators $(X_i)_{i\in[n]}$ satisfying the following conditions:
    \begin{enumerate}
        \item Each \(X_i\) is Hermitian and traceless;
        \item

        \[
        X=\sum_{i=1}^n X_i;
        \]

        \item For every \(i\in[n]\),
        $\tr_i(X_i)=0$.
    \end{enumerate}

    For a transport plan \((X_i)_{i\in[n]}\), its cost vector is the vector
    \[
    \cost((X_i)_i)
    :=
    \left(
    \frac{1}{2}\|X_i\|_{\mathrm{tr}}
    \right)_{i\in[n]} =
    \left(
    \|X_i\|_{W_1}
    \right)_{i\in[n]},
    \]
    where the last equality follows from \cref{Lemma W1 rel tr} below.
    \end{definition}

    With these definitions, the quantum Wasserstein norm can be written as the minimum \(\ell_1\)-cost over all transport plans for \(X\):
    \[
    \|X\|_{W_1}= \min_{(X_i)_i \ \text{transport plan for } X}
    \|{\cost((X_i)_i)}\|_1.
    \]

    Because transport plans are linear objects, they combine in the expected way.

    \begin{prop}[Linearity of transport plans] \label{prop on lin}
        Let $X$ and $Y$ be traceless Hermitian operators with transport plans $(X_i)_{i\in[n]}$ and $(Y_i)_{i\in[n]}$ and corresponding cost vectors $x$ and $y$. For every $\alpha,\beta\in\mathbb R$, the collection $(\alpha X_i+\beta Y_i)_{i\in[n]}$ is a transport plan for $\alpha X+\beta Y$, with cost vector $z$ satisfying \(  z\leq_v|\alpha|x+|\beta|y \).
        \end{prop}

     Quantum Wasserstein distance can be related to trace distance, analogously to what happens in the classical case.

    \begin{lemma}[Relating Wasserstein to trace {\cite[Proposition 2]{DTML21}}]   \label{Lemma W1 rel tr}
        For a traceless Hermitian operator $X$,
        \[
        \frac{1}{2}\|X\|_{\mathrm{tr}} \leq \|X\|_{W_1} \leq \frac{n}{2}\|X\|_{\mathrm{tr}}.
        \]
        Moreover, if $\tr_i(X)=0$ for some $i\in[n]$, then there is a transport plan with cost vector $\frac{1}{2}\|X\|_{\mathrm{tr}}e_i$. In particular, this implies that
        \[
        \frac{1}{2}\|X\|_{\mathrm{tr}}=\|X\|_{W_1}.
        \]
    \end{lemma}

    When \(X\) vanishes after tracing out a small subsystem \(S\), the upper bound can be sharpened. More precisely, if \(X=\rho-\sigma\) and $\tr_S(X)=0$, then \(\rho\) and \(\sigma\) have the same reduced state on the complement \(S^c\). As the following lemma shows, $X$ admits a transport plan whose nonzero components are indexed only by sites in $S$, giving a bound that depends on $|S|$ rather than $n$.

    \begin{lemma}[Relating Wasserstein to trace for few-site deviations {\cite[proof of Proposition~5]{DTML21}}] \label{Lemma W1 with fewsite deviations}
    Let \(X\) be a traceless Hermitian operator, and let \(S\subseteq[n]\), $\tr_S(X)=0$.
    Then there exists a transport plan \((X_i)_{i\in[n]}\) for \(X\) whose cost vector \(x=\cost((X_i)_i)\) satisfies
    \[
    x \leq_v \,\|X\|_{\mathrm{tr}} e_S.
    \]
    In particular,
    \[
    \|X\|_{W_1}\leq|S|\,\|X\|_{\mathrm{tr}}.
    \]

    \end{lemma}

    We will eventually be trying to show that certain maps act as contractions in Wasserstein norm, namely
    \[
    \wnorm {\Phi (X)} \leq q\wnorm X,
    \]
     for $0\leq q < 1$ and for every nonzero traceless Hermitian \(X\). Thus, we need to track how to take a transport plan of $X$ and construct a transport plan for $\Phi(X)$. It is useful to introduce the notion of update matrix. We use a slight extension of the update-matrix definition in \cite[Definition~3.16]{BLMT25}, allowing Hermiticity-preserving linear maps that preserve tracelessness.

    \begin{definition}[Update matrix] \label{Update matrix}
    Let $\Phi$ be a Hermiticity-preserving linear map that maps traceless operators to traceless operators. A matrix $Q\in\mathbb R_{\geq0}^{n\times n}$ is an update matrix for $\Phi$ if, for any traceless Hermitian matrix $X$ and any transport plan $(X_i)_i$ of $X$, the matrix $\Phi(X)$ has a transport plan with cost vector
    \[
        y \leq_v Qx,
    \]
    where $x$ is the cost vector of $(X_i)_i$.
    \end{definition}

    The behavior of cost vectors under evolution need not be linear: we merely bound it in a linear way. A linear bound is natural because transport plans can be combined linearly, and their cost vectors satisfy the corresponding bound by the triangle inequality (\cref{prop on lin}).

    \begin{lemma}[Wasserstein norm of a few-site map {\cite[Lemma 4.16]{BLMT25}}] \label{Lemma few site}
        Given two matrices $K$ and $L$, consider a map $\Psi$ of the form
        \[
        \Psi(\rho) = K\rho L^\dagger + L\rho K^\dagger - \frac{1}{2}\{L^\dagger K + K^\dagger L,\rho\}.
        \]
        Further suppose that \(  S = \supp(K)\cup \supp(L) \). Then this map is trace-annihilating and, if $X$ is a traceless Hermitian matrix, then $\Psi(X)$ has a transport plan $\{Y_k\}_k$ whose cost vector is entry-wise bounded by
        \[
        4\,\|K\| \|L\| \|X\|_{\mathrm{tr}}\, e_S.
        \]
    \end{lemma}

    \section{The model and the Gibbs sampler}
    \subsection{Properties of the Hamiltonian}
    We now describe the class of Hamiltonians considered in this paper. We work with $2$-local Hamiltonians on arbitrary interaction graphs, including fully connected systems as a special case. We impose no spatial locality assumptions. Our high-temperature condition depends on the maximum total interaction strength per site, quantified by $\Lambda$ below, without the additional graph-growth parameters appearing in \cite{BLMT25}.

    \begin{definition}[$2$-local Hamiltonian] \label{Hamiltonian}
        We define a Hamiltonian on $n$ qubits as an operator $H \in \mathbb{C}^{2^n \times 2^n}$ such that
        \[
        H = \sum_{1 \leq i < j \leq n} H_{i,j},
        \]
        where each $H_{i,j} \in \mathbb{C}^{2^n \times 2^n}$ is Hermitian and supported on $\{i,j\}$. Zero interaction terms are allowed. Moreover, for each pair of indices $i<j$, we define
        \[
        w_{i,j} := \| H_{i,j} \| , \qquad w_{j,i} := w_{i,j}, \qquad w_{i,i} := 0,\]
        and we set
        \begin{equation} \label{Lambda}
            \Lambda := \max_{i \in [n]} \sum_{j=1}^n w_{i,j}.
        \end{equation}

    \end{definition}

    The quantity in \eqref{Lambda} is the local interaction budget for this specified decomposition. It is closely related to the local norm of \cite[Definition~3.1]{DP25}: that norm is obtained by minimizing twice the maximum incident sum over all decompositions into supported Hermitian terms, rather than fixing the pair decomposition above.

    The same maximum incident interaction strength appears as the parameter $\alpha$ in \cite[p.~25, after Eq.~(114)]{Bergamaschi26}, where it is further bounded by degree times pairwise interaction strength. Here we retain the individual weights in $\Lambda$.

    Since we will be referring to supports of terms often, we consider the following definition.
    \begin{definition}[Notation for a cluster]
        For a sequence of terms of the Hamiltonian given by the vectors $\bar i, \bar j \in [n]^k,$ with $k \geq 0$ and $\bar i <_v \bar j$, we let
        \begin{equation}
            S_{\bar i, \bar j} = \supp(H_{i_1, j_1}) \cup \cdots \cup \supp(H_{i_k, j_k}) \subseteq\left\{i_1,\,j_1,\,\ldots,\,i_k,\,j_k\right\},
        \end{equation}
        be the support of the associated terms.
    \end{definition}
    The following nested-commutator expansion is also the starting point of \cite[Lemma~3.4 and Eq.~(41)]{Bergamaschi26}. We retain ordered interaction sequences, including repetitions, to track their weights and supports.
    \begin{lemma} \label{Expansion with nested}
        Given our Hamiltonian $H$, for all operators \(X\) and all \(\alpha \in \mathbb{C}\), the following equality holds:
        \[
        e^{\alpha H} X e^{-\alpha H} = \sum_{k=0}^\infty
        \sum_{\substack{\bar i , \bar j \in [n]^k\\\bar i <_v \bar j}} \frac{ \alpha^k}{k!}
        \,[H_{i_k, j_k},[\dots[H_{i_2,j_2},[H_{i_1,j_1}, X]]\dots]].
        \]
        Further, the size of the \(k\)th order term can be bounded as follows
        \[
        \left\|[H_{i_k, j_k},[\dots[H_{i_2,j_2},[H_{i_1,j_1}, X]]\dots]]\right\| \leq 2^k \left(\prod_{\ell=1}^k w_{i_\ell,j_\ell} \right) \|X\|.
        \]
        Moreover, this nested commutator is zero unless, for all \(\ell \in [k]\), $\{i_\ell,\,j_\ell\}$ intersects the support of \(X\) or \(\{i_1,\,j_1,\,\ldots,\,i_{\ell-1},\,j_{\ell-1}\}\). When $(\bar i, \bar j)$ satisfies this criterion, we call it a cluster from \(\supp(X)\).
    \end{lemma}

    \begin{proof}
        We define
        \[
        F(\alpha):= e^{\alpha H} X e^{-\alpha H},
        \]
         and we observe that $F$ solves the following Cauchy problem
        \begin{equation*}
            \begin{cases}
                F'(\alpha) = \operatorname{ad}_H(F(\alpha)) \\
                F(0) = X,
            \end{cases}
        \end{equation*}
        where
        \[
        \operatorname{ad}_H(\cdot) := [H, \cdot].
        \]
         By the uniqueness of the solution we obtain
        \[F(\alpha) = e^{\alpha \operatorname{ad}_H}(X) = \sum_{k=0}^{\infty} \frac{\alpha^k}{k!}\operatorname{ad}_H^k(X).\]
        Thus,
        \begin{align*}
        e^{\alpha H} X e^{-\alpha H} = \sum_{k=0}^{\infty} \frac{\alpha^k}{k!} [H,[H,\dots[H,X]]] = \sum_{k=0}^\infty
        \sum_{\substack{\bar i , \bar j \in [n]^k\\\bar i <_v \bar j}} \frac{\alpha^k}{k!}
        \,[H_{i_k, j_k},[\dots[H_{i_2,j_2},[H_{i_1,j_1}, X]]\dots]].
        \end{align*}
        To bound the operator norm of each nested commutator, observe that for arbitrary
        operators \(A\) and \(B\), \( \|[A,B]\| \leq 2\|A\|\|B\| \). Applying this estimate recursively yields
        \[
        \left\|[H_{i_k, j_k},[\dots[H_{i_2,j_2},[H_{i_1,j_1}, X]]\dots]]\right\| \leq 2^k \left(\prod_{\ell=1}^k w_{i_\ell,j_\ell} \right)\|X\|.
        \]
        Finally, we note that the nested commutator $ \left[H_{i_k, j_k},[\dots[H_{i_1,j_1}, X]\dots]\right] $ vanishes whenever, for some \(1\leq \ell\leq k\), $H_{i_\ell,j_\ell}$ commutes with \( [H_{i_{\ell-1}, j_{\ell-1}},[\dots[H_{i_1,j_1}, X]\dots]] \).
        A sufficient condition for this commutation is that
        \[
        \supp(H_{i_\ell, j_\ell}) \cap \left(
        \supp(H_{i_{\ell-1}, j_{\ell -1}}) \cup
        \cdots \cup \supp(H_{i_1, j_1}) \cup
        \supp(X) \right) = \varnothing,
        \]
        since in that case the two operators act on disjoint subsystems and therefore
        commute.
    \end{proof}
    The next estimates adapt the rooted-tree summation and successive leaf elimination of \cite[Lemma~3.10]{Bergamaschi26}. The marked-site sum also parallels \cite[Corollary~3.8]{Bergamaschi26}; here ordered edge occurrences and transport-support costs give the factors needed for our update matrices.
    \begin{lemma}[Weighted cluster estimates] \label{lem:weighted-clusters}
        For every integer \(t\ge1\), the following two inequalities hold
        \begin{align} \label{eq:rooted-cluster}
            \sup_{i\in [n]}\frac1{t!}
            \sum_{\substack{ \bar i, \bar j \in [n]^t, \bar i <_v \bar j\\
                            (\bar i, \bar j) \, {\rm clust.\ from}\ i}}
            \prod_{\ell=1}^t w_{i_\ell,j_\ell}  &\le\Lambda^t, \\  \label{eq:marked-cluster}
            \sup_{j \in [n]}\frac1{t!}\sum_{i=1}^n
            \sum_{\substack{\bar i, \bar j \in [n]^t, \bar i <_v \bar j\\ (\bar i, \bar j) \, {\rm clust.\ from}\ i, \ j \in S_{\bar i, \bar j}}}
            |S_{\bar i, \bar j}|\prod_{\ell=1}^t w_{i_\ell,j_\ell}
            &\le(t+1)^2\Lambda^t.
        \end{align}
    \end{lemma}

    \begingroup\postdisplaypenalty=10000
    \begin{proof}
        Set $w_{i,i}=0$ for each $i \in [n]$. Fix a root $i$ and a cluster
        \[
        (\bar i,\bar j), \bar i <_v \bar j
        \]
         of order $t$ starting from $i$. We associate to the cluster a sequence of physical labels
        $ x_0,\ldots,x_t\in[n] $ and parent indices
        \[
        p_\ell\in \{0,\ldots,\ell-1\},\qquad 1\le \ell\le t.
        \]
        Start by setting $x_0=i$. We construct the remaining labels inductively. Before processing the $\ell$-th edge $(i_\ell,j_\ell)$, every site that has already been reached by the cluster appears among $ x_0,\ldots,x_{\ell-1} $. This is true initially, and it remains true because at step $\ell$ the only endpoint that can be newly reached is the one recorded as $x_\ell$. For the edge $(i_\ell,j_\ell)$, choose one endpoint $u_\ell$ that has already been reached. If both endpoints have already been reached, choose the smaller one according to the fixed ordering of $[n]$. Let $v_\ell$ denote the other endpoint. Define
        \[
        p_\ell:=\min \{q<\ell:x_q=u_\ell\}, \qquad x_\ell:=v_\ell.
        \]
        Thus, $p_\ell$ records the first occurrence of the already-reached endpoint $u_\ell$, while $x_\ell$ records the other endpoint $v_\ell$. Then an edge is given by $ e_\ell=\{x_{p_\ell},x_\ell\}$. The formal graph \(T\) on vertices \(\{0,\ldots,t\}\), with edges \(\{p_\ell,\ell\}\), is a tree rooted at \(0\), because every \(\ell\ge1\) has one parent with smaller index.

        The word ``formal'' is important: the labels \(x_0,\ldots,x_t\) need not be distinct.  Thus a new formal vertex is created at every edge occurrence even when that edge closes a cycle or repeats a physical edge.
        The encoding is injective. From $x_0=i$ and the pairs $(p_\ell,x_\ell)$ one recovers inductively $x_{p_\ell}$, hence the unordered edge
        \[
        \{i_\ell,j_\ell\}=\{x_{p_\ell},x_\ell\}
        \]
         for every $\ell$. Since $i_\ell<j_\ell$, the unordered edge determines $(i_\ell,j_\ell)$.

        Now, fix any such tree \(T\), any formal vertex \(r\), and the physical label
        \(x_r=a\).  Define
        \begin{equation}
        \label{eq:tree-label-sum}
        Z_T(r,a):=
        \sum_{\substack{x:V(T)\to[n]\\x_r=a}}
        \prod_{\{q,s\}\in E(T)}w_{x_q,x_s}
        \le\Lambda^{|E(T)|}=\Lambda^t.
        \end{equation}
        Thus the sum is over all \(n^t\) assignments of physical sites to the formal vertices other than \(r\); the assigned sites need not be distinct.  We prove the inequality in \eqref{eq:tree-label-sum} by induction on \(|E(T)|\).

        If \(T\) has no edges, then \(Z_T(r,a)=1\).  Otherwise, a tree has a leaf \(q\ne r\). Let \(s\) be its unique neighbor and put \(T':=T\setminus\{q\}\). First sum only the label \(x_q\), keeping all labels on \(T'\) fixed.  Since the only factor containing \(x_q\) is \(w_{x_q,x_s}\), we obtain
        \[
        \begin{aligned}
        Z_T(r,a) =\sum_{\substack{x:V(T')\to[n]\\x_r=a}} \left(\prod_{\{u,v\}\in E(T')}w_{x_u,x_v}\right) \left(\sum_{b=1}^n w_{b,x_s}\right) \le\Lambda Z_{T'}(r,a) \le\Lambda^{|E(T)|}.
        \end{aligned}
        \]
        The first inequality follows from the definition of $\Lambda$ and the second is the induction hypothesis for \(T'\). For example, if \(T\) is the two-edge path \(r-u-v\), then
        \[
        Z_T(r,a) = \sum_{b,c=1}^n w_{a,b}w_{b,c} = \sum_{b=1}^n w_{a,b}\left(\sum_{c=1}^n w_{b,c}\right) \le \Lambda\sum_{b=1}^n w_{a,b} \le \Lambda^2.
        \]
        The induction repeats exactly this operation, one leaf at a time. Notice also that, unless \(r=0\), the original root label \(x_0\) is among the variables summed in \(Z_T(r,a)\).  This observation is what will absorb the sum over the original root \(i\) in the marked estimate.

        There are
        \[
        \prod_{\ell=1}^t\ell=t!
        \]
         possible parent arrays \(p=(p_1,\ldots,p_t)\); denote their formal tree by \(T_p\).  For a fixed \(p\), the total weight of all formal physical-label assignments with \(x_0=i\) is exactly \(Z_{T_p}(0,i)\), and \eqref{eq:tree-label-sum} bounds it by \(\Lambda^t\).  The assignments arising from encoded clusters form a subset of these formal assignments. Since the cluster encoding is injective,
        \begin{equation} \label{subsresult in the cluster proof}
        \sum_{\substack{\bar i, \bar j \in [n]^t, \bar i <_v \bar j\\
                          (\bar i,\bar j) \ {\rm cluster\ from}\ i}}
        \prod_{\ell=1}^tw_{i_\ell, j_\ell}
        \le t!\Lambda^t.
        \end{equation}
        After division by \(t!\), this proves \eqref{eq:rooted-cluster}.

        Fix the marked physical site \(j\).  Given $\bar i, \bar j \in [n]^t$, every site in \(S_{\bar i, \bar j}\) appears at least once among the formal labels \(x_0,\ldots,x_t\).  Consequently,
        \[
        \ind_{\{j\in S_{\bar i, \bar j}\}}
        \le\sum_{r=0}^t \ind_{\{x_r=j\}}.
        \]
        After summing over the original root \(i=x_0\), injectivity of the
        encoding and \eqref{eq:tree-label-sum} give
        \begin{align*}
        &\sum_{i=1}^n
        \sum_{\substack{\bar i, \bar j \in [n]^t\\
                          (\bar i, \bar j) \ {\rm clust. \ from}\ i\\
                          j\in S_{\bar i, \bar j}}}
        \prod_{\ell=1}^tw_{i_\ell, j_\ell} \le
        \sum_p\sum_{r=0}^t Z_{T_p}(r,j)
        \le t!(t+1)\Lambda^t.
        \end{align*}
        For \(r=0\), \(Z_{T_p}(0,j)\) covers clusters rooted at \(i=j\).  For
        \(r\ne0\), the definition of \(Z_{T_p}(r,j)\) sums \(x_0=i\) over all
        sites, so it covers the sum over every possible original root.
        Finally, \(|S_{\bar i, \bar j}|\le t+1\).  Multiplying the last bound by
        \(t+1\) and dividing by \(t!\) proves
        \eqref{eq:marked-cluster}.  Thus the four
        factors in the marked count are completely explicit:
        \[
        \underbrace{t!}_{\text{parent arrays}}\,
        \underbrace{(t+1)}_{\text{formal occurrence of }j}\,
        \underbrace{(t+1)}_{\text{transport support}}\,
        \underbrace{\Lambda^t}_{\text{summed physical labels}}.
        \]
    \end{proof}
    \endgroup

    \subsection{Definition and properties of the Lindbladian}
    We use the Gibbs-preserving Lindbladian of \cite[Definition~4.1]{BLMT25}. Its jump operators and coherent correction fit the structural theory and canonical completion of KMS-symmetric generators \cite{FU10,AC21} and the exact KMS sampler constructions discussed in \cite{DLL25}. The rapid-mixing theorem below establishes uniqueness in the temperature regime considered here.
    \begin{definition} [Balanced Lindbladian]\label{Lindbladian}
        Given the target Gibbs state
        \[
        \sigma = \frac{e^{-\beta H}}{\tr(e^{-\beta H})}
        \]
        at inverse temperature $\beta > 0$, we consider the Lindblad evolution given by
        \[
        \cL(\rho) = \sum_{j=1}^n \cL_j(\rho) = \sum_{j=1}^n \sum_{P \in \cP_j} \cL^P(\rho),
        \]
        where
        \[
        \cL^P(\rho) = -\ii [G^P,\rho] + A^P \rho (A^P)^\dagger - \frac{1}{2}\{(A^P)^\dagger A^P,\rho\}.
        \]
        In other words, there is a Lindbladian term for every 1-local Pauli matrix. For a Hermitian proposal operator $P$, not necessarily a Pauli operator, we define the corresponding jump operator by
        \begin{equation} \label{def of jump operators}
        A^P = \sigma^{1/4} P \sigma^{-1/4}.
        \end{equation}
        Further, we take $G^P$ such that in an eigenbasis of $\sigma$, with matrix indices $a,b\in[N]$,
        \[
        G_{a,b}^P = -\ii \left(
        \frac{\sigma_a^{1/2}-\sigma_b^{1/2}} {2  \sigma_a^{1/4}\sigma_b^{1/4}(\sigma_a^{1/2}+\sigma_b^{1/2})} \right) \left( \sum_{c=1}^N \sigma_c^{1/2} P_{a,c} P_{c,b} \right),
        \]
        where $\{\sigma_c\}_{c=1}^{N}$ are the eigenvalues of $\sigma$.
        Note that $G^P$ is Hermitian and can be rewritten as follows, where $\circ$ denotes entry-wise product:
        \begin{equation} \label{expression for GP}
        G^P = -\frac{\ii }{2}\left((A^P)^\dagger A^P\right)\circ \left\{\frac{\sigma_a^{1/2}-\sigma_b^{1/2}} {\sigma_a^{1/2}+\sigma_b^{1/2}} \right\}_{a,b\in[N]}.
        \end{equation}
    \end{definition}

    This Lindbladian satisfies the KMS detailed balance condition, which also implies the stationarity of the Gibbs state.
    \begin{lemma}[KMS detailed balance {\cite[Lemma A.3]{BLMT25}}] \label{KMS detailed balance}
        The Lindbladians $\cL^P$ and $\cL $ introduced in \cref{Lindbladian} satisfy the Kubo-Martin-Schwinger (KMS) detailed balance, i.e. for any Hermitian matrix $P$, for every operator $B$,
        \begin{equation} \label{KMS condition}
            \cL^P (\sigma^{1/2} B \sigma^{1/2}) = \sigma^{1/2} (\cL^{P})^\dagger(B) \sigma^{1/2}.
        \end{equation}
        Consequently, the Gibbs state $\sigma$ is a stationary state for these Lindbladians: $\cL^P(\sigma)=0$ for every Hermitian matrix $P$ and $\cL(\sigma)=0$.
    \end{lemma}

    \begin{remark}
        For each fixed Hermitian operator $P$, the coherent term $G^P$
        in \eqref{expression for GP} is chosen so that the individual
        generator $\cL^P$ satisfies KMS detailed balance. With the jump operator $A^P$ defined in \eqref{def of jump operators} fixed, this coherent term is uniquely determined up to addition of a real scalar
        multiple of the identity. Explicitly, every other Hermitian
        coherent term yielding KMS detailed balance for the same
        single-jump generator has the form
        \[
            \widetilde G^P = G^P + c_P I,
            \qquad c_P \in \mathbb R.
        \]
    \end{remark}

    We now observe that, by means of the Fourier transform, the coherent terms \(G^P\) appearing in the Lindbladian can be expressed as averages over real-time evolutions.

    \begin{lemma}[Fourier transform expression for the coherent term {\cite[Lemma 4.4]{BLMT25}}] \label{lem:Fourier-transform-coherent-term}
        For the Lindbladian in \cref{Lindbladian}, for $i \in [n]$ and $P \in \cP_i$, we have
        \[
            G^P = \int_{-\infty}^{\infty} e^{\ii H\omega} \left( \frac{1}{2\ii} \left[P\sigma^{1/2}P,\sigma^{-1/2}\right] \right) e^{-\ii H\omega} f_\beta(\omega)\,d\omega,
        \]
        where
        \[
            f_\beta(\omega) = \frac{1}{\beta\cosh\left(\frac{2\pi}{\beta}\omega\right) }.
        \]
    \end{lemma}

    \section{Quantum Dobrushin Condition}

    We follow the Dobrushin-to-mixing framework of \cite[Section~5]{BLMT25}. As in our update-matrix definition, we allow Hermiticity-preserving maps without a positivity assumption; this formulation accommodates the Euler approximations used below. In this section, we define the corresponding quantum generalization of the classical Dobrushin influence matrix associated with a linear map. In the quantum setting, we use the quantum Wasserstein distance as a metric on quantum states. Once this matrix has been defined, we derive a condition showing that contractivity of its \(1\to 1\) norm implies rapid mixing of the corresponding linear map.

    \begin{definition}[Quantum Dobrushin condition] \label{Q Dobrushin Condition}
        Let $\Phi$ be a trace-preserving linear map admitting a decomposition
        \[
            \Phi = \Phi_1 + \cdots + \Phi_n,
        \]
        where each $\Phi_i$ is a Hermiticity-preserving linear map that maps
        traceless matrices to traceless matrices.
        We define the associated Dobrushin influence matrix $D^{(\Phi)}$ to have entries
        \begin{equation}
        D^{(\Phi)}_{i,j} = \max_{\rho,\sigma \text{ $j$-neighbors}} \|\Phi_i(\rho-\sigma)\|_{W_1},
        \end{equation}
        where the maximization is over $\rho$ and $\sigma$ which are density matrices of quantum states and $j$-neighbors, meaning that
        \( \tr_j(\rho-\sigma)=0 \). When $\rho$ and $\sigma$ are neighbors, then $ \|\rho-\sigma\|_{W_1} \leq 1$.
        Given some $0 < \gamma < 1$, we then define the quantum Dobrushin condition to be
        \begin{equation} \label{Dobrushin condition}
            \|D^{(\Phi)}\|_{1 \to 1} = \max_{x:\|x\|_1=1} \|D^{(\Phi)}x\|_1
            \leq 1 - \frac{\gamma}{n}.
        \end{equation}
    \end{definition}
    In this definition, we interpret each \(\Phi_i\) in the decomposition of \(\Phi\) as the component of \(\Phi\) associated with the update at site \(i\). In the classical setting, the Dobrushin influence matrix quantifies the worst-case influence of such local updates using the total variation distance. Our quantum analogue follows the same principle, but uses the quantum Wasserstein distance instead.

    Given the aforementioned condition, we show that the dynamics obtained by repeatedly applying the map $\Phi$ mixes rapidly. We first consider the following proposition.

    \begin{prop} [Defining the Dobrushin influence matrix in terms of Wasserstein growth {\cite[Fact 5.3]{BLMT25}}] \label{Dobrushin condition in terms of W1}
        For a linear map on quantum states $\Phi$, the Dobrushin influence matrix $D^{(\Phi)}$ as defined in \cref{Q Dobrushin Condition} can be equivalently defined as
        \begin{equation} \label{Dobrushin matrix}
        D^{(\Phi)}_{i,j} = \max_{\substack{X = X^\dagger,\tr_j(X)=0 \\ X\neq 0}} \frac{\left\|\Phi_i(X)\right\|_{W_1}}{\|X\|_{W_1}} = \max_{\substack{X=X^\dagger,\tr_j(X)=0 \\ \|X\|_{W_1}=1}}
        \left\|\Phi_i(X)\right\|_{W_1}.
        \end{equation}
    \end{prop}

    Now we prove a key lemma which shows that if the Dobrushin condition holds, then $\Phi$ contracts the space of quantum states in the quantum Wasserstein metric.
    \begin{lemma} [Relating contraction to influence {\cite[Lemma 5.4]{BLMT25}}] \label{Lemma contraction to influence}
        Let $\rho$ and $\varrho$ be trace-$1$ Hermitian matrices. For a trace-preserving map $\Phi$ with corresponding Dobrushin influence matrix
        $D^{(\Phi)}$ and some integer $\tau \geq 0$,
        \[
        \left\|\Phi^{\tau}(\rho)-\Phi^{\tau}(\varrho)\right\|_{W_1}
        \leq
        \left\|D^{(\Phi)}\right\|_{1\to 1}^{\tau}\,
        \|\rho-\varrho\|_{W_1}.
        \]
    \end{lemma}

    We can finally prove that the Dobrushin condition in \eqref{Dobrushin condition} implies the rapid mixing of the dynamics obtained by repeatedly applying the map $\Phi$.

    \begin{lemma}[Dobrushin condition implies rapid mixing {\cite[Lemma 5.2]{BLMT25}}] \label{Dobrushin condition implies rapid mixing}
         Let $\Phi$ be a trace-preserving linear map satisfying the quantum Dobrushin condition (\cref{Q Dobrushin Condition}) with parameter $\gamma$. Then $\Phi$ admits a unique fixed point $\sigma$  among trace-one Hermitian operators. Moreover, for every initial Hermitian operator $\rho$ with $\tr(\rho)=1$, and every error parameter $\varepsilon>0$,
        \[
        \|\Phi^\tau(\rho)-\sigma\|_{\mathrm{tr}} \leq \varepsilon \|\rho-\sigma\|_{\mathrm{tr}},
        \]
        for every integer $\tau\geq0$ satisfying
        \begin{equation} \label{minimum time rapid mixing}
        \tau \geq \frac{n}{\gamma}\log\frac{n}{\varepsilon}.
        \end{equation}
    \end{lemma}

    \begin{remark}
        In this section, $\Phi$ is not necessarily a quantum channel and therefore need not map positive semidefinite matrices to positive semidefinite matrices. Instead, we consider $\Phi$ on the space of trace-one Hermitian matrices, where the Dobrushin condition guarantees a unique fixed point. In our application, stationarity under the Lindbladian implies that the Gibbs state is also a fixed point of the first-order approximation considered below. Whenever the Dobrushin condition holds, this state is therefore the unique trace-one Hermitian fixed point and is a valid quantum state.
    \end{remark}

    We conclude this section with a result relating the Dobrushin influence matrix (\cref{Q Dobrushin Condition}) to update matrices (\cref{Update matrix}).
    \begin{lemma} \label{relating dobrushin matrix to update matrix}
        Let $\Phi$ be a trace-preserving linear map admitting a decomposition
        \[
        \Phi = \sum_{i=1}^n \Phi_i,
        \]
         where each $\Phi_i$ is a Hermiticity-preserving linear map that maps traceless matrices to traceless matrices. For each $i \in [n]$ let $Q^{(i)}$ be an update matrix for $\Phi_i$ and let $D^{(\Phi)}$ denote the Dobrushin influence matrix associated with this decomposition of $\Phi$. Then
        \begin{equation*}
            D^{(\Phi)}_{i,j} \leq \mathbf{1}^{\mathsf T} Q^{(i)}e_j, \qquad i,j \in [n].
        \end{equation*}
        Consequently,
        \begin{equation*}
            \|D^{(\Phi)}\|_{1 \to 1} \leq \left\| \sum_{i=1}^n Q^{(i)} \right\|_{1 \to 1}.
        \end{equation*}
    \end{lemma}

    \begin{proof}
        Let $i,j \in [n]$ and choose a maximizer
        \[X_j \in \operatorname*{arg\,max}_{\substack{X=X^\dagger, \, \tr_j(X) = 0,\\ \wnorm {X} = 1} } \wnorm {\Phi_{i}(X)}.\]
        The transport plan with $X_j$ in component $j$ and zero in every other component has cost vector $e_j$. By the definition of update matrix, $\Phi_i(X_j)$ therefore admits a transport plan with cost vector $y$ satisfying $y\leq_v Q^{(i)}e_j$. Hence,
        \[ D^{(\Phi)}_{i,j} = \wnorm {\Phi_i(X_j)} \leq \mathbf{1}^{\mathsf T} y \leq \mathbf{1}^{\mathsf T} Q^{(i)} e_j.\]
        Finally, since all these matrices are entrywise nonnegative, their induced $1\to1$ norms are their largest column sums. Thus,
        \begin{equation*}
            \|D^{(\Phi)}\|_{1 \to 1} = \max_{j \in [n]} \sum_{i=1}^n D^{(\Phi)}_{i,j} \leq \max_{j \in [n]} \mathbf{1}^{\mathsf T} \left( \sum_{i=1}^n Q^{(i)} \right) e_j =  \left \| \sum_{i=1}^n Q^{(i)} \right \|_{1 \to 1}
        \end{equation*}
    \end{proof}

    \section{Evolution of Wasserstein transport plans} \label{Evolution of Wasserstein transport plans}
    Given a fixed positive parameter $\delta >0$, we study the evolution of the transport plans of the first-order approximation
    \begin{equation} \label{def of Phi}
    \Phi := \cI + \frac{\delta}{n} \cL  = \sum_{i = 1}^n  \left[\frac{1}{n} \left( \cI +\ \delta \cL_i\right)\right] =:  \sum_{i=1}^n \Phi_i.
    \end{equation}
    In particular, by the definition of the Lindbladian (\cref{Lindbladian}), each summand satisfies
    \begin{equation} \label{def of Phi_i}
    n\Phi_i(\rho) = \rho + \delta \sum_{P \in \cP_i}\left(  -\ii [G^P,\rho] + A^P \rho (A^P)^\dagger - \frac{1}{2}\{(A^P)^\dagger A^P,\rho\}\right)
    \end{equation}
    In view of \cref{Dobrushin condition in terms of W1} and \cref{relating dobrushin matrix to update matrix}, we seek update matrices for the maps $\Phi_i$ to control the Dobrushin influence matrix of $\Phi$.

    We decompose each term $n\Phi_i$ into a coherent part
    \begin{equation} \label{coherent part}
        \Phi_{i,\mathrm c}(\rho) := \rho + \delta \sum_{P \in \cP_i}  -\ii [G^P,\rho],
    \end{equation}
    and a dissipative part
    \begin{equation} \label{dissipative part}
        \Phi_{i,\mathrm d}(\rho) := \rho + \delta \sum_{P \in \cP_i} \left( A^P \rho (A^P)^\dagger - \frac{1}{2}\{(A^P)^\dagger A^P,\rho\} \right).
    \end{equation}
    These maps satisfy
    \[
    n\Phi_i=\Phi_{i,\mathrm c}+\Phi_{i,\mathrm d}-\cI.
    \]

    We first study their update matrices separately and then combine the resulting estimates.

    In update-matrix estimates, we will use $o(\delta)$ to denote a matrix $E_\delta$ such that $\|E_\delta\|_{1\to1}/\delta\to0$ as $\delta\to0$; similarly, $\mathcal O(\delta^2)$ means
    \[
    \|E_\delta\|_{1\to1}=\mathcal O(\delta^2).
    \]

    All model and auxiliary parameters are held fixed. The same remainder matrix works for every input operator and transport plan.

    \subsection{Dissipative update}

    Our goal in this subsection is to prove \cref{dissipative evolution}, which describes how the dissipative part of the Lindblad generator associated with site $i$ affects transport plans. We follow the transport-plan construction of \cite[Section~4.3.2, Claim~4.15]{BLMT25}, while retaining the individual interaction strengths in the cluster expansions. Using \cref{lem:weighted-clusters}, we then bound the resulting update matrices in terms of $\Lambda$, without introducing graph-growth parameters. Combining connected cluster expansions and summing over a marked site also parallels \cite[proof of Corollary~3.8]{Bergamaschi26}.

    \begin{teo} [Wasserstein norm of the dissipative evolution] \label{dissipative evolution}
        Let $i \in [n]$ and $ \Phi_{i,\mathrm d}$ as in \eqref{dissipative part} where the jump operators $A^P$ are specified in \eqref{def of jump operators}. Assume $\beta \Lambda < 1 $. Then, for all sufficiently small $\delta >0$, the map $\Phi_{i,\mathrm d}$ admits an update matrix of the form
        \begin{equation} \label{update Q_i^d}
        Q_{\mathrm d}^{(i)} := I +\delta(- 4 E_{\{i\}}+ R_{\mathrm d}^{(i)})+\mathcal O(\delta^2),
        \end{equation}
        where
        \begin{equation} \label{def of R^d_i}
        R_{\mathrm d}^{(i)} =  12 \sum_{t \geq 1}  \frac{ \beta ^{t}} {t !}\sum_{\substack{\bar r, \bar s \in [n]^t, \bar r <_v \bar s\\ (\bar r, \bar s) \text{ cluster from } i}} \left(\prod_{\ell=1}^t w_{r_\ell,s_\ell} \right) E_{S_{\bar r, \bar s}},
        \end{equation}
        and the remainder is entrywise nonnegative and has $1\to1$ norm $\mathcal O(\delta^2)$ as $\delta$ goes to $0$. In particular,
        \begin{equation}
            \left \| \sum_{i=1}^n  R_{\mathrm d}^{(i)} \right \|_{1 \to 1} \leq 12 \sum_{t \geq 1} (t+1)^2 (\beta \Lambda)^t.
        \end{equation}

    \end{teo}

    \begin{proof}
        Throughout the proof, we assume $0<\delta<1/4$. It suffices to prove the statement for normalized single-site contributions. More precisely, fix $j\in[n]$ and let $X$ be a Hermitian operator satisfying $\tr_j(X)=0$, and
        \[
        \|X\|_{W_1} = \frac{1}{2}\|X\|_{\mathrm{tr}} = 1.
        \]
         We will construct a transport plan for $\Phi_{i,\mathrm d}(X)$ with cost vector bounded entrywise by $Q_{\mathrm d}^{(i)}e_j$. By scaling and combining the resulting plans (\cref{prop on lin}), this implies the desired estimate for every traceless Hermitian operator.

        By \cref{Expansion with nested}, we can write our jump operators $A^P$ as
        \[
        A^P = e^{-\frac{\beta}{4}H} P e^{\frac{\beta}{4}H} = \sum_{k \geq 0} \sum_{\substack{\bar i , \bar j \in [n]^k\\ \bar i <_v \bar j}} A^P_{\bar i , \bar j}
        \]
        where each summand $A_{\bar i , \bar j}^P$ is a nested commutator,
        \begin{equation*}
            A_{\bar i , \bar j}^P := \left (-\frac{\beta}{4 } \right)^k \frac{1}{k!} [H_{i_k, j_k},[\dots[H_{i_2,j_2},[H_{i_1,j_1}, P]]\dots]].
        \end{equation*}
        More specifically, $A_{\bar i , \bar j}^P$ is only non-zero when $(\bar i , \bar j)$ is a cluster starting from $i$, and in that case,
        \begin{equation}  \label{eq:jump-cluster-norm}
        \|A_{\bar i , \bar j}^P\| \leq \left(\frac{ \beta}{2} \right)^k \frac{1}{k!} \prod_{\ell=1}^k w_{i_\ell,j_\ell}.
        \end{equation}
        Therefore
        \begin{equation} \label{lab 1}
            \sum_{k \geq 0} \sum_{\substack{\bar i <_v \bar j \in [n]^k \\ (\bar i , \bar j) \text{ cluster from } i}} \|A_{\bar i , \bar j}^P\|  \leq   \sum_{k \geq 0}\frac{1}{k!} \left(\frac{ \beta }{2 } \right)^k  \sum_{\substack{\bar i <_v \bar j \in [n]^k \\ (\bar i , \bar j) \text{ cluster from } i}}\prod_{\ell=1}^k w_{i_\ell,j_\ell} \leq  \sum_{k \geq 0} \left(\frac{ \beta \Lambda }{2 } \right)^k,
        \end{equation}
        where we used the bound in \eqref{subsresult in the cluster proof} from the proof of \cref{lem:weighted-clusters}. Under the hypothesis $\beta \Lambda < 1$, the right-hand side converges.

        Now, we break our argument into two cases: $i=j$ and $i\ne j$.\\
        \textbf{Case $i=j$.}
        This case is where the ``contraction'' appears. We can decompose
        $\Phi_{i,\mathrm d}$ in the following way.
        \begin{align} \label{diss 1}
            \Phi_{i,\mathrm d}(X)
            &= X+\delta \sum_{P\in\cP_i} \sum_{k,h \geq 0} \sum_{\substack{\bar i <_v \bar j \in [n]^k  \\ \bar \ell <_v \bar m \in [n]^h}} \left( A_{\bar i, \bar j}^P X (A_{\bar \ell, \bar m}^P)^\dagger -\frac{1}{2}\{(A_{\bar \ell, \bar m}^P)^\dagger A_{\bar i, \bar j}^P,X\}
            \right) \nonumber\\
            &=  \underbrace{ X+\delta \sum_{P\in \cP_i} (PXP-X) }_{=:Y_0} \nonumber\\
            &\ + \underbrace{\delta \sum_{P\in \cP_i}
            \sum_{(k,h) \ne (0,0)} \sum_{\substack{\bar i <_v \bar j \in [n]^k  \\ \bar \ell <_v \bar m \in [n]^h}} \left( A_{\bar i, \bar j}^P X (A_{\bar \ell, \bar m}^P)^\dagger -\frac{1}{2}\{(A_{\bar \ell, \bar m}^P)^\dagger A_{\bar i, \bar j}^P,X\}
            \right) }_{=:Y_1}.
        \end{align}

        The first summand $Y_0$ is obtained from $X$ by a one-site depolarizing channel, which contracts $X$ when $\tr_i(X)=0$. Concretely, for an $M$ with $\tr_i(M)=0$,
        \[
        M+\sigma_X^{(i)}M\sigma_X^{(i)} + \sigma_Y^{(i)}M\sigma_Y^{(i)} + \sigma_Z^{(i)}M\sigma_Z^{(i)}=0.
        \]
        Thus,
        \[
        Y_0=\left(1-4\delta\right)X.
        \] Because $\tr_i(Y_0)=0$, by \cref{Lemma W1 rel tr}, it has a transport plan with cost
        \begin{equation} \label{bound 1 diss}
            \frac{1}{2}\|Y_0\|_{\mathrm{tr}} \,e_i = \frac{1}{2}\left(1-4\delta\right)\|X\|_{\mathrm{tr}} \,e_i = \left(1-4\delta\right) \, e_i.
        \end{equation}
        Now we consider $Y_1$. We first symmetrize the elements of the sum and we obtain,
        \[ \sum_{P \in \cP_i} \! \sum_{(k,h) \ne (0,0)} \! \sum_{\substack{\bar i <_v \bar j \in [n]^k  \\ \bar \ell <_v \bar m \in [n]^h}} \frac{\delta}{2} \left( A_{\bar i, \bar j}^P X (A_{\bar \ell, \bar m}^P)^\dagger + A_{\bar \ell, \bar m}^P X (A_{\bar i, \bar j}^P)^\dagger -\frac{1}{2}\{(A_{\bar \ell, \bar m}^P)^\dagger A_{\bar i, \bar j}^P + (A_{\bar i, \bar j}^P)^\dagger A_{\bar \ell, \bar m}^P,X\} \right).
        \]
        By \cref{Lemma few site}, for fixed indices of the three sums, each term of the sum has a transport plan with cost vector bounded entry-wise by
        \[
        \frac{\delta}{2}\left(4\|A_{\bar i, \bar j}^P\|\|A_{\bar \ell, \bar m}^P\|\|X\|_{\mathrm{tr}} \,e_{{\supp(A_{\bar i, \bar j}^P) \, \cup \, \supp(A_{\bar \ell, \bar m}^P) }} \right) = 4\delta \|A_{\bar i, \bar j}^P\| \|A_{\bar \ell, \bar m}^P\| \, e_{{\supp(A_{\bar i, \bar j}^P) \, \cup \,\supp(A_{\bar \ell, \bar m}^P) }}.
        \]
        Consequently, $Y_1$ has a transport plan with cost vector at most
        \begin{align} \label{bound 2 diss}
            &4\delta \sum_{P\in \cP_i}
            \sum_{(k,h) \ne (0,0)}  \sum_{\substack{\bar i <_v \bar j \in [n]^k  \\ \bar \ell <_v \bar m \in [n]^h}}
            \|A_{\bar i, \bar j}^P\| \|A_{\bar \ell, \bar m}^P\| \, e_{{\supp(A_{\bar i, \bar j}^P) \, \cup \,\supp(A_{\bar \ell, \bar m}^P) }} \notag \\
            \leq_v& 4\delta \sum_{P\in \cP_i}
            \sum_{(k,h) \ne (0,0)}  \sum_{\substack{\bar i <_v \bar j ; \  \bar \ell <_v \bar m  \\ i \in S_{\bar i, \bar j} \cup S_{\bar \ell, \bar m}} }
            \|A_{\bar i, \bar j}^P\| \|A_{\bar \ell, \bar m}^P\| \, e_{S_{\bar i, \bar j} \cup \,S_{\bar \ell, \bar m}} .
        \end{align}
        The inequality follows because, for every nonzero contribution, the union of the operator supports is contained in $S_{\bar i,\bar j}\cup S_{\bar \ell,\bar m}$. Moreover, since $(k,h)\ne(0,0)$, every nonzero contribution contains a nonempty cluster whose support includes $i$, which justifies restricting the sum to $i\in S_{\bar i,\bar j}\cup S_{\bar \ell,\bar m}$. Combining \eqref{bound 1 diss} with \eqref{bound 2 diss}, we have that for $i=j$, $\Phi_{i,\mathrm d}(X)$ has a transport plan with cost vector at most
        \[
        (1-4\delta)e_i +  4\delta \sum_{P\in \cP_i} \sum_{(k,h) \ne (0,0)}  \sum_{\substack{\bar i <_v \bar j ; \  \bar \ell <_v \bar m  \\ i \in S_{\bar i, \bar j} \cup S_{\bar \ell, \bar m}} } \|A_{\bar i, \bar j}^P\| \|A_{\bar \ell, \bar m}^P\| \, e_{S_{\bar i, \bar j} \cup \,S_{\bar \ell, \bar m}}
        \]
        \textbf{Case $i \ne j$.}
        We now consider an input transport-plan component associated with site $j$ and a jump operator $A^P$ with $P\in\cP_i$. We decompose $A^P$ into the sum of terms whose cluster supports contain $j$ and those whose cluster supports avoid $j$.
        \[
        A^P = \underbrace{\sum_{k \geq 0}\sum_{\bar i <_v \bar j \in [n]^k} \ind_{\{j\notin S_{\bar i, \bar j}\}}\,A_{\bar i , \bar j}^P}_{A^{P,\hat{j}}} + \underbrace{\sum_{k \geq 0} \sum_{\bar i <_v \bar j\in [n]^k} \ind_{\{j\in S_{\bar i, \bar j}\}}\,A_{\bar i , \bar j}^P}_{A^{P,j}}.
        \]
        This gives a corresponding decomposition of $\Phi_{i, \mathrm d}(X)$.
        \begin{align} \label{decomp of phi diss}
            \Phi_{i,\mathrm d}(X) &= \underbrace{ X + \delta \sum_{P\in \cP_i}
            \left( A^{P,\hat{j}} X (A^{P,\hat{j}})^\dagger - \frac{1}{2}\left\{(A^{P,\hat{j}})^\dagger A^{P,\hat{j}},X\right\} \right)
            }_{=:Z_0} \nonumber\\
            &\quad + \delta \sum_{P\in \cP_i}
            \Bigl( A^{P,j}X(A^{P,\hat{j}})^\dagger + A^{P,\hat{j}}X(A^{P,j})^\dagger + A^{P,j}X(A^{P,j})^\dagger
            \nonumber\\
            & \quad -\frac{1}{2}\left\{ (A^{P,\hat{j}})^\dagger A^{P,j} + (A^{P,j})^\dagger A^{P,\hat{j}} + (A^{P,j})^\dagger A^{P,j}, X \right\}  \Bigr).
        \end{align}
        We denote the remaining summand in \eqref{decomp of phi diss} by
        \[
        Z_1:=\Phi_{i,\mathrm d}(X)-Z_0.
        \]

        First, we want to construct a transport plan for $Z_0$. We write
        \[
        Z_0 =\frac{1}{3}  \sum_{P\in \cP_i}
            \left( X + 3 \delta \left (A^{P,\hat{j}} X (A^{P,\hat{j}})^\dagger - \frac{1}{2}\left\{(A^{P,\hat{j}})^\dagger A^{P,\hat{j}},X\right\} \right)\right)
        \]
        Each summand in this average has vanishing partial trace over $j$, since $A^{P,\hat j}$ acts away from $j$. Hence $Z_0$ also has vanishing partial trace over $j$ and, by \cref{Lemma W1 rel tr}, admits a transport plan with cost vector
        \[
        \frac12\|Z_0\|_{\mathrm{tr}}e_j.
        \]
        By \cref{size of lind step}, the trace norm of each unweighted summand in the average is at most
        \[  \left(1 + \frac{1}{2} \left(3 \delta\left\| (A^{P,\hat{j}})^\dagger A^{P,\hat{j}} \right\| \right)^2 \right) \| X\|_{\mathrm{tr}}.\]
        Consequently, combining the bounds for the individual summands, $Z_0$ admits a transport plan with cost vector bounded entrywise by
        \begin{align} \label{eq 1 diss case 2}
            \frac{1}{3}  \sum_{P\in \cP_i} \frac{1}{2} \left(1 + \frac{1}{2} \left( 3\delta\| (A^{P,\hat{j}})^\dagger A^{P,\hat{j}}\| \right)^2\right) \| X\|_{\mathrm{tr}} \, e_j  \leq_v \left(1+\mathcal O(\delta^2)\right) e_j.
        \end{align}
        In the final bound, we used the fact $ A^{P, \hat j} $ is well-defined whenever $\beta \Lambda < 1$ as a consequence of \eqref{lab 1}, and thus its norm is finite. As for the second term $Z_1$, we expand in clusters and symmetrize the summands as in the case $i=j$.
        \[
        Z_1 = \!\!\!\sum_{P\in \cP_i}  \! \sum_{\substack{\bar i <_v \bar j ;   \bar \ell <_v \bar m  \\ j \in {S_{\bar i, \bar j} \cup S_{\bar \ell, \bar m}} }} \frac{\delta}{2}\left( A^{P}_{\bar i, \bar j} X (A^{P}_{\bar \ell, \bar m})^\dagger + A^{P}_{\bar \ell, \bar m} X (A^{P}_{\bar i, \bar j})^\dagger -\frac{1}{2}\left\{  (A^{P}_{\bar \ell, \bar m})^\dagger  A^{P}_{\bar i, \bar j} + (A^{P}_{\bar i, \bar j})^\dagger  A^{P}_{\bar \ell, \bar m}, X \right\}  \right)
        \]
        By applying \cref{Lemma few site}, since in the sum we have $j \in {S_{\bar i, \bar j}\cup S_{\bar \ell, \bar m}} $, each summand of the sum above has a transport plan with cost vector at most
        \[\frac{\delta}{2} 4 \| A^P_{\bar i, \bar j}\| \| A^P_{\bar \ell, \bar m} \| \| X \|_{\mathrm{tr}} e_{S_{\bar i, \bar j} \cup S_{\bar \ell, \bar m}} = 4 \delta  \| A^P_{\bar i, \bar j}\| \| A^P_{\bar \ell, \bar m} \| e_{S_{\bar i, \bar j} \cup S_{\bar \ell, \bar m}}.\]
        Thus, $Z_1$ has a transport plan with cost at most
        \begin{equation} \label{eq 2 diss case 2}
            4 \delta \sum_{P\in \cP_i}  \! \sum_{(k,h) \ne (0,0)}\sum_{\substack{\bar i <_v \bar j ;   \bar \ell <_v \bar m  \\ j \in {S_{\bar i, \bar j} \cup S_{\bar \ell, \bar m}} }}  \| A^P_{\bar i, \bar j}\| \| A^P_{\bar \ell, \bar m} \| \, e_{S_{\bar i, \bar j} \cup S_{\bar \ell, \bar m}}.
        \end{equation}

        Combining \eqref{eq 1 diss case 2}, \eqref{eq 2 diss case 2}, we have that for $i \ne j$, $\Phi_{i, \mathrm d}(X)$ has a transport plan with cost vector at most

        \[\left(1+\mathcal O(\delta^2)\right) e_j + 4 \delta \sum_{P\in \cP_i}  \! \sum_{(k,h) \ne (0,0)}\sum_{\substack{\bar i <_v \bar j ;   \bar \ell <_v \bar m  \\ j \in {S_{\bar i, \bar j} \cup S_{\bar \ell, \bar m}} }}  \| A^P_{\bar i, \bar j}\| \| A^P_{\bar \ell, \bar m} \| \, e_{S_{\bar i, \bar j} \cup S_{\bar \ell, \bar m}}.  \]
        \textbf{Combining the two cases.}
        What we have shown is that, for a general $j$, $\Phi_{i,\mathrm d}(X)$ has a transport plan with cost vector at most
        \begin{align}
            &\left(1 - 4\delta\,\ind_{\{i=j\}} + \mathcal O(\delta^2)\right)e_j + 4 \delta \sum_{P\in \cP_i} \sum_{(k,h) \ne (0,0)} \sum_{\substack{\bar i <_v \bar j ; \bar \ell <_v \bar m  \\ j \in {S_{\bar i, \bar j} \cup S_{\bar \ell, \bar m}} }}  \| A^P_{\bar i, \bar j}\| \| A^P_{\bar \ell, \bar m} \| e_{S_{\bar i, \bar j} \cup S_{\bar \ell, \bar m}} \notag
            \\
            =& \left( \left(1 + \mathcal O(\delta^2)\right)I - 4\delta E_{\{i\}} + 4 \delta \sum_{P\in \cP_i} \sum_{(k,h) \ne (0,0)} \sum_{\substack{\bar i <_v \bar j ; \bar \ell <_v \bar m  \\ j \in {S_{\bar i, \bar j} \cup S_{\bar \ell, \bar m}} }}   \| A^P_{\bar i, \bar j}\| \| A^P_{\bar \ell, \bar m} \| E_{S_{\bar i, \bar j} \cup S_{\bar \ell, \bar m}}
            \right)e_j . \label{eq bound}
        \end{align}
        We have therefore constructed a transport plan for $\Phi_{i,\mathrm d}(X)$ whose cost vector is bounded componentwise by a matrix applied to the input cost vector. We consider some bounds to simplify the expression of this matrix. So far, we have used very little about the jump operators: we only used that they can be decomposed into operators of different-sized supports. We recall \eqref{eq:jump-cluster-norm}, and observe that $A_{\bar i , \bar j}^P = 0$ unless $(\bar i,  \bar j)$ is a cluster starting from $i$, the site where $P$ is supported. So \eqref{eq bound} is entry-wise bounded by
        \begin{equation}\label{eq bound 2}
        \left(
        \begin{aligned}
            &\left(1+\mathcal O(\delta^2)\right)I-4\delta E_{\{i\}}\\
            &{}+12\delta\sum_{(k,h)\ne(0,0)}
            \sum_{\substack{\bar i<_v\bar j,\ \bar\ell<_v\bar m\\
                  j\in S_{\bar i,\bar j}\cup S_{\bar\ell,\bar m}\\
                  (\bar i,\bar j),(\bar\ell,\bar m)\text{ clusters from }i}}
            \left(\frac\beta2\right)^{k+h}
            \left(\prod_{s=1}^k w_{i_s,j_s}\right)
            \left(\prod_{r=1}^h w_{\ell_r,m_r}\right)
            \frac{E_{S_{\bar i,\bar j}\cup S_{\bar\ell,\bar m}}}{k!h!}
        \end{aligned}
        \right)e_j.
        \end{equation}

        Now, we re-parametrize and concatenate $(\bar i, \bar j)$ with $(\bar \ell, \bar m)$ as two vectors, $(\bar r, \bar s) \in [n]^{k+h}$ of length $k+h>0$. When $(\bar i, \bar j)$ and $(\bar \ell, \bar m)$ are clusters from $i$, so is $(\bar r, \bar s)$. Further,
        \[
        S_{\bar i, \bar j} \cup S_{\bar \ell, \bar m} = S_{\bar r, \bar s}.
        \]
         Thus, \eqref{eq bound 2} is bounded by
        \begin{align}
            & \left( (1 + \mathcal O(\delta^2))I - 4\delta E_{\{i\}} + 12 \delta \sum_{t >0 } \sum_{\substack{\bar r <_v \bar s \in [n]^t \\ j \in {S_{\bar r, \bar s}} \\ (\bar r,    \bar s) \text{ cl. from } i}} \left( \sum_{s = 0}^t   \left (\frac{\beta }{2 } \right)^{t} \left(\prod_{\ell=1}^t w_{r_\ell,s_\ell} \right)\frac{1}{s! (t-s) !} \right)  E_{S_{\bar r, \bar s} }\right)e_j \notag\\
            =& \left( (1 + \mathcal O(\delta^2))I - 4\delta E_{\{i\}} + 12 \delta \sum_{t > 0}  \frac{\beta^t}{t !} \sum_{\substack{\bar r <_v \bar s \in [n]^t \\ j \in {S_{\bar r, \bar s}} \\ (\bar r, \bar s) \text{ cl. from } i}} \left(\prod_{\ell=1}^t w_{r_\ell,s_\ell} \right)  E_{S_{\bar r, \bar s} }\right)e_j, \label{final estim for dissipative part}
        \end{align}
        which proves \eqref{update Q_i^d}.

        Finally, defining $R_{\mathrm d}^{(i)}$ as in \eqref{def of R^d_i}, we get from \eqref{eq:marked-cluster}
        \[\left\| \sum_{i=1}^n R_{\mathrm d}^{(i)} \right\|_{1 \to 1} = 12 \max_{j \in [n]} \sum_{t \geq 1} \frac{ \beta^{t}} {t !}  \sum_{i=1}^n  \sum_{\substack{\bar i <_v \bar j \in [n]^t \\ (\bar i, \bar j)\ \text{cluster from } i \\ j \in S_{\bar i, \bar j}}} \left(\prod_{\ell=1}^t w_{i_\ell,j_\ell} \right)|S_{\bar i, \bar j}| \leq  12 \sum_{t \geq 1} (t+1)^2 ( \beta \Lambda)^{t}. \]

    \end{proof}

   \subsection{Coherent update} \label{subsec growth cohe evolution}
   For the coherent contribution, substituting the representation obtained in \cref{lem:Fourier-transform-coherent-term} into \eqref{coherent part}, we obtain
    \begin{equation}
        \Phi_{i,\mathrm c}(\rho) = \rho - \sum_{P \in \cP_i} \frac{\ii  \delta}{\beta} \left[\int_{\mathbb{R}} \sech\left(\frac{2\pi \xi}{\beta}\right)e^{\ii H\xi} B^P e^{-\ii H\xi} \,d\xi,\rho \right],
    \end{equation}
    where for each $P \in \cP_i$, we define
    \[
        B^P = \frac1{2\ii } [P\sigma^{1/2}P,\sigma^{-1/2}] = \Skew\left(Pe^{-\beta H/2}Pe^{\beta H/2}\right).
    \]

    The following theorem bounds the first-order contribution of the coherent evolution to the update matrices in terms of the total interaction strength $\Lambda$.

        \begin{teo}[All-time coherent update-matrix bound] \label{coherent evolution}
        Suppose that
        \[
             12\beta\Lambda<2\pi.
        \]
        Then the coherent update $\Phi_{i,\mathrm c}$ admits an update matrix of the form
        \[
            I+\delta R_{\mathrm c}^{(i)}+o(\delta),
        \]
        where $R_{\mathrm c}^{(i)}$ is nonnegative and
        \begin{equation} \label{eq:coherent-column-bound-all-time}
            \left\|\sum_{i=1}^n R_{\mathrm c}^{(i)}\right\|_{1\to1} \leq \frac{600}{2\pi-12\beta\Lambda} \sum_{k\geq2}(k+1)^2(\beta\Lambda)^k.
        \end{equation}
    \end{teo}

    In order to obtain this result, we first expand the operators $B^P$ into local Hermitian terms and show that the contributions of orders zero and one in $\beta$ vanish. The next step is to control the coherent update generated by the real-time evolution of each such term. To explain the role of the reversible lift, let $K$ be a local Hermitian operator and write

    \[
    \tau_\xi(K)=e^{\ii H\xi}Ke^{-\ii H\xi}
    \]
     and $\cC_K=-\ii [K,\cdot]$. Unitary conjugation gives
    \[
        \cI +\delta\cC_{\tau_\xi(K)}
        =\tau_\xi\circ(\cI +\delta\cC_K)\circ\tau_{-\xi}.
    \]
    Estimating the transport costs of these three factors separately can lose the cancellation between $\tau_\xi$ and $\tau_{-\xi}$: the resulting bound may already increase the cost at $\delta=0$. Such a bound cannot capture the competition with the dissipative contraction, which is of order $\delta$.

    We therefore construct linear maps $\mathfrak U_\xi$ acting on
    transport-plan tuples $\boldsymbol X=(X_j)_{j\in[n]}$. These maps preserve the partial-trace constraints and satisfy
    \[
        \sum_j(\mathfrak U_\xi\boldsymbol X)_j
        =\tau_\xi\!\left(\sum_jX_j\right),
        \qquad
        \mathfrak U_\xi\mathfrak U_{-\xi}=I.
    \]
    The first identity says that $\mathfrak U_\xi$ lifts the physical evolution to transport plans; the second retains its exact reversibility. If $\mathfrak k_K$ similarly lifts the generator $\cC_K$ to transport plans, then
    \[
        \mathfrak U_\xi(I+\delta\mathfrak k_K)\mathfrak U_{-\xi}
        =I+\delta\mathfrak U_\xi\mathfrak k_K\mathfrak U_{-\xi}.
    \]
    Applied to an input plan, this gives a plan for the desired coherent update. Crucially, its zeroth-order term is exactly the original plan, so no transport cost is introduced at $\delta=0$.

    To estimate the first-order term, we separate a componentwise
    unitary evolution generated by observables acting trivially on the site indexing each component. This evolution preserves both the partial-trace constraints and each component's trace norm. We bound the remaining contribution using weighted real-time propagation estimates, obtaining an update matrix of the  form $I+\delta R_{\xi,K}+\mathcal O(\delta^2)$. Finally, we sum these bounds over the imaginary-time clusters and integrate against the Fourier kernel to prove \cref{coherent evolution}.

    \subsubsection{Imaginary-time expansion and cancellations} We apply the nested-commutator method of \cref{Expansion with nested}, also used in \cite[Lemma~3.4]{Bergamaschi26}, to the operator $B^P$. The cancellation of its first two orders is checked explicitly below.

    \begin{lemma}[Cluster expansion of \(B^P\)] \label{lem:B-cluster-cancellation}
        Let \(P \in \cP_i\). For every $\bar i,\bar j \in [n]^k, \bar i <_v \bar j$, define
        \[
        B_{\bar i, \bar j}^P :=\Skew\left( \left(-\frac{\beta}{2} \right)^k \frac{1}{k!}\,P [H_{i_k,j_k},[\dots[H_{i_1,j_1},P]\dots]] \right).
        \]
        Then
        \begin{equation} \label{eq:B-cluster}
        B^P =\sum_{k\ge2} \sum_{\substack{\bar i <_v \bar j \in [n]^k\\
                          (\bar i, \bar j) \ {\rm clust. \ from}\ i}}
        B_{\bar i, \bar j}^P, \quad \text{with } \ \left\|{B_{\bar i,\bar j}^P} \right \|
        \le  \frac{\beta^k}{k!}\prod_{\ell=1}^k w_{i_\ell,j_\ell}.
        \end{equation}
        In particular, every \(B_{\bar i, \bar j}^P\) is Hermitian and is supported on  \(S_{\bar i, \bar j}\).
    \end{lemma}

    \begin{proof}
        Applying \cref{Expansion with nested}, we obtain
        \[
        e^{-\beta H/2}Pe^{\beta H/2}
        =\sum_{k\ge0} \left(-\frac{\beta}{2} \right)^k \frac{1}{k!} \sum_{\substack{\bar i,\bar j\in[n]^k\\\bar i<_v\bar j}}[H_{i_k,j_k},[\cdots[H_{i_1,j_1},P]\cdots]].
        \]
        Multiplication by \(P\) on the left and linearity of \(\Skew\) give the equality for $B^P$ in \eqref{eq:B-cluster}. The same lemma shows that each summand vanishes unless \((\bar i, \bar j)\) is a cluster from \(i\), and that its support is contained in \(S_{\bar i, \bar j}\). Moreover, since \( \|\Skew(X) \| \leq \| X \| \), we have
        \[
        \|B_{\bar i, \bar j}^P\| \le \frac{(\beta/2)^k}{k!} 2^k\prod_{\ell=1}^k w_{i_\ell, j_\ell} = \frac{\beta^k}{k!}\prod_{\ell=1}^kw_{i_\ell, j_\ell}.
        \]
        The definition of \(\Skew\) also shows directly that

        \[
        (B_{\bar i, \bar j}^P)^\dagger=B_{\bar i, \bar j}^P.
        \]
         Finally, for \(k=0\), the expression inside \(\Skew\) is \(P^2=I\), so that \( B_{\varnothing}^P=\Skew(I)=0 \). For \(k=1\) and an individual edge \(\{ i, j \}\), apart from the real coefficient \(-\beta/2\), we have
        \[
        \Skew\left(P\left[H_{i,j},P\right]\right) = \Skew\left(P H_{i,j}P - H_{i,j}\right) = 0.
        \]
        Thus every term of orders \(0\) and \(1\) vanishes, and
        the expansion starts at \(k=2\), as asserted.
    \end{proof}

    \subsubsection{Real-time propagation by a reversible transport-plan lift}
    Define the nonnegative symmetric matrix
    \begin{equation} \label{eq:M-def}
    M_{i,j}=
    \begin{cases}
    w_{i,j},&i\ne j,\\
    d_i :=\sum_{k \ne i} w_{i,k},&i=j.
    \end{cases}
    \end{equation}
    For a fixed column \(j\),
    \[
    \sum_{i=1}^n M_{i,j} = M_{j,j} + \sum_{i\ne j}M_{i,j} = d_j+\sum_{i \ne j} w_{i,j} = 2d_j.
    \]
    Consequently,
    \begin{equation} \label{eq:M-norm}
    \| M \|_{1\to 1} = \| M \|_{\infty\to\infty} =2\Lambda.
    \end{equation}
    For every \(\xi\in\mathbb R\), we define the linear map
    \begin{equation} \label{eq:tau-definition}
    \tau_\xi:\mathcal B(\cH_n) \longrightarrow\mathcal B(\cH_n), \qquad \tau_\xi(K):=e^{\ii H\xi}Ke^{-\ii H\xi},
    \end{equation}
    which defines the Heisenberg evolution for real time \(\xi\). Because \(e^{\ii  H\xi}\) is unitary, \(\tau_\xi\) is a \(*\)-automorphism. In particular,
    \[
    \tau_\xi(KL)=\tau_\xi(K)\tau_\xi(L),\qquad
    \tau_\xi(K^\dagger)=\tau_\xi(K)^\dagger,\qquad
    \|\tau_\xi(K) \|=\| K \|.
    \]
    The map need not preserve the spatial support of \(K\); controlling the spreading of that support is precisely the real-time transport problem that we want to study. We construct the transport plan itself. The construction begins by defining the following finite-dimensional space.
    \begin{definition}[Plan space, summation map, and block majorant]
    \label{def:plan-lift-space}
    Let \(\mathcal B_{\mathrm h}(\cH_n)\) be the real vector space of
    Hermitian operators and define
    \[
    \mathfrak X:=\bigoplus_{u=1}^n\mathcal B_{\mathrm h}(\cH_n),
    \qquad
    \mathfrak P_i:=\{X=X^\dagger:\tr_i X=0\},
    \qquad
    \mathfrak P:=\bigoplus_{i=1}^n\mathfrak P_i.
    \]
    An element
    \[
    \boldsymbol X=(X_1,\ldots,X_n)\in\mathfrak P
    \]
     is called a plan tuple. Its summation and cost maps are
    \[
    \Sigma\boldsymbol X:=\sum_{i=1}^n X_i,
    \qquad
    \cost(\boldsymbol X)_i:=\frac12 \|X_i \|_{\mathrm{tr}}.
    \]
    Thus \(\boldsymbol X\) is exactly a transport plan for \(\Sigma\boldsymbol X\), written as one vector.

    \noindent
    Let
    \[
    \mathfrak A:\mathfrak X\to\mathfrak X
    \]
     be a real-linear, Hermiticity-preserving map, and write
    \[
    \mathfrak A_{i,j} :  \mathcal B_{\mathrm h}(\cH_n) \to  \mathcal B_{\mathrm h}(\cH_n)
    \]
     for its block from input component \(j\) to output component \(i\). Its block trace-norm majorant is the nonnegative matrix
    \begin{equation} \label{eq:block-majorant-def}
    \left|\mathfrak A\right|_{\mathrm{bl},i,j} :=  \sup_{0\ne X=X^\dagger}
    \frac{\|\mathfrak A_{i,j}(X)\|_{\mathrm{tr}}}{\| X \|_{\mathrm{tr}}}.
    \end{equation}
    \end{definition}
    The block majorant is only a bookkeeping device. The triangle inequality gives, entry by entry,
    \begin{equation} \label{eq:block-majorant-rules}
    |\mathfrak A+\mathfrak B|_{\mathrm{bl}} \le|\mathfrak A|_{\mathrm{bl}}+|\mathfrak B|_{\mathrm{bl}}, \qquad |\mathfrak A\mathfrak B|_{\mathrm{bl}} \le|\mathfrak A|_{\mathrm{bl}}|\mathfrak B|_{\mathrm{bl}}.
    \end{equation}
    The second inequality follows from
    \[
    (\mathfrak A\mathfrak B)_{i,j}=\sum_{r=1}^n\mathfrak A_{i,r}\mathfrak B_{r,j}
    \]
     and
    \[
    \|(\mathfrak A\mathfrak B)_{i,j}(X)\|_{\mathrm{tr}} \le \sum_r|\mathfrak A|_{\mathrm{bl},i,r} |\mathfrak B|_{\mathrm{bl},r,j} \| X \|_{\mathrm{tr}}.
    \]
    In particular, if
    \[
    \boldsymbol x=\cost(\boldsymbol X),
    \]

    then
    \begin{equation}  \label{eq:block-majorant-cost}
    \cost(\mathfrak A\boldsymbol X) \leq_v|\mathfrak A|_{\mathrm{bl}}\boldsymbol x.
    \end{equation}
    For a site \(i\), define the normalized conditional expectation
    \begin{equation}
    \label{eq:conditional-expectation}
    \mathbb E_i(O):=\frac{I_i}{2}\otimes\tr_iO.
    \end{equation}
    It is completely positive, trace preserving, unital, idempotent, and
    self-adjoint for the Hilbert--Schmidt inner product.  Consequently it
    is contractive in trace norm on Hermitian operators and in operator
    norm on all operators. The following are true
    \begin{align} \label{eq:Pauli-twirl-Eu}
    \mathbb E_i(O) = \frac14 \left( O + \sum_{P \in \cP_i} P O P \right), \qquad
    \|\mathbb E_i(O) \|  \le\| O \|.
    \end{align}

    \begin{lemma} \label{lem:two-site-splitting}
    Let \(i,j \in [n], i \ne j\), and let \(Y=Y^\dagger\) satisfy \(\mathbb E_i\mathbb E_j(Y)=0\), or equivalently \(\tr_{\{i,j\}}Y=0\). Define
    \begin{align}
    D_i^{\{i,j\}}(Y)&:=\frac12\left(Y-\mathbb E_iY+\mathbb E_j Y\right).
    \label{eq:Da-def}
    \end{align}
    Then
    \begin{equation}
    \label{eq:two-site-splitting-properties}
    D_i^{\{i,j\}}(Y)+D_j^{\{i,j\}}(Y)=Y,
    \qquad
    \tr_iD_i^{\{i,j\}}(Y)=0,
    \end{equation}
    and
    \begin{equation} \label{eq:two-site-splitting-norm}
    \| D_i^{\{i,j\}}(Y) \|_{\mathrm{tr}} \leq \frac32 \| Y \|_{\mathrm{tr}}.
    \end{equation}
    \end{lemma}

    \begin{proof}
        Adding \eqref{eq:Da-def} with both indices $i$ and $j$ cancels both conditional expectation terms and gives \(Y\). Since conditional expectations on
        different sites commute and are idempotent,
        \begin{align*}
        \mathbb E_i D_i^{\{i,j\}}(Y)  =\frac12\left(\mathbb E_i Y - \mathbb E_iY +\mathbb E_i\mathbb E_j Y\right)=0.
        \end{align*}
        By \eqref{eq:conditional-expectation}, the identity \(\mathbb E_i Z=0\) is equivalent to \(\tr_i Z=0\). This proves the partial-trace identity. Finally, the triangle inequality and \eqref{eq:Pauli-twirl-Eu} give
        \[
        \|D_i^{\{i,j\}}(Y) \|_{\mathrm{tr}} \le\frac12\left(\| Y \|_{\mathrm{tr}} + \| \mathbb E_iY \|_{\mathrm{tr}} + \| \mathbb E_jY \|_{\mathrm{tr}} \right) \le\frac32 \| Y \|_{\mathrm{tr}}.
        \]
    \end{proof}
    \begin{definition}
        For a Hermitian operator \(L\), introduce the two commutator
        superoperators
        \begin{equation}  \label{eq:A-C-commutators}
        \mathcal A_L(X):=\ii [L,X], \qquad \cC_L(X):=-\ii [L,X].
        \end{equation}
    \end{definition}
    The first generates the Heisenberg evolution \(\tau_\xi\); the second
    is the coherent generator acting on a transport-plan component.  Both
    obey
    \begin{equation}
    \label{eq:commutator-trace-norm}
    \|\mathcal A_L(X) \|_{\mathrm{tr}},\ \| \cC_L(X) \|_{\mathrm{tr}} \le2 \| L \| \| X \|_{\mathrm{tr}}.
    \end{equation}
    \begin{lemma}[Exact reversible lift of real-time evolution] \label{lem:reversible-plan-lift}
        There exists a linear generator

        \[
        \mathfrak H:\mathfrak X \to\mathfrak X
        \]
         such that
        \[\mathfrak H\mathfrak P\subseteq\mathfrak P, \qquad \Sigma\mathfrak H=\mathcal A_H\Sigma\] and for
        \[
        \mathfrak U_\xi:=e^{\xi\mathfrak H},
        \]

        \begin{equation} \label{eq:lift-intertwining}
            \mathfrak U_\xi\mathfrak P\subseteq\mathfrak P, \qquad \Sigma\mathfrak U_\xi=\tau_\xi\Sigma,  \qquad \mathfrak U_\xi^{-1}=\mathfrak U_{-\xi};
        \end{equation}
        Moreover, taking $T_t:=e^{3tM}$ with $t\geq0$, one has
        \begin{equation} \label{eq:lift-block-bound}
            |\mathfrak U_\xi|_{\mathrm{bl}}\le T_{|\xi|} \quad\text{entrywise}, \qquad \|{T_t} \|_{1\to1} =\|{T_t}\|_{\infty\to\infty} \le e^{6\Lambda t}.
        \end{equation}
    \end{lemma}
    \begin{proof}
        We fix \(1 \leq i < j \leq n\) and define a linear map
        \[
        \mathfrak h_{i,j}: \mathfrak X \to \mathfrak X
        \]
         by describing what it does to one input component \(X_r\). If \(r\notin \{i,j\}\), put \(\mathcal A_{H_{i,j}}(X_r)\) in output component \(r\) and put zero in every other component. Otherwise, put \(D_i^{\{i,j\}}(\mathcal A_{H_{i,j}}(X_r))\) in output component \(i\) and \(D_j^{\{i,j\}}(\mathcal A_{H_{i,j}}(X_r))\) in output component \(j\). Since $H_{i,j}$ is supported on $\{i,j\}$, cyclicity of the partial trace over these sites gives, for every input component $X_r$,
        \begin{equation*}
            \tr_{\{i,j\}}(\mathcal{A}_{H_{i,j}}(X_r)) = \ii \tr_{\{i,j\}}(H_{i,j}X_r  - X_rH_{i,j}) = 0.
        \end{equation*}
        Thus, \cref{lem:two-site-splitting} applies to $\mathcal{A} _{H_{i,j}}(X_r)$. Moreover, if \(X_r\in\mathfrak P_r\) and \(r\notin \{i,j\}\), then \(H_{i,j}\) acts away from \(r\) and
        \[
        \tr_r\mathcal A_{H_{i,j}}(X_r) =0.
        \]
         If \(r\in \{i,j\}\), the two outputs belong to \(\mathfrak P_i\) and \(\mathfrak P_j\) by \eqref{eq:two-site-splitting-properties}. Hence

        \[
        \mathfrak h_{i,j}\mathfrak P\subseteq\mathfrak P.
        \]
         Moreover, the sum of the output components equals \(\mathcal A_{H_{i,j}}(X_r)\) in both cases. By linearity in all input components,
        \begin{equation} \label{eq:edge-lift-intertwining}
        \Sigma\mathfrak h_{i,j}=\mathcal A_{H_{i,j}}\Sigma.
        \end{equation}
        Define
        \begin{equation} \label{eq:H-lift-def}
            \mathfrak H:=\sum_{1\leq i < j \leq n}\mathfrak h_{i,j}.
        \end{equation}
        The sum is finite. The invariance just proved for each summand gives
        \[
        \mathfrak H\mathfrak P\subseteq\mathfrak P.
        \]
         Summing \eqref{eq:edge-lift-intertwining} and using the definition of \(H\) gives
        \[
        \Sigma\mathfrak H=\mathcal A_H\Sigma.
        \]

        Induction yields
        \[
        \Sigma\mathfrak H^t=\mathcal A_H^t\Sigma
        \]
         for every integer \(t\ge0\). Using the norm-convergent exponential series in the finite-dimensional space then gives
        \[
        \Sigma e^{\xi\mathfrak H} =e^{\xi\mathcal A_H}\Sigma =\tau_\xi\Sigma.
        \]
         The same exponential series shows that \(e^{\xi\mathfrak H}\) leaves the invariant subspace \(\mathfrak P\) invariant. Finally, it also implies
        \[
        e^{\xi\mathfrak H}e^{-\xi\mathfrak H}=I,
        \]
         which proves all three
        claims in \eqref{eq:lift-intertwining}. This exact inverse is the reversibility that will be retained later.

        Now, for a site \(r\), we set
        \[
        H^{\hat r}:=\sum_{i<j:\,r\notin \{i,j\}}H_{i,j}.
        \]
         Define the block-diagonal generator \(\mathfrak D \) by

        \[
        (\mathfrak D\boldsymbol X)_r := \mathcal A_{H^{\hat r}}(X_r),
        \]
         and put
        \[
        \mathfrak B:=\mathfrak H-\mathfrak D.
        \]
         Thus \(\mathfrak D\) contains exactly the edges that do not touch the current plan mark, while \(\mathfrak B\) contains the incident edges that can move that mark. We observe that every diagonal block of \(e^{t\mathfrak D}\) is unitary conjugation by \(e^{\ii  tH^{\hat r}}\).  It is therefore a trace-norm isometry.

        We now verify the block inequality explicitly. Fix an input column \(r\) and one incident edge \(\{r,i\}\), where without loss of generality $r<i$. On that column, \(\mathfrak B\) contains the edge lift \(\mathfrak h_{r,i}\). By definition, it has exactly two possibly nonzero output blocks: the block from \(r\) to \(r\) is \( D_r^{\{r,i\}}\!\left(\mathcal A_{H_{r,i}}(X_r)\right) \), and the block from \(r\) to \(i\) is
        \( D_i^{\{r,i\}}\!\left(\mathcal A_{H_{r,i}}(X_r)\right) \). For either output \(j\in\{r,i\}\), \eqref{eq:commutator-trace-norm} and \eqref{eq:two-site-splitting-norm} give
        \[
        \trnorm{D_j^{\{r,i\}}(\mathcal A_{H_{r,i}}(X_r))} \leq 3w_{r,i}\trnorm{X_r}.
        \]
        No edge not incident on \(r\) contributes to column \(r\), because all such edges were placed in \(\mathfrak D\). Therefore, on the diagonal block of that column,
        \[
        |\mathfrak B|_{\mathrm{bl},r,r} \leq 3\sum_{i <j :\{i,j\}\ni r} w_{i,j} = 3M_{r,r}.
        \]
        If \(i\ne r\), the only possible contribution to block \((i,r)\) is
        the unique pair edge \(\{i,r\}\), and analogously
        \[
        |\mathfrak B|_{\mathrm{bl},i,r} \le3w_{i,r}=3M_{i,r}.
        \]
         These two cases cover every matrix entry and prove
        \begin{equation} \label{eq:B-block-bound}
        |\mathfrak B|_{\mathrm{bl}} \le3\sum_{1 \leq i < j \leq n} w_{i,j}E_{\{i,j\}} = 3M.
        \end{equation}
        Finally, we prove the propagation bound. Assume that \(\xi=t\ge0\) and put
        \[
        \mathfrak V_t:=e^{-t\mathfrak D}e^{t(\mathfrak D+\mathfrak B)} = e^{-t\mathfrak D}\mathfrak U_t.
        \]
         Differentiation gives
        \[
        \frac{d}{dt}\mathfrak V_t = \mathfrak B(t)\mathfrak V_t, \qquad \mathfrak B(t):=e^{-t\mathfrak D}\mathfrak B e^{t\mathfrak D}, \qquad \mathfrak V_0=I.
        \]
        Pre- and post-composition by the block-diagonal trace-norm isometries
        does not change any block bound, so \( |\mathfrak B(t)|_{\mathrm{bl}}\le3M \). Integrating the differential equation gives
        \[
        \mathfrak V_t = I + \int_0^t\mathfrak B(s)\mathfrak V_s\,ds.
        \]
        The block-majorant product rule in \eqref{eq:block-majorant-rules} and \(|\mathfrak B(s)|_{\mathrm{bl}}\le3M\) imply the entrywise inequality
        \[
        |\mathfrak V_t|_{\mathrm{bl}}\le I+3\int_0^tM |\mathfrak V_s|_{\mathrm{bl}}\,ds.
        \]
        Apply the componentwise matrix form of Grönwall's inequality to each column of $|\mathfrak V_t|_{\mathrm{bl}}$. Since \(M\) is nonnegative, it gives
        \[
        |\mathfrak V_t|_{\mathrm{bl}}\le e^{3tM}=T_t
        \]
         entrywise. Multiplication on the left by \(e^{t\mathfrak D}\) again changes no block trace norm, so the same bound holds for \(\mathfrak U_t\). For \(\xi<0\), repeat the argument with generators \(-\mathfrak D,-\mathfrak B\); their isometry and block bounds are unchanged.  This proves the entrywise inequality in \eqref{eq:lift-block-bound}.

        Finally, \(M\) is nonnegative and symmetric. Its exponential is therefore nonnegative and symmetric, and \eqref{eq:M-norm} gives
        \[
        \opnorm{T_t}_{1\to1} \le e^{3t\opnorm M_{1\to1}} = e^{6\Lambda t}.
        \]
         Symmetry gives the same \(\infty\to\infty\) bound.
    \end{proof}

    \begin{lemma}[Lift of a local coherent kick] \label{lem:local-kick-lift}
        Let \(K=K^\dagger\) be supported on a nonempty ordered set
        \(S=\{s_1,\ldots,s_m\}\). Then there exists a linear map
        \[
        \mathfrak k_K:\mathfrak X\to\mathfrak X
        \]
         such that
        \begin{equation}
        \label{eq:kick-lift-properties}
        \mathfrak k_K\mathfrak P\subseteq\mathfrak P, \qquad \Sigma\mathfrak k_K=\cC_K\Sigma.
        \end{equation}
        Let \(\overline{\cC }_K\) denote the block-diagonal map that
        applies \(\cC_K\) to every component, and set
        \begin{equation}
        \label{eq:qK-def}
        q_K:=\mathfrak k_K-\overline{\cC }_K.
        \end{equation}
        Then
        \begin{equation}
        \label{eq:qK-bound}
        |q_K|_{\mathrm{bl}}\le6\opnorm K E_S.
        \end{equation}
    \end{lemma}

    \begin{proof}
    Put
    \[
    F_0:=\cI, \qquad F_r:=\mathbb E_{s_1}\cdots\mathbb E_{s_r},
    \qquad \Delta_r:=F_{r-1}-F_r \quad(1\le r\le m).
    \]
    The differences telescope:
    \begin{equation}
    \label{eq:local-telescope-sum}
    \sum_{r=1}^m\Delta_r =\cI-\mathbb E_{s_1}\cdots\mathbb E_{s_m}.
    \end{equation}
    Moreover,
    \begin{equation} \label{properties of delta}
    \mathbb E_{s_r}\Delta_r=0, \qquad \trnorm{\Delta_r(Y)}\le2\trnorm Y.
    \end{equation}
    We define $\mathfrak k_K$ in the following way by linearity. For one input component \(X_i\), if \(i\notin S\), we put \(\cC_K(X_i)\) in output component \(i\). If \(i\in S\), we put \(\Delta_r(\cC_K(X_i))\) in output component \(s_r\), for each \(1\le r\le m\).

    Because \(K\) is supported on \(S\), by partial-trace cyclicity
    \[
    \mathbb E_{s_1}\cdots\mathbb E_{s_m} \cC_K(X_i)=0.
    \]
     Therefore \eqref{eq:local-telescope-sum} shows that the output components sum to \(\cC_K(X_i)\) when \(i\in S\). When \(i\notin S\), this is true by definition. Thus \(\Sigma\mathfrak k_K=\cC_K\Sigma\). If \(i\in S\), every output at \(s_r\) has zero partial trace there, because of \eqref{properties of delta}. If \(i\notin S\) and \(X_i\in\mathfrak P_i\), then \(K\) acts away from \(i\), so \( \tr_i\cC_K(X_i) = 0 \). This proves \eqref{eq:kick-lift-properties}.

    For \(i\notin S\), \(\mathfrak k_K\) and \(\overline{\cC }_K\) agree on input component \(i\), so the corresponding column of \(q_K\) is zero. For \(i\in S\), \eqref{eq:commutator-trace-norm} and \eqref{properties of delta} give \( \trnorm{\Delta_r\cC_K(X_i)} \le4\opnorm K\trnorm{X_i} \), for $1\le r\le m$. For the diagonal block \(s_r=i\), subtracting \(\cC_K(X_i)\) adds at most \(2\|K\|\trnorm{X_i}\).  Thus every block in \(S\times S\) is bounded by \(6\|K\|\), and every other block vanishes. This is precisely \eqref{eq:qK-bound}.
    \end{proof}
    We will also need the defect of the real-time lift from the uniform componentwise Heisenberg generator. Let \(\overline{\mathcal A}_L\) apply \(\mathcal A_L\) on every diagonal block, and define for $1 \leq i<j \leq n$
    \begin{equation} \label{eq:qedge-def}
    q_{i,j}:=\mathfrak h_{i,j}-\overline{\mathcal A}_{H_{i,j}}, \qquad
    q_H:=\sum_{i < j }q_{i,j} =\mathfrak H-\overline{\mathcal A}_H.
    \end{equation}

    \subsubsection{Comparison bounds for evolved coherent updates.}

    \begin{lemma} \label{lem:qedge-source}
        For every edge \(1\leq i < j \leq n\),
        \begin{equation} \label{eq:qedge-block-bound}
        |q_{i,j}|_{\mathrm{bl}}\le 5w_{i,j}E_{\{i,j\}}.
        \end{equation}
        For a Hermitian observable \(L\), we define
        \begin{equation} \label{eq:eta-def}
        \eta_i(L):=\opnorm{L-\mathbb E_iL}.
        \end{equation}
        Then
        \begin{equation} \label{eq:qedge-commutator-bound}
        \left|[q_{i,j},\overline{\cC }_L]\right|_{\mathrm{bl}} \le 20w_{i,j} \left(\eta_i(L)+\eta_j(L)\right)E_{\{i,j\}}.
        \end{equation}
    \end{lemma}

    \begin{proof}
        For an input component \( r \notin \{i,j\}\), \(\mathfrak h_{i,j}\) equals the uniform block action \(\mathcal A_{H_{i,j}}\), so \(q_{i,j}\) vanishes. For an input \( r \in \{i,j\}\), each endpoint output of \(\mathfrak h_{i,j}\) has block norm at most \(3w_{i,j}\), as proved in \eqref{eq:B-block-bound}. In the diagonal output block we additionally subtract \(\mathcal A_{H_{i,j}}\), whose norm is at most \(2w_{i,j}\). Thus, \eqref{eq:qedge-block-bound} follows.

        Put \(L_0:=\mathbb E_i\mathbb E_jL\). This observable is proportional to the identity on \(i,j\), hence is supported away from \(\{i,j\}\). The uniform map \(\overline{\cC }_{L_0}\) commutes with \(\mathfrak h_{i,j}\): it commutes with \(\overline{\mathcal A}_{H_{i,j}}\), with \(\mathbb E_i,\mathbb E_j\), and with the mere reassignment of an operator from one tuple component to another. Consequently,
        \[
        [q_{i,j},\overline{\cC }_{L_0}]=0.
        \]
         Let \(R:=L-L_0\). Contractivity of \(\mathbb E_i\) in operator norm and commutation of the two conditional expectations give
        \begin{align}   \label{eq:R-eta-bound}
        \opnorm R =\opnorm{L-\mathbb E_i\mathbb E_jL} \le\opnorm{L-\mathbb E_iL} +\opnorm{\mathbb E_i(L-\mathbb E_jL)} \le\eta_i(L)+\eta_j(L).
        \end{align}
        The block-diagonal map \(\overline{\cC }_R\) has block majorant at most \(2\|R\|I\), by \eqref{eq:commutator-trace-norm}. Combining this fact, \eqref{eq:block-majorant-rules}, and \eqref{eq:qedge-block-bound} gives
        \begin{align*}
        |[q_{i,j},\overline{\cC }_L]|_{\mathrm{bl}} =|[q_{i,j},\overline{\cC }_R]|_{\mathrm{bl}} \le |q_{i,j}|_{\mathrm{bl}}\left(2\opnorm R I\right) + \left(2\opnorm R I\right)|q_{i,j}|_{\mathrm{bl}} \le20 w_{i,j}\opnorm R E_{\{i,j\}}.
        \end{align*}
        Substitution of \eqref{eq:R-eta-bound} proves \eqref{eq:qedge-commutator-bound}.
    \end{proof}

    The following finite-system estimate is a weighted Lieb--Robinson bound using local oscillations in the spirit of \cite{RW96}. The use of an integral comparison for local oscillations also parallels \cite[proof of Lemma~4.7]{Bergamaschi26}; the estimate here concerns operator-norm propagation under Hamiltonian evolution.

    \begin{lemma} \label{lem:local-oscillation-LR}
        Let \(K=K^\dagger\) be supported on \(S\), and put \(K_\xi:=\tau_\xi(K)\). Then, entrywise,
        \begin{equation} \label{eq:eta-LR-bound}
        \left(\eta_i(K_\xi)\right)_{i=1}^n \le 2\opnorm K e^{4|\xi|M} e_S.
        \end{equation}
        In particular,
        \begin{equation} \label{eq:eta-LR-infinity}
        \max_i \eta_i(K_\xi) \le2\opnorm K e^{8\Lambda|\xi|}.
        \end{equation}
    \end{lemma}

    \begin{proof}
        For a Hermitian observable \(O\), define
        \begin{equation} \label{eq:chi-def}
        \chi_i(O):= \sup\left\{\opnorm{[O,U]}:U\text{ is unitary and supported on } \{i\}\right\}.
        \end{equation}
        The exact Pauli-twirl identity in \eqref{eq:Pauli-twirl-Eu} gives
        \begin{align} \label{eta leq chi}
        \eta_i(O) =\left\|\frac14\sum_{P\in \cP_i} (O-POP)\right\| = \left\| \frac14 \sum_{P \in \cP_i} [O,P]P\right\| \le \frac14 \sum_{P \in \cP_i}\opnorm{[O,P]} \le\chi_i(O).
        \end{align}
        If \(B\) is supported on \(\{i,j\}\), then \(\mathbb E_i\mathbb E_j(O)\) is proportional to the identity on those two sites and hence commutes with \(B\). Thus
        \begin{align}
        \opnorm{[O,B]} =\opnorm{[O-\mathbb E_i\mathbb E_j(O),B]} &\le 2\opnorm B\opnorm{O-\mathbb E_i\mathbb E_jO} \notag \\
        &\le 2\opnorm B\left(\eta_i(O)+\eta_j(O)\right) \le 2\opnorm B\left(\chi_i(O)+\chi_j(O)\right).  \label{eq:two-site-oscillation}
        \end{align}
        Now we assume that \(\xi=t\ge0\) and we set \(O_t:=\tau_t(K)\). We fix a
        unitary \(U\) supported on \(\{i\}\). Since \(\dot O_t=\ii[H,O_t]\), the Jacobi identity gives
        \begin{align}
        \frac d{dt}[O_t,U] =\ii[[H,O_t],U] =\ii[H,[O_t,U]]-\ii[O_t,[H,U]].
        \label{eq:commutator-evolution}
        \end{align}
        Thus, using Duhamel's formula
        \[
        [O_t,U]  = \tau_t([K,U]) -\ii \int_0^t \tau_{t-s}\left([O_s,[H,U]]\right)\,ds,
        \]
        which implies
        \begin{equation} \label{eq:chi-integral-pre}
        \opnorm{[O_t,U]} \le\opnorm{[K,U]} +\int_0^t\opnorm{[O_s,[H,U]]}\,ds.
        \end{equation}
        Inside the integral, only edges incident on \(i\) occur in
        \[
        [H,U]=\sum_{j<k: i \in \{j,k\}}[H_{j,k},U].
        \]
         Applying \eqref{eq:two-site-oscillation} to \(B=[H_{j,k},U]\), \eqref{eq:chi-integral-pre} becomes
        \[
        \opnorm{[O_t,U]} \le\opnorm{[K,U]} +4\int_0^t\sum_{j < k: i \in \{j,k\}}w_{j,k} \left(\chi_j(O_s)+\chi_k(O_s)\right)\,ds.
        \]
        Taking the supremum over \(U\) and writing
        \[
        \boldsymbol\chi(t)=(\chi_i(O_t))_i
        \]
         yields the entrywise vector inequality
        \begin{equation} \label{eq:chi-vector-integral}
        \boldsymbol\chi(t) \le\boldsymbol\chi(0)+4\int_0^tM\boldsymbol\chi(s)\,ds,
        \end{equation}
        where we used the fact that, for each $i\in[n]$,
        \[
            (M\boldsymbol\chi)_i = d_i\chi_i+\sum_{\ell\ne i}w_{i,\ell}\chi_\ell =\sum_{j<k : i\in\{j,k\}} w_{j,k}(\chi_j+\chi_k).
        \]
        The vector \(\boldsymbol\chi(t)\) is continuous: \(O_t\) depends continuously on \(t\), and the supremum in \eqref{eq:chi-def} is over the compact one-qubit unitary group. Since \(M\) is nonnegative, the componentwise matrix form of Grönwall's inequality, applied directly to \eqref{eq:chi-vector-integral}, gives
        \begin{equation} \label{eq:chi-matrix-gronwall}
        \boldsymbol\chi(t)\le e^{4tM}\boldsymbol\chi(0).
        \end{equation}
        If \(i\notin S\), then \([K,U]=0\); if \(i\in S\), then \(\|[K,U]\|\le 2\|K\|\). Therefore $\boldsymbol\chi(0)\le2\opnorm K e_S$. Together with \eqref{eta leq chi}, this proves \eqref{eq:eta-LR-bound} for positive time. For negative time, replace \(H\) by \(-H\); all weights \(w_{i,j}\), and therefore \(M\), are unchanged. Finally,
        \[
        \opnorm{e^{4|\xi|M}}_{\infty\to\infty} \le e^{4|\xi|\opnorm M_{\infty\to\infty}} = e^{8\Lambda|\xi|},
        \]
        which proves \eqref{eq:eta-LR-infinity}.
    \end{proof}

    For a nonnegative vector \(p\), we define the nonnegative matrix
    \begin{equation} \label{eq:S-map-def}
    \mathcal S(p) :=\sum_{1\le i<j\le n}w_{i,j}(p_i+p_j)E_{\{i,j\}}.
    \end{equation}
    The sum of a column \(i\) is
    \[
    \mathbf{1}^{\mathsf T}\mathcal S(p)e_i =2\sum_{j\ne i}w_{i,j}(p_i+p_j)
    =2(Mp)_i.
    \]
     Consequently,
    \begin{equation} \label{eq:S-map-norm}
    \opnorm{\mathcal S(p)}_{1\to1} \le 2\opnorm M_{\infty\to\infty}\opnorm p_\infty
    =4\Lambda\opnorm p_\infty.
    \end{equation}

    We now combine the reversible lift with the weighted propagation
    estimates to control the coherent updates generated by a family
    of time-evolved observables. The bound depends on the initial
    supports and operator norms through the matrix $A_0$, and holds
    for every real time. We will subsequently apply it to the local
    terms in the imaginary-time expansion of $B^P$.
    \begin{teo}[Update bounds for families of time-evolved observables] \label{lem:two-mark-input}
        Let \((K_\alpha)_\alpha\) be a finite or countable family of Hermitian operators. Suppose \(K_\alpha\) is supported on a nonempty set \(S_\alpha\), and choose \(a_\alpha\ge\opnorm{K_\alpha}\). Define
        \begin{equation} \label{eq:all-time-A0}
        A_0:=\sum_\alpha a_\alpha E_{S_\alpha},
        \end{equation}
        and assume that this nonnegative matrix series converges in \(1\to1\) norm.
        Then, for every $\xi\in\mathbb R$ and every $\alpha$, there exists an entrywise nonnegative matrix $R_{\xi,\alpha}$ such that  \(  I+\delta R_{\xi,\alpha}+o(\delta) \) is an update matrix for \( \cI -\ii\delta[\tau_\xi(K_\alpha),\,\cdot\,] \) as $\delta$ goes to $0$. These matrices can be chosen so that
        \begin{equation} \label{eq:all-time-family-bound}
        \left\|\sum_\alpha R_{\xi,\alpha}\right\|_{1\to1} \le 50e^{12\Lambda|\xi|}\opnorm{A_0}_{1\to1}.
        \end{equation}
    \end{teo}

    \begin{proof}
        We first construct the update matrix for one \(K=K_\alpha\), and then sum the construction. It is enough to treat \(\xi=t\ge0\).  The case \(\xi<0\) follows by replacing \(H\) by \(-H\), which leaves \(w_{i,j},M,\Lambda\) unchanged. Let \(\mathfrak U_t\) be the exact group from \cref{lem:reversible-plan-lift}, and let \(\mathfrak k_K\) be the lift
        from \cref{lem:local-kick-lift}. Define
        \begin{equation} \label{eq:conjugated-kick-lift}
        \mathfrak k_{K,t} :=\mathfrak U_t\mathfrak k_K\mathfrak U_{-t}.
        \end{equation}
        Each factor preserves \(\mathfrak P\), so
        \[
        \mathfrak k_{K,t}\mathfrak P\subseteq\mathfrak P.
        \]
         Moreover, using the two intertwining identities and \(\tau_t\cC_K\tau_{-t}=\cC_{K_t}\) with $K_t:=\tau_t(K)$, we obtain
        \[
        \Sigma\mathfrak k_{K,t} = \cC_{K_t}\Sigma.
        \]
         Thus \(I+\delta\mathfrak k_{K,t}\), applied to an input plan tuple,
        gives an exact output plan tuple for \(\cI+\delta\cC_{K_t}\). We observe that by \eqref{eq:qK-def} and \eqref{eq:qedge-def},
        \[
        \mathfrak H=\overline{\mathcal A}_H+q_H
        \]
         and
        \[
        \mathfrak k_K=\overline{\cC }_K+q_K,
        \]
         and we define
        \begin{equation} \label{eq:eKt-def}
        \mathfrak e_{K,t} :=\mathfrak U_t\overline{\cC }_K\mathfrak U_{-t}
         -\overline{\cC }_{K_t}.
        \end{equation}
        We derive the evolution equation one term at a time. Since \(\mathfrak U_t=e^{t\mathfrak H}\),
        \begin{align*}
        \frac d{dt} \left( \mathfrak U_t\overline{\cC }_K\mathfrak U_{-t} \right) = \mathfrak H\mathfrak U_t \overline{\cC }_K\mathfrak U_{-t} - \mathfrak U_t\overline{\cC }_K \mathfrak U_{-t}\mathfrak H = \left[  \mathfrak H,\, \mathfrak U_t\overline{\cC }_K\mathfrak U_{-t} \right].
        \end{align*}
        Moreover, \(K_t\) satisfies $\dot K_t=\mathcal A_H(K_t)$ and from the linearity of \(K\mapsto\overline{\cC }_K\),
        \[
        \frac d{dt}\overline{\cC }_{K_t} =\overline{\cC }_{\dot K_t}.
        \]

        Thus, for any operator \(X\), the Jacobi identity gives
        \begin{align*}
        \left[ \mathcal A_H,\cC_{K_t} \right](X) = [H,[K_t,X]]-[K_t,[H,X]] =[[H,K_t],X] = \cC_{\ii[H,K_t]}(X) = \cC_{\dot K_t}(X).
        \end{align*}
        Applying this identity on every diagonal block yields

        \[
        \frac d{dt}\overline{\cC }_{K_t} =[\overline{\mathcal A}_H, \overline{\cC }_{K_t}].
        \]

        Now we use \eqref{eq:qedge-def} and \eqref{eq:eKt-def}. Subtracting the two derivative formulas and expanding the commutator gives
        \begin{equation} \label{eq:eKt-ODE}
        \frac d{dt}\mathfrak e_{K,t} = \left[  \mathfrak H,\, \mathfrak U_t\overline{\cC }_K\mathfrak U_{-t} \right] - [\overline{\mathcal A}_H, \overline{\cC }_{K_t}] = [\mathfrak H,\mathfrak e_{K,t}]
         +[q_H,\overline{\cC }_{K_t}], \qquad \mathfrak e_{K,0}=0.
        \end{equation}
        We also prove the integral formula explicitly. Fix the final time \(t\) and, for \(0\le s\le t\), define
        \[
        F_t(s):= \mathfrak U_{t-s}\mathfrak e_{K,s}\mathfrak U_{-(t-s)}.
        \]
         Differentiating all three factors and using \eqref{eq:eKt-ODE} gives
        \begin{align*}
        \frac d{ds}F_t(s) &= -\mathfrak H\mathfrak U_{t-s} \mathfrak e_{K,s}\mathfrak U_{-(t-s)} + \mathfrak U_{t-s} \frac d{ds}\mathfrak e_{K,s} \mathfrak U_{-(t-s)} + \mathfrak U_{t-s}\mathfrak e_{K,s} \mathfrak U_{-(t-s)} \mathfrak H \\
        &= \mathfrak U_{t-s} \left( \frac d{ds}\mathfrak e_{K,s} -[\mathfrak H,\mathfrak e_{K,s}] \right) \mathfrak U_{-(t-s)}\\
        &= \mathfrak U_{t-s} [q_H,\overline{\cC }_{K_s}] \mathfrak U_{-(t-s)}.
        \end{align*}
        At the two endpoints we have \( F_t(t)=\mathfrak e_{K,t} \) and
        \[
        F_t(0)=\mathfrak U_t\mathfrak e_{K,0}\mathfrak U_{-t}=0.
        \]
         Integrating \(F_t'(s)\) from \(0\) to \(t\) proves
        \begin{equation} \label{eq:eKt-Duhamel}
        \mathfrak e_{K,t}
        =\int_0^t
        \mathfrak U_{t-s}
        [q_H,\overline{\cC }_{K_s}]
        \mathfrak U_{-(t-s)}\,ds.
        \end{equation}
        By \cref{lem:qedge-source} and the definition in \eqref{eq:S-map-def},
        \begin{equation} \label{eq:qH-source-matrix}
        |[q_H,\overline{\cC }_{K_s}]|_{\mathrm{bl}} \le20\mathcal S(\boldsymbol\eta(K_s)), \qquad \boldsymbol\eta(K_s):=(\eta_i(K_s))_i.
        \end{equation}
        Using \eqref{eq:block-majorant-rules}, the two bounds
        \[
        |\mathfrak U_{\pm(t-s)}|_{\mathrm{bl}}\le T_{t-s},
        \]
         and \eqref{eq:eKt-Duhamel}, we get the explicit entrywise estimate
        \begin{equation} \label{eq:eKt-block-bound}
        |\mathfrak e_{K,t}|_{\mathrm{bl}} \le 20\int_0^t T_{t-s}\mathcal S(\boldsymbol\eta(K_s))T_{t-s}\,ds.
        \end{equation}
        Finally, substitute

        \[
        \mathfrak k_K=\overline{\cC }_K+q_K
        \]
         into
        \eqref{eq:conjugated-kick-lift} and use
        \eqref{eq:eKt-def}:
        \begin{equation}
        \label{eq:kKt-three-parts}
        \mathfrak k_{K,t}
        =\overline{\cC }_{K_t}
         +\mathfrak e_{K,t}
         +\mathfrak U_tq_K\mathfrak U_{-t}.
        \end{equation}
        This identity is where the exact inverse \(\mathfrak U_{-t}\) is
        essential.

        Now, for each plan mark \(i\), the operator \(\mathbb E_i K_t\) is proportional to the identity on site \(i\). Hence the block-diagonal map
        \[
        (\mathfrak A_t^{\mathrm{away}}\boldsymbol X)_i :=\cC_{\mathbb E_i K_t}(X_i)
        \]
         preserves \(\mathfrak P\), and \(e^{\delta\mathfrak A_t^{\mathrm{away}}} \) preserves the trace norm of every component.  The difference between the uniform first term in \eqref{eq:kKt-three-parts} and this away-site generator is block diagonal, with block \(i\) equal to \(\cC_{K_t-\mathbb E_i K_t}\). By \eqref{eq:commutator-trace-norm}, its block norm is at most \(2\eta_i(K_t)\). Let
        \[
        \mathfrak r_{K,t}:= \mathfrak k_{K,t}-\mathfrak A_t^{\mathrm{away}}.
        \]
         Both terms on the right preserve \(\mathfrak P\), so
        \[
        \mathfrak r_{K,t}\mathfrak P\subseteq\mathfrak P.
        \]
         \eqref{eq:qK-bound}, \eqref{eq:lift-block-bound}, \eqref{eq:eKt-block-bound}, and \eqref{eq:kKt-three-parts} give
        \begin{equation} \label{eq:single-all-time-R-def}
        |\mathfrak r_{K,t}|_{\mathrm{bl}} \leq {}2\Diag(\boldsymbol\eta(K_t)) + 6\opnorm K\,T_tE_ST_t + 20\int_0^t T_{t-s}\mathcal S(\boldsymbol\eta(K_s))T_{t-s}\,ds =: R_{t,K},
        \end{equation}
        where $R_{t,K}$ is a nonnegative matrix. For an input plan tuple \(\boldsymbol X\), we consider the exact decomposition
        \begin{align} \label{eq:Euler-plan-decomposition}
        (I+\delta\mathfrak k_{K,t})\boldsymbol X ={} e^{\delta\mathfrak A_t^{\mathrm{away}}}\boldsymbol X +\delta\mathfrak r_{K,t}\boldsymbol X + \left(I+\delta\mathfrak A_t^{\mathrm{away}} - e^{\delta\mathfrak A_t^{\mathrm{away}}}\right) \boldsymbol X.
        \end{align}
        The first term is a plan tuple with exactly the original cost vector.
        The second is a plan tuple and, by \eqref{eq:block-majorant-cost}, \eqref{eq:single-all-time-R-def}, has cost at most \(\delta R_{t,K}\cost(\boldsymbol X)\). The last term is also a plan tuple. On the block indexed by \(i\), apply the exact integral remainder
        \[
        e^{\delta\mathcal A}-\cI-\delta\mathcal A = \int_0^\delta(\delta-s)\mathcal A^2e^{s\mathcal A}\,ds \] with
        \[
        \mathcal A := \cC_{\mathbb E_i K_t}.
        \]
         Since \(\|\mathcal A\|_{1\to1}\le2\|K\|\), this identity gives
        \[
        \left\|\cI+\delta\cC_{\mathbb E_i K_t} - e^{\delta\cC_{\mathbb E_i K_t}}\right\|_{1\to1} \le \int_0^\delta(\delta-s)\left(2\opnorm K\right)^2 e^{2s\opnorm K}\,ds \le2\delta^2\opnorm K^2e^{2\delta\opnorm K}.
        \]
        Thus, taking componentwise costs in \eqref{eq:Euler-plan-decomposition} gives the update matrix
        \[
        I+\delta R_{t,K}
        +2\delta^2\opnorm{K}^2 e^{2\delta\opnorm{K}}I
        \]
        for the coherent Euler update. For negative real time, the same construction with \(-H\) gives \eqref{eq:single-all-time-R-def} with the evolution \(K_s=\tau_{-s}(K)\).

        We now explicitly choose the matrices in the statement. For every real \(\xi\), put \(t=|\xi|\) and define
        \[
        \begin{aligned}
        R_{\xi,\alpha}:={}&
        2\Diag\!\left(\boldsymbol\eta\!\left(\tau_\xi(K_\alpha)\right)\right)
        +6\opnorm{K_\alpha}\,T_tE_{S_\alpha}T_t\\
        &+20\int_0^t T_{t-s}
        \mathcal S\!\left(
            \boldsymbol\eta\!\left(
                \tau_{\operatorname{sgn}(\xi)s}(K_\alpha)
            \right)
        \right)T_{t-s}\,ds.
        \end{aligned}
        \]
        The integral vanishes at zero real time. Every term is entrywise nonnegative, and the preceding argument shows that
        \[
        I+\delta R_{\xi,\alpha}
        +2\delta^2\opnorm{K_\alpha}^2
            e^{2\delta\opnorm{K_\alpha}}I
        \]
        is an update matrix for
        \[
        \cI-\ii\delta\left[\tau_\xi(K_\alpha),\,\cdot\,\right].
        \]
        This choice also gives the zero matrix when \(K_\alpha=0\). All subsequent estimates depend only on the absolute real time \(t\).

        Finally, we define
        \[
        b_0:=\sum_\alpha a_\alpha e_{S_\alpha}.
        \]
         Since the diagonal of \(E_S\) equals \( e_S\), \(b_0\) is
        the diagonal of \(A_0\). Thus,
        \begin{equation} \label{eq:b0-A0-bound}
        \opnorm{b_0}_\infty\le\opnorm{A_0}_{1\to1}.
        \end{equation}
        The local-oscillation estimate in \eqref{eq:eta-LR-bound} gives, entrywise,
        \begin{equation} \label{eq:summed-eta-bound}
        \sum_\alpha\boldsymbol\eta((K_\alpha)_{\pm s}) \le2e^{4sM}b_0.
        \end{equation}
        We now sum the three lines of \eqref{eq:single-all-time-R-def} separately. First,
        \begin{align} \label{eq:all-time-diag-sum}
        \left\|2\Diag\!\left( \sum_\alpha\boldsymbol\eta((K_\alpha)_{\pm t})\right) \right\|_{1\to1} \le 4\opnorm{e^{4tM}b_0}_\infty \le4e^{8\Lambda t} \opnorm{A_0}_{1\to1}.
        \end{align}
        Second, using \(a_\alpha\ge\|K_\alpha\|\),
        \begin{align} \label{eq:all-time-qK-sum}
        \left\|6T_t \left(\sum_\alpha\opnorm{K_\alpha}E_{S_\alpha}\right)
        T_t\right\|_{1\to1} \le 6\opnorm{T_t}_{1\to1}^2 \opnorm{A_0}_{1\to1} \le 6e^{12\Lambda t} \opnorm{A_0}_{1\to1}.
        \end{align}
        For the integral term, linearity and monotonicity of \(\mathcal S\),
        followed by \eqref{eq:summed-eta-bound}, give
        \[
        20\mathcal S\!\left( \sum_\alpha\boldsymbol\eta((K_\alpha)_{\pm s})\right) \le40\mathcal S(e^{4sM}b_0).
        \]
        \eqref{eq:S-map-norm}, \eqref{eq:b0-A0-bound} and the inequality
        \[
        \|e^{4sM}\|_{\infty\to\infty} \le e^{8\Lambda s}
        \]
         therefore imply
        \begin{align} \label{eq:all-time-e-integral-pre}
        \left\| 20\int_0^tT_{t-s} \mathcal S\!\left( \sum_\alpha\boldsymbol\eta((K_\alpha)_{\pm s})\right) T_{t-s}\,ds \right\|_{1\to1} \le 160\Lambda\opnorm{A_0}_{1\to1} \int_0^t e^{12\Lambda(t-s)}e^{8\Lambda s}\,ds.
        \end{align}
        If \(\Lambda>0\), the scalar integral is
        \[
        \int_0^t e^{12\Lambda(t-s)}e^{8\Lambda s}\,ds  \le\frac{e^{12\Lambda t}}{4\Lambda}.
        \]
         Substitution into \eqref{eq:all-time-e-integral-pre} gives
        \begin{equation} \label{eq:all-time-e-sum}
        \text{integral contribution} \le 40e^{12\Lambda t}\opnorm{A_0}_{1\to1}.
        \end{equation}
        If \(\Lambda=0\), every \(w_{i,j}=0\), so \(\mathcal S=0\) and the same
        inequality is immediate. Since
        \[
        e^{8\Lambda t}\le e^{12\Lambda t},
        \]
         adding \eqref{eq:all-time-diag-sum}, \eqref{eq:all-time-qK-sum}, and
        \eqref{eq:all-time-e-sum} gives
        \[
        \left\|\sum_\alpha R_{\xi,\alpha}\right\|_{1\to1} \le (4+6+40)e^{12\Lambda|\xi|} \opnorm{A_0}_{1\to1},
        \]
        which is \eqref{eq:all-time-family-bound}.

        For a countable family, apply the preceding estimates to finite partial sums of $\sum_\alpha R_{\xi,\alpha}$. These partial sums increase entrywise and their $1\to1$ norms are uniformly bounded by the right-hand side of \eqref{eq:all-time-family-bound}. Since $n$ is finite, the series converges in $1\to1$ norm, and the same bound holds for its sum. This completes the proof.
    \end{proof}

    \begin{remark}
        The comparison matrix in \eqref{eq:single-all-time-R-def}
        remains well defined when $\Lambda=0$. In this case,
        $M=0$, $T_t=I$, $\mathcal S=0$, and $\mathfrak U_t=I$,
        so the construction reduces to the comparison matrix
        for the original local kick.
    \end{remark}

    \subsubsection{From imaginary-time clusters to one nonnegative matrix}

    \begin{lemma} \label{lem:A0}
        Let
        \[
        \alpha=(i,P,k, (\bar i, \bar j)),
        \]
         where \(i\in[n]\), \(P\in\cP_i\), $k\geq 2$ and \((\bar i, \bar j)\) (with $\bar i <_v \bar j$) is a cluster of length \(k\) from \(i\). Define
        \begin{equation} \label{eq:alpha-data}
        a_\alpha:= \frac{\beta^k}{k!}\prod_{\ell=1}^kw_{i_\ell,j_\ell}, \qquad
        S_\alpha:=S_{\bar i, \bar j}, \qquad
        A_0:=\sum_\alpha a_\alpha E_{S_\alpha}.
        \end{equation}
        If \(\beta\Lambda<1\), then the series for \(A_0\) converges in \(1\to1\) norm and
        \begin{equation} \label{eq:A0-bound}
        \opnorm{A_0}_{1\to1} \leq 3\sum_{k\ge2}(k+1)^2 (\beta\Lambda)^k.
        \end{equation}
    \end{lemma}

    \begin{proof}
        All entries of \(A_0\) are nonnegative. The \(j\)-th column is
        \[
        A_0e_j = \sum_\alpha a_\alpha  e_{S_\alpha} \ind_{\{j\in S_\alpha\}}.
        \]
         Taking the sum of the entries of this column and applying \eqref{eq:marked-cluster} from \cref{lem:weighted-clusters} gives
        \begin{equation*} 
        \mathbf{1}^{\mathsf T}A_0e_j = \sum_\alpha a_\alpha
        |S_\alpha|\ind_{\{j\in S_\alpha\}} = \sum_{k\geq 2}  \sum_{i=1}^n  \sum_{P\in \cP_i} \sum_{\substack{ \bar i, \bar j \in [n]^k:  j \in S_{\bar i, \bar j}\\ (\bar i, \bar j) {\rm clust. \, from}\ i} } \frac{\beta^k}{k!} |S_{\bar i, \bar j}| \prod_{\ell=1}^kw_{i_\ell,j_\ell} \leq 3 \sum_{k\geq 2} (k+1)^2 (\beta\Lambda)^k.
        \end{equation*}
        Since $\beta\Lambda<1$, every column series converges. For a nonnegative matrix, the induced \(1\to1\) norm is the largest column sum. Hence, since the previous inequality holds for each $j \in [n]$ and $n$ is finite, $A_0$ converges in $1\to1$ norm with
        \[
        \opnorm{A_0}_{1\to1} =\max_{j\in [n]}\mathbf{1}^{\mathsf T}A_0e_j \leq 3 \sum_{k\geq 2} (k+1)^2 (\beta\Lambda)^k,
        \]
        which proves \eqref{eq:A0-bound}.
    \end{proof}

    \medskip
    \noindent
    \subsubsection{Adding first-order updates and integrating real time}

    \begin{lemma} \label{lem:first-order-sum}
        Let \(\mathcal K_q\) be finitely many trace-annihilating, Hermiticity-preserving generators and $R_q$ finitely many nonnegative matrices. If \(I+\delta R_q+o(\delta)\) is an update matrix for \(\cI+\delta\mathcal K_q\), then
        \[
        I+\delta\sum_qR_q+o(\delta)
        \]
        is an update matrix for \(\cI+\delta\sum_q\mathcal K_q\). The same conclusion holds for a countable family, provided that both series $\sum_q\mathcal K_q$ and $\sum_qR_q$ converge in their respective $1\to1$ norms.

        Furthermore, let $(\mathcal K(\omega),R(\omega))$ be a Bochner-measurable family on a finite measure space $(\Omega,\mu)$. Assume that, for almost every $\omega$, the map $\mathcal K(\omega)$ is Hermiticity-preserving and trace-annihilating, $R(\omega)$ is entrywise  nonnegative, and $I+\delta R(\omega)+o(\delta)$ is an update matrix for $\cI +\delta\mathcal K(\omega)$. If
        \[
        \int_\Omega \left( \opnorm{\mathcal K(\omega)}_{1\to1} +\opnorm{R(\omega)}_{1\to1} \right)d\mu(\omega)<\infty,
        \]
        then
        \[
            I+\delta\int_\Omega R(\omega)\,d\mu(\omega)+o(\delta)
        \]
        is an update matrix for
        \[
            \cI+\delta\int_\Omega\mathcal K(\omega)\,d\mu(\omega).
        \]
    \end{lemma}

    \begin{proof}
        \proofstep{Finite sums.}
        Suppose first that there are \(m\) generators.  The exact identity
        \begin{equation} \label{eq:first-order-average-identity}
        \cI+\delta\sum_{q=1}^m\mathcal K_q = \frac1m\sum_{q=1}^m \left(\cI+m\delta\mathcal K_q\right)
        \end{equation}
        expresses the desired map as an average. For fixed \(m\), replacing \(\delta\) by \(m\delta\) in the assumed comparison gives the update matrix \( I+m\delta R_q+o(\delta) \) for the \(q\)-th map on the right of \eqref{eq:first-order-average-identity}. Their average is
        \[
        \frac1m\sum_{q=1}^m \left(I+m\delta R_q+o(\delta)\right) = I+\delta\sum_{q=1}^mR_q+o(\delta).
        \]
        By \cref{prop on lin}, it is an update matrix for \eqref{eq:first-order-average-identity}.

        \proofstep{Countable sums.}
        Let
        \[
        \mathcal K^{(m)}:=\sum_{q\le m}\mathcal K_q, \qquad R^{(m)}:=\sum_{q\le m}R_q,
        \]
        and let \(\mathcal K\) and \(R\) be their norm limits. The finite case gives, for each fixed \(m\), an update matrix
        \( I+\delta R^{(m)}+o_m(\delta) \) for \(\cI+\delta\mathcal K^{(m)}\). Set
        \[
        \varepsilon_m:=\|\mathcal K-\mathcal K^{(m)}\|_{1\to1}.
        \]
         If an input plan component \(X_j\) has cost
        \(x_j=\trnorm{X_j}/2\), then
        \[
        \trnorm{ \delta(\mathcal K-\mathcal K^{(m)})(X_j)}  \le 2\delta\varepsilon_mx_j.
        \]
        The second inequality in \cref{Lemma W1 rel tr} gives a transport plan for this tail with total cost at most \(n\delta\varepsilon_mx_j\). Hence the same tail is bounded by the nonnegative update matrix
        \[
        \delta U_m, \qquad U_m := n\varepsilon_m \mathbf{1}\mathbf{1}^{\mathsf T}, \qquad \|U_m\|_{1\to1}= n^2\varepsilon_m\xrightarrow{m\to\infty}0.
        \]
        Thus the finite comparison and the tail correction give an update matrix for the full generator. Since the matrices \(R_q\) are nonnegative, we may enlarge the partial-sum matrix to \(R\). We may also replace each remainder by its entrywise absolute value. Write
        \[
        E_m(\delta):=\left|o_m(\delta)\right|,
        \qquad
        \lim_{\delta\downarrow0}
            \frac{\opnorm{E_m(\delta)}_{1\to1}}{\delta}=0.
        \]
        For every fixed \(m\) and all sufficiently small positive \(\delta\), an update matrix for the full Euler map is therefore
        \[
        I+\delta R+E_m(\delta)+\delta U_m.
        \]

        The diagonal choice can be made as follows. Initialize
        \[
        m_0=h_0=0,
        \]
        and, for each positive integer \(k\), choose
        \[
        m_k:=\min\left\{
            m>m_{k-1}:\opnorm{U_m}_{1\to1}\le2^{-k}
        \right\}.
        \]
        This minimum exists because the tail matrices tend to zero in norm. Let \(h_k\) be the least integer strictly larger than \(h_{k-1}\) for which the finite comparison at index \(m_k\) is valid and
        \[
        \opnorm{E_{m_k}(\delta)}_{1\to1}
            \le2^{-k}\delta
        \qquad\text{for every }0<\delta\le2^{-h_k}.
        \]
        Such an integer exists by the fixed-index remainder estimate. Set
        \[
        \delta_k:=2^{-h_k},
        \qquad
        k(\delta):=\max\left\{k\ge1:\delta\le\delta_k\right\},
        \qquad
        m(\delta):=m_{k(\delta)}
        \]
        for positive \(\delta\) at most \(\delta_1\). The maximum is finite because the thresholds decrease strictly to zero, and its value tends to infinity as \(\delta\) tends to zero. The nonnegative remainder
        \[
        F_\delta:=E_{m(\delta)}(\delta)+\delta U_{m(\delta)}
        \]
        satisfies
        \[
        \frac{\opnorm{F_\delta}_{1\to1}}{\delta}
        \le2^{1-k(\delta)}
        \xrightarrow[\delta\downarrow0]{}0.
        \]
        Consequently, the update matrices just constructed have the required form with a remainder independent of the input operator and its transport plan. This proves the countable statement.

        \proofstep{Bochner integrals.}
        If $\mu(\Omega)=0$, both integrals vanish and the claim is immediate. Assume henceforth that $\mu(\Omega)>0$.
        Since $\mu$ is finite and the pair $(\mathcal K(\omega),R(\omega))$ is Bochner integrable, we can approximate it in $L^1(\mu)$ by simple (i.e., finite-valued measurable) pairs
        \[
        \left(\mathcal K_m(\omega),R_m(\omega)\right)_{m \in \mathbb{N}}
        \]
         whose values are chosen from the range of $(\mathcal K,R)$ outside the exceptional null set. This ensures that each sampled $R_m$ is nonnegative and that $I +\delta R_m + o_m(\delta)$ is an update matrix for $\cI +\delta\mathcal K_m$.

        The integral of each simple pair is a finite positive linear combination. Positive coefficients are absorbed by rescaling \(\delta\), and the finite-sum argument gives update matrix \(I +\delta\int_\Omega R_m\,d\mu + o(\delta)\) for \(\cI + \delta\int_\Omega\mathcal K_m\,d\mu\). The generator error is controlled exactly as in the countable case, with
        \[
        \varepsilon_m := \left\| \int_\Omega(\mathcal K-\mathcal K_m)\,d\mu \right\|_{1\to1} \le \int_\Omega\|\mathcal K-\mathcal K_m\|_{1\to1}\,d\mu \xrightarrow{m\to\infty}0.
        \]
        Moreover,
        \[
            \left\|\int_\Omega(R-R_m)\,d\mu\right\|_{1\to1}\longrightarrow0.
        \]
        Since enlarging an update matrix entrywise preserves its validity, we obtain the update matrix
        \[
        I+\delta\int_\Omega R\,d\mu + o_m(\delta)+\delta U_m + \delta\left|\int_\Omega(R_m-R)\,d\mu\right|,
        \]
        where $\delta U_m$ controls the generator approximation error and $\|U_m\|_{1\to1}\to0$, as in the countable case.
        For numerical matrices, entrywise absolute values preserve the induced column-sum norm. Define
        \[
        W_m:=U_m+
            \left|\int_\Omega\left(R_m-R\right)\,d\mu\right|.
        \]
        These matrices are nonnegative and tend to zero in norm. Let \(E_m(\delta)\) now denote the entrywise absolute value of the remainder for the integrated simple pair. Explicitly, repeat the preceding construction with
        \[
        m_0=h_0=0,
        \qquad
        m_k:=\min\left\{
            m>m_{k-1}:\opnorm{W_m}_{1\to1}\le2^{-k}
        \right\}.
        \]
        At each step, let \(h_k\) be the least integer larger than \(h_{k-1}\) for which this integrated finite comparison is valid and
        \[
        \opnorm{E_{m_k}(\delta)}_{1\to1}
            \le2^{-k}\delta
        \qquad\text{for every }0<\delta\le2^{-h_k}.
        \]
        Define the thresholds and the indices \(k(\delta)\) and \(m(\delta)\) by the same displayed formulas as in the countable case. Then
        \[
        I+\delta\int_\Omega R\,d\mu
          +E_{m(\delta)}(\delta)+\delta W_{m(\delta)}
        \]
        is an update matrix for the desired integrated Euler map, and
        \[
        \frac{\opnorm{E_{m(\delta)}(\delta)
                      +\delta W_{m(\delta)}}_{1\to1}}{\delta}
        \le2^{1-k(\delta)}
        \xrightarrow[\delta\downarrow0]{}0.
        \]
        This supplies the claimed nonnegative, input-independent remainder and proves the Bochner-integral statement.

    \end{proof}

    \begin{lemma}[The Fourier-kernel integral] \label{lem:coherent-kernel}
        If \(a\ge0\), \(\beta>0\), and \(a\beta<2\pi\), then
        \begin{equation} \label{eq:kernel-integral}
        \int_{\mathbb R} \frac{e^{a|\xi|}} {\beta\cosh(2\pi\xi/\beta)}\,d\xi \le \frac4{2\pi-a\beta}.
        \end{equation}
    \end{lemma}

    \begin{proof}
        By definition,
        \[
        \cosh (y) \ge\frac{e^y}{2},
        \]
         for every  \(y\ge0\). Thus, since the integrand is even,
        \begin{align*}
        \int_{\mathbb R} \frac{e^{a|\xi|}} {\beta\cosh(2\pi\xi/\beta)}\,d\xi = \frac2\beta\int_0^\infty \frac{e^{a\xi}}{\cosh(2\pi\xi/\beta)}\,d\xi \le  \frac4\beta\int_0^\infty e^{-\left(2 \frac{\pi}{\beta}-a\right)\xi}\,d\xi = \frac4{2\pi-a\beta}.
        \end{align*}
        The integral converges because the hypothesis implies \(2\pi/\beta-a>0\).
    \end{proof}

    Integration of real-time estimates against a hyperbolic-secant kernel has a parallel in \cite[Lemma~5.4]{Bergamaschi26}. Here the Fourier representation comes from \cite[Lemma~4.4]{BLMT25}, and we integrate transport-plan update matrices.

    \begin{proof}[Proof of \cref{coherent evolution}]
    For each $i\in[n]$, let
        \[
            \mathscr A_i :=
            \left\{
                \alpha=(i,P,k,(\bar i,\bar j)):
                \begin{array}{l}
                    P\in\cP_i,\quad k\ge2,\\
                    \bar i,\bar j\in[n]^k,\quad \bar i<_v\bar j,\\
                    (\bar i,\bar j)\text{ is a cluster from }i
                \end{array}
            \right\}.
        \]
        For
        \[
        \alpha=(i,P, k,(\bar i,\bar j))\in\mathscr A_i,
        \]
         we use the abbreviation \( B_\alpha:=B_{\bar i,\bar j}^P \) and, for $\xi\in\mathbb R$, we introduce the elementary coherent generator
        \[
        \cL_{\xi,\alpha} :=
            -\ii[\tau_\xi(B_\alpha),\,\cdot\,].
        \]
         Denote by
            \[
            \mathcal K_i^{\mathrm c} = \sum_{P\in\cP_i}-\ii[G^P,\,\cdot\,]
            \]
             the coherent generator, so that
            \[
            \Phi_{i,\mathrm c} = \cI+\delta\mathcal K_i^{\mathrm c}.
            \]
             Moreover \cref{lem:Fourier-transform-coherent-term} gives
        \begin{equation}  \label{eq:coherent-fourier-generator}
            -\ii[G^P,\,\cdot\,] = \int_{\mathbb R} f_\beta(\xi) \left(-\ii[\tau_\xi(B^P),\,\cdot\,]\right) \,d\xi, \qquad  f_\beta(\omega) = \frac{1}{\beta\cosh\left(\frac{2\pi}{\beta}\omega\right) }.
        \end{equation}
        On the other hand, \cref{lem:B-cluster-cancellation} gives the cluster expansion
        \[
        B^P = \sum_{\alpha\in\mathscr A_{i,P}}B_\alpha,
        \]
         where $\mathscr A_{i,P}\subset\mathscr A_i$ denotes the subset of indices with fixed $P$.

        We now justify substituting this expansion into \eqref{eq:coherent-fourier-generator}. The estimates on $\| B_\alpha\|$ obtained in \cref{lem:B-cluster-cancellation} and the rooted-cluster estimate in \eqref{eq:rooted-cluster} imply
        \begin{equation} \label{eq:cluster-B-absolute-convergence}
            \sum_{\alpha\in\mathscr A_{i,P}}\|B_\alpha\| \le \sum_{k\geq2}(\beta\Lambda)^k <\infty,
        \end{equation}
        where the last inequality follows from $\beta\Lambda<1$. Consequently, the cluster expansion of $B^P$ is absolutely convergent in operator norm. Moreover, real-time conjugation preserves the operator norm
        \[
        \|\tau_\xi(B_\alpha)\|=\|B_\alpha\|
        \]
         and thus
        \begin{equation} \label{eq:elementary-coherent-generator-bound}
            \|\cL_{\xi,\alpha}\|_{1\to1} \leq2\|B_\alpha\|.
        \end{equation}
        The Fourier kernel is nonnegative and satisfies
        \[
        \int_{\mathbb R}f_\beta(\xi)\,d\xi=\frac12.
        \]
         Therefore,
        \begin{align} \label{eq:generator-cluster-integrability}
            \int_{\mathbb R} f_\beta(\xi) \sum_{\alpha\in\mathscr A_{i,P}} \|\cL_{\xi,\alpha}\|_{1\to1} \,d\xi \leq 2 \int_{\mathbb R}f_\beta(\xi)\,d\xi \sum_{\alpha\in\mathscr A_{i,P}}\|B_\alpha\| = \sum_{\alpha\in\mathscr A_{i,P}}\|B_\alpha\| <\infty.
        \end{align}
        Thus the family \( (\xi,\alpha) \longmapsto f_\beta(\xi)\cL_{\xi,\alpha} \) is absolutely Bochner integrable. In particular, we may interchange the cluster sum and the Fourier integral. Hence
        \[
        -\ii[G^P,\,\cdot\,] = \int_{\mathbb R} f_\beta(\xi) \sum_{\alpha\in\mathscr A_{i,P}} \cL_{\xi,\alpha} \,d\xi.
        \]
         Summing over the three choices of $P$, we obtain the fundamental decomposition
        \begin{equation}
            \label{eq:coherent-elementary-decomposition}
            \mathcal K_i^{\mathrm c} = \int_{\mathbb R} f_\beta(\xi) \sum_{\alpha\in\mathscr A_i} \cL_{\xi,\alpha} \,d\xi.
        \end{equation}
        Now, fix $\xi\in\mathbb R$ and $\alpha\in\mathscr A_i$ such that $S_\alpha \ne \varnothing$. By \cref{lem:B-cluster-cancellation}, the operator $B_\alpha$ is Hermitian, is supported on $S_\alpha$, and satisfies \( \|B_\alpha\|\leq a_\alpha \), where $a_\alpha$ is defined in \eqref{eq:alpha-data}. Therefore, \cref{lem:two-mark-input} provides an entrywise nonnegative matrix $R_{\xi,\alpha}$ such that \( I+\delta R_{\xi,\alpha}+o_{\xi,\alpha}(\delta) \) is an update matrix for \( \cI+\delta\cL_{\xi,\alpha} \). If $S_\alpha = \varnothing$ then $B_\alpha = 0$ and we choose $R_{\xi,\alpha}=0$. For fixed $\xi$ and $i$, define
        \[
            \cL_{\xi,i} := \sum_{\alpha\in\mathscr A_i}\cL_{\xi,\alpha},
            \qquad R_{\xi,i} := \sum_{\alpha\in\mathscr A_i}R_{\xi,\alpha}.
        \]
        We also set
        \[
        R_\xi := \sum_{i=1}^n R_{\xi,i} = \sum_\alpha R_{\xi,\alpha}.
        \]
         From \eqref{eq:cluster-B-absolute-convergence} and \eqref{eq:elementary-coherent-generator-bound}, the sum defining $\cL_{\xi,i}$ is absolutely convergent in $1\to1$ operator norm. The construction in \cref{lem:two-mark-input}, with $A_0$ as in \cref{lem:A0}, gives
        \begin{equation} \label{eq:all-time-cluster-family}
            \|R_\xi\|_{1\to1} = \left\|\sum_\alpha R_{\xi,\alpha}\right\|_{1\to1} \leq 50e^{12\Lambda|\xi|} \|A_0\|_{1\to1}.
        \end{equation}
        Since the terms of the convergent series defining $R_\xi$ are entrywise nonnegative, the subseries defining $R_{\xi,i}$ also converges in $1\to1$ norm. Together with the norm convergence of $\cL_{\xi,i}$
        established above, this allows us to apply the countable-sum part of \cref{lem:first-order-sum}. Hence, \( I+\delta R_{\xi,i}+o_{\xi,i}(\delta) \) is an update matrix for $\cI +\delta\cL_{\xi,i}$.

        We now integrate the fixed-time generators and the corresponding matrices $R_{\xi,i}$. It is convenient to introduce the finite positive measure
        \[
        d\nu_\beta(\xi):=f_\beta(\xi)\,d\xi.
        \]
         In this notation,
        \[
        \nu_\beta(\mathbb R)=\frac12
        \]
         and \eqref{eq:coherent-elementary-decomposition} becomes
        \[
        \mathcal K_i^{\mathrm c} = \int_{\mathbb R}\cL_{\xi,i}\,d\nu_\beta(\xi).
        \]
         To apply the Bochner-integral part of \cref{lem:first-order-sum}, we verify measurability and integrability.

        First, $\xi\mapsto\cL_{\xi,\alpha}$ is continuous in $1\to1$ operator norm, because \( \xi\longmapsto\tau_\xi(B_\alpha) \) is continuous in operator norm. Since $\cL_{\xi,i}$ is the norm limit of its partial sums, it is Borel measurable.

        Next, the explicit choice of $R_{\xi,\alpha}$ in the proof of \cref{lem:two-mark-input}, based on \eqref{eq:single-all-time-R-def}, is built out of matrix exponentials, $\tau_\xi(B_\alpha)$, and operator norms. All these quantities depend continuously on $\xi$. The positive-time and
        negative-time formulas agree at $\xi=0$, so \( \xi\longmapsto R_{\xi,\alpha} \) is Borel measurable on all of $\mathbb R$. Since $R_{\xi,i}$ is the norm limit of finite sums of such matrices, it is also Borel measurable.

        The generator family is integrable because
        \begin{align*}
            \int_{\mathbb R} \|\cL_{\xi,i}\|_{1\to1}\,d\nu_\beta(\xi) \leq \int_{\mathbb R} \sum_{\alpha\in\mathscr A_i}
            \|\cL_{\xi,\alpha}\|_{1\to1} \,d\nu_\beta(\xi) \leq 2\nu_\beta(\mathbb R) \sum_{\alpha\in\mathscr A_i}\|B_\alpha\| = \sum_{\alpha\in\mathscr A_i}\|B_\alpha\| <\infty.
        \end{align*}
        For the matrix family,
        \[
        \|R_{\xi,i}\|_{1\to1} \leq \|R_\xi\|_{1\to1}
        \]
         and hence, using \eqref{eq:all-time-cluster-family}, we obtain
        \begin{align} \label{eq:R-xi-integrability}
            \int_{\mathbb R} \|R_{\xi,i}\|_{1\to1}\,d\nu_\beta(\xi) \leq 50\|A_0\|_{1\to1} \int_{\mathbb R} f_\beta(\xi)e^{12\Lambda|\xi|} \,d\xi.
        \end{align}
        The right-hand side is finite. Indeed, $\|A_0\|_{1\to1}<\infty$ follows from $\beta\Lambda<1$, whereas \cref{lem:coherent-kernel} gives
        \[
        \int_{\mathbb R} f_\beta(\xi)e^{12\Lambda|\xi|} \,d\xi<\infty
        \]
         under the assumption \( 12\beta\Lambda<2\pi \). All the assumptions of the Bochner-integral part of \cref{lem:first-order-sum} are therefore satisfied. Define
        \begin{equation} \label{eq:Ri-coherent-definition}
            R_{\mathrm c}^{(i)} := \int_{\mathbb R} R_{\xi,i}\,d\nu_\beta(\xi) = \int_{\mathbb R} f_\beta(\xi)R_{\xi,i}\,d\xi.
        \end{equation}
        Since each $R_{\xi,i}$ is entrywise nonnegative and $f_\beta(\xi)\geq0$, the matrix $R_{\mathrm c}^{(i)}$ is entrywise
        nonnegative. Furthermore, \cref{lem:first-order-sum} shows that \( I+\delta R_{\mathrm c}^{(i)}+o(\delta) \) is an update matrix for
        \[
        \cI+\delta\mathcal K_i^{\mathrm c} = \Phi_{i,\mathrm c}.
        \]
         Since the sum over $i$ is finite, we may exchange it with the integral in \eqref{eq:Ri-coherent-definition}, obtaining
        \begin{equation} \label{eq:sum-Ri-coherent}
            \sum_{i=1}^n R_{\mathrm c}^{(i)} = \int_{\mathbb R} f_\beta(\xi) \sum_{i=1}^n R_{\xi,i}
            \,d\xi = \int_{\mathbb R} f_\beta(\xi)R_\xi\,d\xi.
        \end{equation}

        Taking the $1\to1$ norm in \eqref{eq:sum-Ri-coherent}, and then applying \eqref{eq:all-time-cluster-family}, we find
        \begin{align} \label{eq:coherent-before-final-bounds}
            \left\|\sum_{i=1}^n R_{\mathrm c}^{(i)}\right\|_{1\to1} \leq \int_{\mathbb R} f_\beta(\xi)\|R_\xi\|_{1\to1}\,d\xi \leq 50\|A_0\|_{1\to1} \int_{\mathbb R} f_\beta(\xi)e^{12\Lambda|\xi|} \,d\xi.
        \end{align}
        By \cref{lem:coherent-kernel}, with $a=12\Lambda$, and  \eqref{eq:A0-bound}, we obtain
        \begin{align*}
            \left\|\sum_{i=1}^n R_{\mathrm c}^{(i)}\right\|_{1\to1} \leq 50 \left( 3\sum_{k\geq2}(k+1)^2(\beta\Lambda)^k
            \right) \frac{4}{2\pi-12\beta\Lambda} = \frac{600}{2\pi-12\beta\Lambda} \sum_{k\geq2}(k+1)^2(\beta\Lambda)^k.
        \end{align*}

        This proves \eqref{eq:coherent-column-bound-all-time}.
    \end{proof}

    \section{Rapid mixing for continuous dynamics}

    Combining the estimates established in \cref{Evolution of Wasserstein transport plans}, namely \cref{dissipative evolution} and \cref{coherent evolution}, with \cref{Dobrushin condition implies rapid mixing}, we establish rapid mixing for the continuous-time Lindblad evolution described by the quantum Markov semigroup in \eqref{semigroup}. This holds in a high-temperature regime that depends on the local interaction strength $\Lambda$. To establish this result, we show that above a certain temperature, the first-order approximation associated with the Lindbladian satisfies the quantum Dobrushin condition we introduced in \cref{Q Dobrushin Condition}.

    First, we show that the Dobrushin condition holds for the Euler map
    \[
        \cI+\frac{\delta}{n}\cL.
    \]

    \begin{prop}[Dobrushin condition at high temperature]
    \label{lem:high-temperature-dobrushin}
    Let $H$ be a Hamiltonian as in \cref{Hamiltonian} and let
    $\sigma$ be its Gibbs state at inverse temperature $\beta$ satisfying
    \begin{equation}
    \beta \Lambda <  \frac{1}{100}.
    \end{equation}
    Then, for the Lindbladian $\cL $ defined in \cref{Lindbladian} and for a sufficiently small $\delta>0$, the map
    \begin{equation}
    \Phi = \cI  + \frac{\delta}{n} \cL
    \end{equation}
    satisfies the quantum Dobrushin condition in \eqref{Dobrushin condition} with parameter $\gamma=3\delta$.
    \end{prop}

    The proof of \cref{lem:high-temperature-dobrushin} combines the preceding bounds. Taking the local generator $\cL_i$, \cref{dissipative evolution} assigns to the dissipative part an exact local contraction \(-4E_{\{i\}}\) plus a nonnegative
    error \(R_{\mathrm d}^{(i)}\), while \cref{coherent evolution} assigns to the coherent part a nonnegative error \(R_{\mathrm c}^{(i)}\). The weighted cluster estimates bound the column sums of both error families. At \(\beta\Lambda<1/100\), their combined column sum is smaller than $1$, so at least $3$ units remain from the four-unit contractive core.

    \begin{proof}
        Let $i \in [n]$, $\delta >0$ sufficiently small, and recall the definition of $\Phi_i$ in \eqref{def of Phi}. We consider the following identity
        \begin{equation} \label{decomp of Phi i}
        n\Phi_i = \frac12(\cI+2\delta\mathcal K_i^{\mathrm d}) + \frac12(\cI+2\delta\mathcal K_i^{\mathrm c}),
        \end{equation}
        where
        \[
        \mathcal K_i^{\mathrm d} = \sum_{P\in\cP_i}\left(  A^P \cdot (A^P)^\dagger - \frac{1}{2}\{(A^P)^\dagger A^P,\cdot\}\right)
        \]
         and
        \[
        \mathcal K_i^{\mathrm c} = \sum_{P\in\cP_i}\left(  -\ii [G^P,\cdot] \right).
        \]

        Now we apply \cref{dissipative evolution} to the first map on the right-hand side of \eqref{decomp of Phi i} and \cref{coherent evolution} to the second. With step size \(2\delta\), their update matrices are
        \[
        I-8\delta E_{\{i\}}+2\delta R_{\mathrm d}^{(i)} +\mathcal O(\delta^2)
        \]
        and
        \[
        I+2\delta R_{\mathrm c}^{(i)}+o(\delta),
        \]
        respectively. By \cref{prop on lin}, averaging two maps with weights $1/2,1/2$ allows us to combine their corresponding output transport plans with the same weights. The resulting cost vector is bounded componentwise by the average of the two cost vectors. Hence we have obtained the following update matrix for $n\Phi_i$
        \begin{equation} \label{eq:main-local-Q}
            Q_i(\delta) = I-4\delta E_{\{i\}} + \delta R_{\mathrm d}^{(i)} + \delta R_{\mathrm c}^{(i)} + o(\delta).
        \end{equation}
        The remainder is nonnegative and has \(1\to1\) norm \(o(\delta)\); the \(\mathcal O(\delta^2)\) dissipative remainder is absorbed into it.

        By \eqref{def of Phi},
        \[
        \Phi = \cI + \frac{\delta}{n}\cL=\sum_i\Phi_i.
        \]
         Thus we average the \(n\) matrices in \eqref{eq:main-local-Q} and use the linearity of transport plans (\cref{prop on lin}). Since
        \[
            \sum_{i=1}^nE_{\{i\}}=I,
        \]
        the averaged update matrix is
        \begin{equation} \label{eq:main-averaged-Q}
        Q_\delta = I+\frac{\delta}{n} \left( -4I+\sum_{i=1}^n R_{\mathrm d}^{(i)}+\sum_{i=1}^n R_{\mathrm c}^{(i)} \right) + o\left(\frac \delta n \right).
        \end{equation}
        Here \(n\) is fixed while \(\delta\) goes to \(0\). A finite average of little-\(o(\delta)\) matrices is again \(o(\delta)\), equivalently \( o\left(\frac \delta n \right)\) for fixed \(n\). We recall that \cref{dissipative evolution} and \cref{coherent evolution} give
        \begin{align}
        \left\|\sum_iR_{\mathrm d}^{(i)}\right\|_{1\to1} &\le 12\sum_{t\ge1}(t+1)^2(\beta \Lambda)^t,
        \label{eq:main-d-error}\\
        \left\|\sum_iR_{\mathrm c}^{(i)}\right\|_{1\to1} &\le \frac{600} {2\pi-12\beta \Lambda} \sum_{t\ge2}(t+1)^2(\beta \Lambda)^t.
        \label{eq:main-c-error}
        \end{align}
        We set $x:=\beta\Lambda$. For $0\le x<1$, differentiating

        \[
        \sum_{t\ge0}(t+1)x^{t+1}=\frac{x}{(1-x)^2}
        \]
         gives

        \[
        \sum_{t\ge0}(t+1)^2x^t=\frac{1+x}{(1-x)^3}.
        \]
         Subtracting the term $t=0$
        (equal to $1$) and then the term $t=1$ (equal to $4x$) yields
        \[
        \sum_{t\ge1}(t+1)^2x^t=\frac{x(4-3x+x^2)}{(1-x)^3},
        \qquad
        \sum_{t\ge2}(t+1)^2x^t=\frac{x^2(9-11x+4x^2)}{(1-x)^3}.
        \]
        Suppose \(0\le x<1/100\); then both update-matrix theorems are applicable. Moreover, for $x$ in this range
        \[
        \frac1{(1-x)^3} < \frac{11}{10}
        \]

        and \(4-3x+x^2\le4\). Therefore the dissipative error satisfies
        \begin{align} \label{eq:main-d-error-explicit}
        \left\|\sum_iR_{\mathrm d}^{(i)}\right\|_{1\to1} \le 12\frac{x(4-3x+x^2)}{(1-x)^3}  \le \frac{66}{125}.
        \end{align}
        For the coherent denominator,
        \[
        2\pi-12x > 6-\frac3{25} = \frac{147}{25} \Rightarrow \frac{600}{2\pi-12x} < 103.
        \]
        Similarly, \(9-11x+4x^2\le9\) in this interval. Therefore the coherent error satisfies
        \begin{align}
        \left\|\sum_iR_{\mathrm c}^{(i)}\right\|_{1\to1} \le \frac{600}{2\pi-12x} \frac{x^2(9-11x+4x^2)}{(1-x)^3}  < \frac{13}{125}.
        \label{eq:main-c-error-explicit}
        \end{align}
        The triangle inequality and \eqref{eq:main-d-error-explicit}, \eqref{eq:main-c-error-explicit} give
        \begin{equation} \label{eq:main-total-error}
        \left\| \sum_iR_{\mathrm d}^{(i)}+\sum_iR_{\mathrm c}^{(i)} \right\|_{1\to1} < \frac{66}{125}+\frac{13}{125} = \frac{79}{125} =: c < 1.
        \end{equation}
        Finally, for a nonnegative matrix, its \(1\to1\) norm is its largest column sum. Using \eqref{eq:main-averaged-Q}, the \(j\)-th column sum is at most
        \[
        1-\frac{4\delta}{n} + \frac{\delta}{n} \mathbf{1}^{\mathsf T} \left(\sum_iR_{\mathrm d}^{(i)}+\sum_iR_{\mathrm c}^{(i)}\right)e_j + o\left(\frac \delta n \right).
        \]
        By \eqref{eq:main-total-error}, the middle column sum is at most \(c\). Taking the maximum over \(j\) gives the following bound for the norm of the update matrix
        \[
        \opnorm{Q_\delta}_{1\to1} \le 1-\frac{(4-c)\delta}{n}+ o\left(\frac \delta n \right).
        \]
        Since $Q_i(\delta)$ is an update matrix for $n\Phi_i$, the matrix $Q_i(\delta)/n$ is an update matrix for $\Phi_i$. Applying \cref{relating dobrushin matrix to update matrix} to the decomposition
        \[
        \Phi=\sum_i\Phi_i
        \]
         therefore gives
        \[
        \|D^{(\Phi)}\|_{1\to1} \leq \left\|\frac1n\sum_{i=1}^n Q_i(\delta)\right\|_{1\to1} = \|Q_\delta\|_{1\to1}.
        \]
        Since $c<1$, the preceding estimate implies, for sufficiently small $\delta>0$,
        \[
        \|D^{(\Phi)}\|_{1\to1}
        \leq 1-\frac{3\delta}{n},
        \]
        which is the required Dobrushin condition.
    \end{proof}

    With this result, we can now prove rapid mixing for continuous dynamics.
    \begin{teo}[Rapid mixing for high temperature Gibbs states] \label{satisfaction of Dobrushin condition}
        Let $H$ be a Hamiltonian as in \cref{Hamiltonian} and let $\sigma$ be its Gibbs state at inverse temperature $\beta$ such that
        \[
        \beta \Lambda < \frac{1}{100}.
        \]
        Then, for every initial state $\rho_0$ and every
        $0<\varepsilon<1$,  the Lindbladian evolution in \cref{Lindbladian} satisfies
        \[
        \|e^{t \cL}\rho_0-\sigma\|_{\mathrm{tr}}\leq\varepsilon, \qquad \mathrm{for \ all \ } t \geq \frac{1}{3}\log \left( \frac{2n}{\varepsilon}\right).
        \]
    \end{teo}

    \begin{proof}
        By \cref{lem:high-temperature-dobrushin}, for all sufficiently small $\delta>0$, the map
        \[
        \Phi_\delta := \cI + \frac{\delta}{n} \cL
        \]
         satisfies the Dobrushin condition with $\gamma = 3 \delta$. Moreover, \cref{KMS detailed balance} gives $\Phi_\delta(\sigma)=\sigma$. Hence, by \cref{Dobrushin condition implies rapid mixing}, $\sigma$ is its unique fixed point.

        We fix an initial state $\rho_0$, an error parameter $0 <\varepsilon < 1$ and a time $t$ satisfying
        \begin{equation} \label{bound for time}
        t \geq \frac{1}{3}\log\left(\frac{2n}{\varepsilon} \right).
        \end{equation}
        For each integer $m\geq 1$, we define
        \[
        \delta_m = \frac{nt}{m}
        \]
         and correspondingly $\gamma_m = 3\delta_m$. Since $\delta_m\to0$ as $m \to \infty$, the preceding conclusions apply to $\Phi_{\delta_m}$ for all sufficiently large $m$. Furthermore, \eqref{bound for time} implies
        \[m \geq \frac{m}{3t} \log\left(\frac{2n}{\varepsilon} \right) = \frac{n}{\gamma_m }  \log\left(\frac{2n}{\varepsilon} \right).\]
        Thus, the bound in \eqref{minimum time rapid mixing} holds with error parameter $\varepsilon/2$. Applying \cref{Dobrushin condition implies rapid mixing}, we obtain
        \[\|\Phi_{\delta_m}^m(\rho_0)-\sigma\|_{\mathrm{tr}} \leq \frac\varepsilon 2 \|\rho_0-\sigma\|_{\mathrm{tr}}.\]
        Finally, the exponential limit formula gives
        \begin{equation*}
            \Phi_{\delta_m}^m = \left( \cI + \frac{t}{m} \cL \right)^m \longrightarrow e^{t \cL}, \quad\text{as }  m \to \infty
        \end{equation*}
        Passing to the limit as $m\to\infty$ and using continuity of the trace norm yields
        \[\| e^{t \cL}(\rho_0)-\sigma\|_{\mathrm{tr}} \leq \frac\varepsilon 2 \|\rho_0-\sigma\|_{\mathrm{tr}} \leq \varepsilon,\]
        which concludes the proof.
    \end{proof}

    \section{Application of the main result}
    In this section, we illustrate the main result for square-root and mean-field normalizations of the interactions.

    Let $n\geq2$ and $J>0$. We assume that each site interacts with at most $d$ other sites, where $1\leq d\leq n-1$, and consider the Hamiltonian
    \begin{equation}\label{new Hamiltonian}
        H=\frac1{\sqrt d}\sum_{1\leq i<j\leq n}h_{i,j},
    \end{equation}
    where each $h_{i,j}$ is Hermitian and supported on $\{i,j\}$, at most $d$ terms incident to each site are nonzero, and $\|h_{i,j}\|\leq J$.  Thus, the interaction terms in our earlier notation are $H_{i,j}=h_{i,j}/\sqrt d$, and
    \[
        \Lambda
        =\max_{i\in[n]}\sum_{j\neq i}\frac{\|h_{i,j}\|}{\sqrt d}
        \leq J\sqrt d,
    \]
    where $h_{j,i}=h_{i,j}$ for $i<j$.
    With this bound on $\Lambda$, we can now reformulate \cref{satisfaction of Dobrushin condition} in terms of the number of interactions of the system. In particular,
    \begin{cor} \label{cor:degree-normalized-mixing}
        Let $H$ be a Hamiltonian as in \eqref{new Hamiltonian} and let
        $\sigma$ be its Gibbs state at inverse temperature $\beta$ such that
        \[
        \beta < \frac{1}{100 J \sqrt{d}}.
        \]
        Then, for every initial state \(\rho\) and every \(0<\varepsilon<1\), the evolution determined by \cref{Lindbladian} satisfies
        \[ \| e^{t \cL}(\rho) -\sigma \|_{\mathrm{tr}} \leq \varepsilon,\]
        for a time
        \[
        t = \mathcal O(\log(n/\varepsilon)).
        \]

    \end{cor}
    \begin{proof}
        Since
        \[
        \beta\Lambda \leq \beta J\sqrt{d} < 1/100,
        \]
         the claim follows from \cref{satisfaction of Dobrushin condition}.
    \end{proof}

    \begin{remark}
        We observe that when the Hamiltonian is fully connected, i.e. each site interacts with every other site ($d=n-1$) and $J =1$, then the temperature regime for which we have rapid-mixing for the Gibbs sampling, is given by
        \begin{equation*}
            \beta < \beta(n) := \frac{1}{100 (n-1)^{1/2}}.
        \end{equation*}
        In particular, this regime allows the norm of the exponent in the Gibbs state to grow linearly in $n$. Indeed,
        \begin{equation*}
            \beta(n) \| H \| \leq \sum_{1 \leq i < j \leq n} \frac{\| h_{i,j} \|}{100 (n-1)} \leq \frac{n}{200}.
        \end{equation*}
        For example, if we take
        \[H=\frac{1}{\sqrt{n-1}} \sum_{i<j}\sigma_Z^{(i)}\sigma_Z^{(j)}, \]
        we have
        \[
            \|H\|=\frac{n\sqrt{n-1}}2.
        \]
        Choosing
        \[
            \beta=\frac{c}{\sqrt{n-1}},\qquad 0<c<\frac1{100},
        \]
        satisfies the rapid-mixing condition and gives
        \[
            \beta\|H\|=\frac{cn}2.
        \]
    \end{remark}

    \begin{ex}[Mean-field normalization]
    For $n\geq2$, consider
    \[
        H=\frac{1}{n-1}\sum_{i<j}h_{i,j},
        \qquad \|h_{i,j}\|\leq J,
    \]
    where each $h_{i,j}$ is Hermitian and supported on $\{i,j\}$.
    Then $\Lambda\leq J$, so
    \cref{satisfaction of Dobrushin condition} gives rapid
    mixing whenever
    \[
    \beta <\frac{1}{100J},
    \]
     with a temperature threshold
    independent of the system size.
    \end{ex}

    \sectionnotoc{Use of AI-assisted tools}
    An AI-assisted language model was used during manuscript preparation
    for literature discovery, editorial restructuring, and consistency
    checks. Responsibility for the final text and for every mathematical
    claim remains with the authors.

    \sectionnotoc{Acknowledgements}
    NB and GDP have been supported by the UNA EUROPA SeedFunding project QUANTUMUnaE (CUP J37G25000380006).
    NB and GDP are members of the ``Gruppo Nazionale per la Fisica Matematica (GNFM)'' of the ``Istituto Nazionale di Alta Matematica ``Francesco Severi'' (INdAM)''.

    \bibliography{references} 

@inproceedings{BLMT25,
  author = {Bakshi, Ainesh and Liu, Allen and Moitra, Ankur and Tang, Ewin},
  title = {{A {Dobrushin} condition for quantum {Markov} chains: Rapid mixing and conditional mutual information at high temperature}},
  booktitle = {Proceedings of the 58th Annual ACM Symposium on Theory of Computing},
  publisher = {ACM},
  year = {2026},
  doi = {10.1145/3798129.3800859},
  eprint = {2510.08542},
  archivePrefix = {arXiv},
  note = {Full version: arXiv:2510.08542v1},
  url = {https://arxiv.org/abs/2510.08542v1},
}

@misc{BT26,
  author = {Bakshi, A. and Moitra, A. and Tan, X.},
  title = {{Algorithmic Aspects of the Fermi--Hubbard Model}},
  year = {2026},
  eprint = {2604.08408},
  archivePrefix = {arXiv},
  primaryClass = {quant-ph},
  note = {Version 2, revised 6 October 2026},
  url = {https://arxiv.org/abs/2604.08408v2},
}

@misc{Bergamaschi26,
  author = {Bergamaschi, T.},
  title = {{Fast mixing of all-to-all quantum systems at high temperatures}},
  year = {2026},
  eprint = {2606.26090},
  archivePrefix = {arXiv},
  primaryClass = {quant-ph},
  url = {https://arxiv.org/abs/2606.26090v1},
}

@inproceedings{BubleyDyer97,
  author = {Bubley, R. and Dyer, M.},
  title = {{Path coupling: A technique for proving rapid mixing in Markov chains}},
  booktitle = {Proceedings 38th Annual Symposium on Foundations of Computer Science},
  pages = {223--231},
  publisher = {IEEE},
  year = {1997},
  doi = {10.1109/SFCS.1997.646111},
}

@article{CGKR25,
  author = {Capel, {\'A}. and Gondolf, P. and Kochanowski, J. and Rouz{\'e}, C.},
  title = {{Quasi-optimal sampling from Gibbs states via non-commutative optimal transport metrics}},
  journal = {Annales Henri Poincar\'e},
  year = {2025},
  doi = {10.1007/s00023-025-01637-0},
}

@article{CL21,
  author = {Chen, C.-F. and Lucas, A.},
  title = {{Operator Growth Bounds from Graph Theory}},
  journal = {Communications in Mathematical Physics},
  volume = {385},
  number = {3},
  pages = {1273--1323},
  publisher = {Springer Science and Business Media LLC},
  year = {2021},
  month = {jul},
  doi = {10.1007/s00220-021-04151-6},
  issn = {1432-0916},
}

@misc{CKG23,
  author = {Chen, Chi-Fang and Kastoryano, Michael J. and Gily{\'e}n, Andr{\'a}s},
  title = {{An efficient and exact noncommutative quantum {Gibbs} sampler}},
  year = {2023},
  eprint = {2311.09207},
  archivePrefix = {arXiv},
  primaryClass = {quant-ph},
  note = {Version 2, revised 14 October 2025},
  url = {https://arxiv.org/abs/2311.09207v2},
}

@article{DTML21,
  author = {{De Palma}, G. and Marvian, M. and Trevisan, D. and Lloyd, S.},
  title = {{The Quantum Wasserstein Distance of Order 1}},
  journal = {IEEE Transactions on Information Theory},
  volume = {67},
  number = {10},
  pages = {6627--6643},
  publisher = {Institute of Electrical and Electronics Engineers (IEEE)},
  year = {2021},
  month = {Oct},
  doi = {10.1109/tit.2021.3076442},
  issn = {1557-9654},
}

@article{Dobrushin68,
  author = {Dobrushin, R. L.},
  title = {{The problem of uniqueness of a Gibbsian random field and the problem of phase transitions}},
  journal = {Functional Analysis and Its Applications},
  volume = {2},
  pages = {302--312},
  year = {1968},
  doi = {10.1007/BF01075682},
}

@article{Dobrushin70,
  author = {Dobrushin, R. L.},
  title = {{Prescribing a System of Random Variables by Conditional Distributions}},
  journal = {Theory of Probability and Its Applications},
  volume = {15},
  number = {3},
  pages = {458-486},
  year = {1970},
  doi = {10.1137/1115049},
}

@article{DGJ09,
  author = {Dyer, M. and Goldberg, L. A. and Jerrum, M.},
  title = {{Matrix norms and rapid mixing for spin systems}},
  journal = {The Annals of Applied Probability},
  volume = {19},
  number = {1},
  pages = {71--107},
  publisher = {Institute of Mathematical Statistics},
  year = {2009},
  month = {Feb},
  doi = {10.1214/08-aap532},
  issn = {2168-8737},
}

@misc{GKZ24,
  author = {Gamarnik, D. and Kiani, B. T. and Zlokapa, A.},
  title = {{Slow Mixing of Quantum Gibbs Samplers}},
  year = {2024},
  eprint = {2411.04300},
  archivePrefix = {arXiv},
  primaryClass = {quant-ph},
  note = {Version 2, revised 12 December 2024},
  url = {https://arxiv.org/abs/2411.04300v2},
}

@article{GKSL76,
  author = {Gorini, V. and Kossakowski, A. and Sudarshan, E. C. G.},
  title = {{Completely positive dynamical semigroups of \(N\)-level systems}},
  journal = {Journal of Mathematical Physics},
  volume = {17},
  pages = {821--825},
  year = {1976},
  doi = {10.1063/1.522979},
}

@article{KB16,
  author = {Kastoryano, M. J. and Brand{\~a}o, F. G. S. L.},
  title = {{Quantum Gibbs samplers: The commuting case}},
  journal = {Communications in Mathematical Physics},
  volume = {344},
  pages = {915--957},
  year = {2016},
  doi = {10.1007/s00220-016-2641-8},
}

@article{KACR25,
  author = {Kochanowski, J. and Alhambra, {\'A}. M. and Capel, {\'A}. and Rouz{\'e}, C.},
  title = {{Rapid Thermalization of Dissipative Many-Body Dynamics of Commuting Hamiltonians}},
  journal = {Communications in Mathematical Physics},
  volume = {406},
  number = {8},
  pages = {176},
  publisher = {Springer Science and Business Media LLC},
  year = {2025},
  month = {jul},
  doi = {10.1007/s00220-025-05353-y},
  issn = {1432-0916},
}

@article{LR72,
  author = {Lieb, E. H. and Robinson, D. W.},
  title = {{The finite group velocity of quantum spin systems}},
  journal = {Communications in Mathematical Physics},
  volume = {28},
  pages = {251--257},
  year = {1972},
  doi = {10.1007/BF01645779},
}

@article{Lindblad76,
  author = {Lindblad, G.},
  title = {{On the generators of quantum dynamical semigroups}},
  journal = {Communications in Mathematical Physics},
  volume = {48},
  pages = {119--130},
  year = {1976},
  doi = {10.1007/BF01608499},
}

@article{MZ95,
  author = {Majewski, A. W. and Zegarli{\'n}ski, B.},
  title = {{Quantum stochastic dynamics I: Spin systems on a lattice}},
  journal = {Mathematical Physics Electronic Journal},
  volume = {1},
  year = {1995},
  note = {Paper 2, 37 pp.},
  url = {https://web.ma.utexas.edu/mpej/Vol/1/2.pdf},
}

@article{MannHelmuth21,
  author = {Mann, R. L. and Helmuth, T.},
  title = {{Efficient algorithms for approximating quantum partition functions}},
  journal = {Journal of Mathematical Physics},
  volume = {62},
  number = {2},
  pages = {022201},
  publisher = {AIP Publishing},
  year = {2021},
  month = {Feb},
  doi = {10.1063/5.0013689},
  issn = {1089-7658},
}

@article{NR04,
  author = {Neto{\v c}n{\'y}, K. and Redig, F.},
  title = {{Large Deviations for Quantum Spin Systems}},
  journal = {Journal of Statistical Physics},
  volume = {117},
  number = {3-4},
  pages = {521--547},
  publisher = {Springer Science and Business Media LLC},
  year = {2004},
  month = {Nov},
  doi = {10.1007/s10955-004-3452-4},
  issn = {1572-9613},
}

@inproceedings{RS24,
  author = {Ramkumar, Akshar and Soleimanifar, Mehdi},
  title = {{Mixing Time of Quantum {Gibbs} Sampling for Random Sparse {Hamiltonians}}},
  booktitle = {20th Conference on the Theory of Quantum Computation, Communication and Cryptography (TQC 2025)},
  series = {Leibniz International Proceedings in Informatics (LIPIcs)},
  volume = {350},
  pages = {3:1--3:23},
  publisher = {Schloss Dagstuhl -- Leibniz-Zentrum f{\"u}r Informatik},
  year = {2025},
  doi = {10.4230/LIPIcs.TQC.2025.3},
  eprint = {2411.04454},
  archivePrefix = {arXiv},
}

@article{RW96,
  author = {Richter, S. and Werner, R. F.},
  title = {{Ergodicity of quantum cellular automata}},
  journal = {Journal of Statistical Physics},
  volume = {82},
  number = {3-4},
  pages = {963--998},
  publisher = {Springer Science and Business Media LLC},
  year = {1996},
  month = {Feb},
  doi = {10.1007/bf02179798},
  issn = {1572-9613},
}

@article{RFA25,
  author = {Rouz{\'e}, C. and {Stilck Fran{\c c}a}, D. and Alhambra, {\'A}. M.},
  title = {{Efficient thermalization and universal quantum computing with quantum Gibbs samplers}},
  journal = {Nature Physics},
  volume = {22},
  number = {7},
  pages = {1123--1127},
  publisher = {Springer Science and Business Media LLC},
  year = {2026},
  month = {Apr},
  doi = {10.1038/s41567-026-03246-y},
  issn = {1745-2481},
}

@article{RFA26,
  author = {Rouz{\'e}, C. and {Stilck Fran{\c c}a}, D. and Alhambra, {\'A}. M.},
  title = {{Optimal Quantum Algorithm for Gibbs State Preparation}},
  journal = {Physical Review Letters},
  volume = {136},
  number = {6},
  pages = {060601},
  publisher = {American Physical Society (APS)},
  year = {2026},
  month = {Feb},
  doi = {10.1103/lhht-svmn},
  issn = {1079-7114},
}

@misc{RSF26,
  author = {Rouz{\'e}, C. and {Stilck Fran{\c c}a}, D.},
  title = {{Rapid mixing of Gibbs samplers via quantum Dobrushin--Shlosman conditions}},
  year = {2026},
  eprint = {2610.03578},
  archivePrefix = {arXiv},
  primaryClass = {quant-ph},
  url = {https://arxiv.org/abs/2610.03578v1},
}

@book{Wat18,
  author = {Watrous, J.},
  title = {{The Theory of Quantum Information}},
  publisher = {Cambridge University Press},
  year = {2018},
  month = {apr},
  doi = {10.1017/9781316848142},
  isbn = {9781107180567},
}

@article{Weitz05,
  author = {Weitz, D.},
  title = {{Combinatorial Criteria for Uniqueness of Gibbs Measures}},
  journal = {Random Structures \& Algorithms},
  volume = {27},
  number = {4},
  pages = {445--475},
  year = {2005},
  doi = {10.1002/rsa.20073},
}

@misc{SMBB26,
  author = {{\v S}m{\'i}d, {\v S}. and Meister, R. and Berta, M. and Bondesan, R.},
  title = {{Rapid Mixing of Quantum Gibbs Samplers for Weakly-Interacting Quantum Systems}},
  year = {2025},
  eprint = {2510.04954},
  archivePrefix = {arXiv},
  primaryClass = {quant-ph},
  note = {Version 2, revised 19 April 2026},
  url = {https://arxiv.org/abs/2510.04954v2},
}

@article{DLL25,
  author = {Ding, Zhiyan and Li, Bowen and Lin, Lin},
  title = {{Efficient Quantum {Gibbs} Samplers with {Kubo--Martin--Schwinger} Detailed Balance Condition}},
  journal = {Communications in Mathematical Physics},
  volume = {406},
  number = {3},
  pages = {67},
  year = {2025},
  doi = {10.1007/s00220-025-05235-3},
  eprint = {2404.05998},
  archivePrefix = {arXiv},
}

@inproceedings{BLMT24,
  author = {Bakshi, Ainesh and Liu, Allen and Moitra, Ankur and Tang, Ewin},
  title = {{High-Temperature {Gibbs} States are Unentangled and Efficiently Preparable}},
  booktitle = {2024 IEEE 65th Annual Symposium on Foundations of Computer Science (FOCS)},
  pages = {1027--1036},
  publisher = {IEEE},
  year = {2024},
  doi = {10.1109/FOCS61266.2024.00068},
  eprint = {2403.16850},
  archivePrefix = {arXiv},
}

@inproceedings{BC26,
  author = {Bergamaschi, Thiago and Chen, Chi-Fang},
  title = {{Fast Mixing of Quantum Spin Chains at All Temperatures}},
  booktitle = {Proceedings of the 58th Annual ACM Symposium on Theory of Computing},
  publisher = {ACM},
  year = {2026},
  doi = {10.1145/3798129.3800798},
  eprint = {2510.08533},
  archivePrefix = {arXiv},
  note = {Full version: arXiv:2510.08533v2},
  url = {https://arxiv.org/abs/2510.08533v2},
}

@misc{Kim26,
  author = {Kim, Leeseok},
  title = {{Rapid mixing of quantum spin chains at any finite temperature}},
  year = {2026},
  eprint = {2610.01190},
  archivePrefix = {arXiv},
  primaryClass = {quant-ph},
  url = {https://arxiv.org/abs/2610.01190},
}

@misc{CD26,
  author = {Chen, Hongrui and Ding, Zhiyan},
  title = {{Spectral Analysis Beyond the Spectral Gap: Rapid Mixing of Quantum {Gibbs} Samplers in One Dimension}},
  year = {2026},
  eprint = {2610.05848},
  archivePrefix = {arXiv},
  primaryClass = {quant-ph},
  note = {Version 2, revised 21 October 2026},
  url = {https://arxiv.org/abs/2610.05848v2},
}

@article{CKBG25,
  author = {Chen, Chi-Fang and Kastoryano, Michael and Brand{\~a}o, Fernando G. S. L. and Gily{\'e}n, Andr{\'a}s},
  title = {{Efficient quantum thermal simulation}},
  journal = {Nature},
  volume = {646},
  number = {8085},
  pages = {561--566},
  year = {2025},
  doi = {10.1038/s41586-025-09583-x},
}

@misc{GCDK26,
  author = {Gily{\'e}n, Andr{\'a}s and Chen, Chi-Fang and Doriguello, Joao F. and Kastoryano, Michael J.},
  title = {{Quantum generalizations of {Glauber} and {Metropolis} dynamics}},
  year = {2024},
  eprint = {2405.20322},
  archivePrefix = {arXiv},
  primaryClass = {quant-ph},
  note = {Version 2, revised 23 January 2026},
  url = {https://arxiv.org/abs/2405.20322v2},
}

@article{DPR22,
  author = {{De Palma}, Giacomo and Rouz{\'e}, Cambyse},
  title = {{Quantum Concentration Inequalities}},
  journal = {Annales Henri Poincar{\'e}},
  volume = {23},
  number = {9},
  pages = {3391--3429},
  year = {2022},
  doi = {10.1007/s00023-022-01181-1},
  eprint = {2106.15819},
  archivePrefix = {arXiv},
  primaryClass = {quant-ph},
}

@article{CDS15,
  author = {Crawford, Nicholas and {De Roeck}, Wojciech and Sch{\"u}tz, Marius},
  title = {{Uniqueness regime for {Markov} dynamics on quantum lattice spin systems}},
  journal = {Journal of Physics A: Mathematical and Theoretical},
  volume = {48},
  number = {42},
  pages = {425203},
  year = {2015},
  doi = {10.1088/1751-8113/48/42/425203},
  eprint = {1506.00530},
  archivePrefix = {arXiv},
  primaryClass = {math-ph},
}

@article{KT13,
  author = {Kastoryano, Michael J. and Temme, Kristan},
  title = {{Quantum logarithmic {Sobolev} inequalities and rapid mixing}},
  journal = {Journal of Mathematical Physics},
  volume = {54},
  pages = {052202},
  year = {2013},
  doi = {10.1063/1.4804995},
  eprint = {1207.3261},
  archivePrefix = {arXiv},
  primaryClass = {quant-ph},
}

@misc{CRF21,
  author = {Capel, {\'A}ngela and Rouz{\'e}, Cambyse and {Stilck Fran{\c c}a}, Daniel},
  title = {{The modified logarithmic {Sobolev} inequality for quantum spin systems: classical and commuting nearest neighbour interactions}},
  year = {2020},
  eprint = {2009.11817},
  archivePrefix = {arXiv},
  primaryClass = {quant-ph},
  note = {Version 2, revised 3 June 2021},
  url = {https://arxiv.org/abs/2009.11817v2},
}

@article{BCGLPR24,
  author = {Bardet, Ivan and Capel, {\'A}ngela and Gao, Li and Lucia, Angelo and P{\'e}rez-Garc{\'i}a, David and Rouz{\'e}, Cambyse},
  title = {{Entropy Decay for {Davies} Semigroups of a One Dimensional Quantum Lattice}},
  journal = {Communications in Mathematical Physics},
  volume = {405},
  pages = {42},
  year = {2024},
  doi = {10.1007/s00220-023-04869-5},
  eprint = {2112.00601},
  archivePrefix = {arXiv},
  primaryClass = {quant-ph},
}

@article{FU10,
  author = {Fagnola, Franco and Umanit{\`a}, Veronica},
  title = {{Generators of {KMS} Symmetric {Markov} Semigroups on $\mathcal{B}(h)$ Symmetry and Quantum Detailed Balance}},
  journal = {Communications in Mathematical Physics},
  volume = {298},
  pages = {523--547},
  year = {2010},
  doi = {10.1007/s00220-010-1011-1},
  eprint = {0908.0967},
  archivePrefix = {arXiv},
  primaryClass = {quant-ph},
}

@misc{TK15,
  author = {Temme, Kristan and Kastoryano, Michael J.},
  title = {{How fast do stabilizer {Hamiltonians} thermalize?}},
  year = {2015},
  eprint = {1505.07811},
  archivePrefix = {arXiv},
  primaryClass = {quant-ph},
  url = {https://arxiv.org/abs/1505.07811},
}

@misc{LPP25,
  author = {Lucia, Angelo and P{\'e}rez-Garc{\'i}a, David and P{\'e}rez-Hern{\'a}ndez, Antonio},
  title = {{Spectral Gap Bounds for Quantum {Markov} Semigroups via Correlation Decay}},
  year = {2025},
  eprint = {2505.08991},
  archivePrefix = {arXiv},
  primaryClass = {quant-ph},
  url = {https://arxiv.org/abs/2505.08991},
}

@misc{CPS26,
  author = {Capel, {\'A}ngela and P{\'e}rez-Garc{\'i}a, David and Scandi, Matteo},
  title = {{Modified logarithmic Sobolev inequality for 1D non-commuting Hamiltonians}},
  year = {2026},
  eprint = {2610.03683},
  archivePrefix = {arXiv},
  primaryClass = {quant-ph},
  note = {Version 1, submitted 2 October 2026},
  url = {https://arxiv.org/abs/2610.03683v1},
}

@misc{GGW26,
  author = {Gao, Li and Guo, Jingyu and Wang, Sihao},
  title = {{Complete Modified Logarithmic Sobolev Inequalities for one-dimensional quantum Spin Chain}},
  year = {2026},
  eprint = {2610.28600},
  archivePrefix = {arXiv},
  primaryClass = {quant-ph},
  note = {Version 1, submitted 20 October 2026},
  url = {https://arxiv.org/abs/2610.28600v1},
}

@misc{CDL26,
  author = {Chen, Hongrui and Ding, Zhiyan and Lu, Jianfeng},
  title = {{Size-independent Logarithmic Sobolev Inequalities for Quantum Gibbs Samplers}},
  year = {2026},
  eprint = {2610.31571},
  archivePrefix = {arXiv},
  primaryClass = {quant-ph},
  note = {Version 1, submitted 22 October 2026},
  url = {https://arxiv.org/abs/2610.31571v1},
}

@misc{CZZ26,
  author = {Cai, Yiyi and Zhan, Yongtao and Zlokapa, Alexander},
  title = {{Fast mixing of the SYK model at high temperatures}},
  year = {2026},
  eprint = {2610.05953},
  archivePrefix = {arXiv},
  primaryClass = {quant-ph},
  note = {Version 1, submitted 5 October 2026},
  url = {https://arxiv.org/abs/2610.05953v1},
}

@article{AC21,
  author = {Amorim, {\'E}rik and Carlen, Eric A.},
  title = {{Complete positivity and self-adjointness}},
  journal = {Linear Algebra and its Applications},
  volume = {611},
  pages = {389--439},
  year = {2021},
  doi = {10.1016/j.laa.2020.10.038},
  eprint = {2009.05850},
  archivePrefix = {arXiv},
  url = {https://arxiv.org/abs/2009.05850v1},
}

@article{GR24,
  author = {Gao, Li and Rouz{\'e}, Cambyse},
  title = {{Coarse Ricci curvature of quantum channels}},
  journal = {Journal of Functional Analysis},
  volume = {286},
  number = {8},
  pages = {110336},
  year = {2024},
  doi = {10.1016/j.jfa.2024.110336},
  eprint = {2108.10609},
  archivePrefix = {arXiv},
  url = {https://arxiv.org/abs/2108.10609v1},
}

@article{BCLPR21,
  author = {Bardet, Ivan and Capel, {\'A}ngela and Lucia, Angelo and P{\'e}rez-Garc{\'i}a, David and Rouz{\'e}, Cambyse},
  title = {{On the modified logarithmic Sobolev inequality for the heat-bath dynamics for 1D systems}},
  journal = {Journal of Mathematical Physics},
  volume = {62},
  number = {6},
  pages = {061901},
  year = {2021},
  doi = {10.1063/1.5142186},
  eprint = {1908.09004},
  archivePrefix = {arXiv},
  url = {https://arxiv.org/abs/1908.09004v2},
}

@incollection{Martinelli99,
  author = {Martinelli, Fabio},
  title = {{Lectures on {Glauber} dynamics for discrete spin models}},
  booktitle = {Lectures on Probability Theory and Statistics},
  editor = {Bernard, Pierre},
  series = {Lecture Notes in Mathematics},
  volume = {1717},
  pages = {93--191},
  publisher = {Springer},
  address = {Berlin, Heidelberg},
  year = {1999},
  doi = {10.1007/978-3-540-48115-7_2},
}

@article{Davies74,
  author = {Davies, E. B.},
  title = {{{Markovian} master equations}},
  journal = {Communications in Mathematical Physics},
  volume = {39},
  pages = {91--110},
  year = {1974},
  doi = {10.1007/BF01608389},
}

@article{DP25,
  author = {De Palma, Giacomo and Pastorello, Davide},
  title = {{Quantum Concentration Inequalities and Equivalence of the Thermodynamical Ensembles: An Optimal Mass Transport Approach}},
  journal = {Journal of Statistical Physics},
  volume = {192},
  number = {6},
  pages = {87},
  year = {2025},
  doi = {10.1007/s10955-025-03464-3},
  eprint = {2403.18617},
  archivePrefix = {arXiv},
  primaryClass = {math-ph},
  url = {https://arxiv.org/abs/2403.18617v2},
}
    \bibliographystyle{abbrvurl}

\end{document}